\documentclass[a4paper,fontsize=11,numbers=endperiod,abstracton]{scrartcl}

\usepackage[american]{babel}
\usepackage[T1]{fontenc}
\usepackage{lmodern}
\usepackage{anyfontsize}
\usepackage{xstring}
\usepackage{amsmath}
\usepackage{amsfonts}
\usepackage{amssymb}
\usepackage{mathtools}
\usepackage[thmmarks,amsmath]{ntheorem}
\usepackage{dsfont}
\usepackage{relsize}
\usepackage{scalerel}
\usepackage{tensor}
\usepackage{stackrel}

\addtokomafont{sectionentry}{\mdseries}
\RedeclareSectionCommand[tocbeforeskip=0.5ex plus 0.5pt minus 0.5pt]{section}

\usepackage[a4paper]{geometry}
\usepackage[hang]{footmisc}
\usepackage{tikz}
\usepackage{tkz-euclide}
\usepackage{tikz-cd}
\usepackage{pgfplots}
\pgfplotsset{compat=1.13}

\usepackage{silence}
\usepackage{braids}

\usetkzobj{all}
\usetikzlibrary{arrows,calc}
\usetikzlibrary{shapes}
\usetikzlibrary{positioning}
\usetikzlibrary{decorations.pathreplacing,decorations.markings,intersections}

\allowdisplaybreaks

\theoremstyle{break}
\newtheorem{theorem}{Theorem}[section]
\newtheorem{proposition}[theorem]{Proposition}
\newtheorem{corollary}[theorem]{Corollary}
\newtheorem{lemma}[theorem]{Lemma}
\newtheorem{problem}[theorem]{Problem}

\theorembodyfont{\normalfont}
\newtheorem{definition}[theorem]{Definition}

\theoremstyle{nonumberbreak}
\theoremheaderfont{\bfseries}
\theorembodyfont{\normalfont}
\theoremseparator{.}
\theoremsymbol{\linebreak[1]\hspace*{.5em plus 1fill}\ensuremath{\square}}
\newtheorem{proof}{Proof}

\makeatletter
\newcommand{\proofpart}[2]{%
  \par
  \addvspace{\medskipamount}%
  \noindent\emph{Part #1: #2}\par\nobreak
  \addvspace{\smallskipamount}%
  \@afterheading
}
\makeatother

\numberwithin{equation}{section}

\DeclareMathOperator{\ev}{ev}
\DeclareMathOperator{\coev}{coev}
\DeclareMathOperator{\id}{id}
\DeclareMathOperator{\op}{op}
\DeclareMathOperator{\rev}{rev}
\DeclareMathOperator{\bal}{bal}

\DeclareMathOperator*{\bigbox}{\raisebox{-0.55ex}{\ensuremath{\mathlarger{\mathlarger{\mathlarger{\boxtimes}}}}}}

\newcommand{\Hom}{\mathsf{Hom}}
\newcommand{\End}{\mathsf{End}}
\newcommand{\Fun}{\mathsf{Fun}}
\newcommand{\Vectk}{\mathsf{Vect}_{\mathds{k}}}
\newcommand{\Bimod}{\mathsf{Bimod}}
\mathchardef\mh="2D 
\newcommand{\Mod}[2][0]{%
    \IfStrEqCase{#1}{%
        {r}{\mathsf{Mod}\mh #2}%
        {l}{#2 \mh\mathsf{Mod}}%
        {0}{\mathsf{Mod}}
    }[%
    ]%
}
\newcommand{\AlgBiMod}[2]{#1 \mh\mathsf{Bimod}\mh #2}

\def\dunderline#1{\underline{\underline{#1}}}

\newcommand*{\defeq}{\mathrel{\vcenter{\baselineskip0.5ex \lineskiplimit0pt
                     \hbox{\scriptsize.}\hbox{\scriptsize.}}}%
                     =}
\newcommand*{\defeqinv}{=\mathrel{\vcenter{\baselineskip0.5ex \lineskiplimit0pt
                     \hbox{\scriptsize.}\hbox{\scriptsize.}}}%
                      }

\newcommand{\equaltext}[1]{\ensuremath{\stackrel{#1}{=}}}
\newcommand{\equaltextab}[2]{\ensuremath{\stackrel[#2]{#1}{=}}}

\makeatletter
\newcommand*{\Bigggg}[1]{{\hbox{$\left#1\vbox to 25.5\p@{}\right.\n@space$}}}
\makeatother

\usepackage{hyperref}

\title{{\LARGE Permutation actions on modular tensor categories of topological multilayer phases}}
\author{{\normalsize Albert Georg Passegger\thanks{E-mail: albert.georg.passegger@univie.ac.at}}\\
{\normalsize Faculty of Mathematics, University of Vienna, Austria}}
\date{\vspace{-5ex}}

\begin{document}

\tikzset{
        every edge/.style={draw,postaction={decorate,decoration={markings,mark=at position 0.5 with {\arrow{stealth}}}}}
        ,
        ->-/.style={decoration={
          markings,
          mark=at position #1 with {\arrow{stealth}}},postaction={decorate}},
        -<-/.style={decoration={
                    markings,
                    mark=at position #1 with {\arrowreversed{stealth}}},postaction={decorate}},
        circle node/.style={circle,draw,inner sep=0pt,minimum size=0.4cm},
        }

\maketitle

\vspace{0.5cm}

\begin{abstract}
We find a non-trivial representation of the symmetric group $S_n$ on the $n$-fold Deligne product $\mathcal{C}^{\boxtimes n}$ of a modular tensor category $\mathcal{C}$ for any $n \geq 2$. This is accomplished by checking that a particular family of $\mathcal{C}^{\boxtimes n}$-bimodule categories representing adjacent transpositions satisfies the symmetric group relations with respect to the relative Deligne product. The bimodule categories are based on a permutation action of $S_2$ on $\mathcal{C}\boxtimes\mathcal{C}$ discussed by Fuchs and Schweigert in \cite{fuchs-schweigert}, for which we show that it is, in a certain sense, unique. In the context of condensed matter physics, the $S_n$-representation corresponds to the specification of permutation twist surface defects in a (2+1)-dimensional topological multilayer phase, which are relevant to topological quantum computation and could promote the explicit construction of the data of an $S_n$-gauged phase.\bigskip

\noindent\textit{Mathematics Subject Classification (MSC2010):} 16D90, 18D10, 18D35, 81T45 \medskip

\end{abstract}

\vspace{0.5cm}

\tableofcontents


\newpage

\section{Introduction}

Bimodule categories over braided monoidal categories are the proper categorified analogs of bimodules over commutative rings. In a braided monoidal category objects can be multiplied using a ``tensor product'', and two objects in a tensor product can be exchanged using a natural isomorphism called ``braiding'', which makes the tensor product commutative up to isomorphism. For many applications, especially in mathematical physics, one considers modular tensor categories \cite{egno,kassel,mueger-mtc,turaev-quantum-invariants}, which are braided monoidal categories with additional structure, where the braiding satisfies a non-degeneracy condition. Bimodule categories over modular tensor categories encode the necessary data to describe boundary conditions and domain walls in topological and rational conformal field theory \cite{fuchs-runkel-schweigert,fs-symmetries-defects,kitaev-kong}. An interesting class of bimodule categories in this context can be obtained from group actions, which are interpreted as symmetries in the field theory setting \cite{BBCW,ffrs-defect-lines}.\medskip

If one considers the Deligne product \cite{deligne} of multiple copies of a modular tensor category, there is a special type of group action given by permutation. For any modular tensor category $\mathcal{C}$ and $n\in\mathds{N}$, $n \geq 2$, there is an action of the symmetric group $S_n$ on the $n$-fold Deligne product $\mathcal{C}^{\boxtimes n}$ (with the same modular tensor category structure on each copy of $\mathcal{C}$) obtained by permuting the factors. This defines a representation of $S_n$ as braided automorphisms of $\mathcal{C}^{\boxtimes n}$. In general, braided automorphisms are in bijection with invertible $\mathcal{C}^{\boxtimes n}$-bimodule categories \cite[Theorem 5.2]{eno09}. It remains to describe them explicitly.\medskip

The case $n=2$ was elaborated by Fuchs and Schweigert \cite{fuchs-schweigert}. For each element of the symmetric group $S_2$ there is a corresponding $\mathcal{C}\boxtimes\mathcal{C}$-bimodule category: $\mathcal{C}\boxtimes\mathcal{C}$ as trivial $\mathcal{C}\boxtimes\mathcal{C}$-bimodule category, and a $\mathcal{C}\boxtimes\mathcal{C}$-bimodule category $\mathcal{P}$, corresponding to the single non-trivial cycle in $S_2$, with underlying category $\mathcal{C}$.\medskip

Our main theorem (Theorem \ref{thm:rep-Sn}) extends this result. For every $n\in\mathds{N}$, $n \geq 2$, we construct a non-trivial representation of $S_n$ on $\mathcal{C}^{\boxtimes n}$ in the sense of Definition \ref{def:representation-on-cat} below. We find an invertible $\mathcal{C}^{\boxtimes n}$-bimodule category with underlying category $\mathcal{C}^{\boxtimes (n-1)}$ for every adjacent transposition of the symmetric group $S_n$. Since $S_n$ is generated by the set of adjacent transpositions satisfying certain relations (the Coxeter presentation of the symmetric group), the specification of bimodule categories representing the adjacent transpositions is sufficient for the definition of an $S_n$-representation. Starting from the two-dimensional case of \cite{fuchs-schweigert} there is an evident choice for this family of bimodule categories, for which we check that the symmetric group relations with respect to the relative Deligne product over $\mathcal{C}^{\boxtimes n}$ (cf. Subsection \ref{sec:appendix-relative-deligne}) are satisfied. For $n=2$ we recover the representation of \cite{fuchs-schweigert}. \medskip

The mathematical analysis of such permutation actions is motivated by the theory of topological phases of matter \cite{levin-wen,anyon-tqc,categorical-anyons}. A system of many particles confined to a two-dimensional plane can have local point-like excitations of its ground state, called ``quasiparticles''. If there is an energy gap between the ground state and excited states (so long-range interactions between the particles are prevented), then these quasiparticles can have ``anyonic'' statistics \cite{leinaas-myrheim,wilczek}: their state acquires an arbitrary phase factor or even a non-trivial unitary transformation under the interchange of two quasiparticles. Particles showing this symmetry behavior are called abelian or non-abelian ``anyons'', respectively. They have the property that the transformation of the state only depends on the topology of the path along which anyons are interchanged, i.e. paths which can be deformed continuously into each other lead to the same transformation. Consequently, the low-energy long-distance effective field theory of the system can be described by a topological quantum field theory. Interacting many-body systems with such characteristics are called ``topological phases.'' 

In a topological phase which supports non-abelian anyons, there is a finite number of lowest energy states. The degeneracy between the different ground states is topologically protected; only adiabatic motions of the quasiparticles (which amounts to the braiding of worldlines if the exchange process is pictured in (2+1)-dimensional spacetime with time direction drawn vertically) lead to non-trivial unitary operations on the degenerate Hilbert space of states (which amounts to acting with operators from a representation of the braid group \cite{artin-zopf}). 

The statistical properties of topological phases with anyons can be extended by adding ``twist defects'' \cite{BaJQ1,teo}. These are extrinsic point-like objects, which may be created and controlled by external fields and have long-ranged confining interactions. The state of an anyon can be transformed when the anyon encircles such a twist defect. A more general class of defects is given by ``surface defects'', which are domain walls between different topological phases \cite{BaJQ3}.\medskip

Tensor category theory is the natural language to describe such systems mathematically \cite{levin-wen,categorical-anyons}. The kinematic aspects of a (2+1)-dimensional topological phase of anyons can be expressed in terms of algebraic rules describing trajectories, worldline braiding and fusion rules of anyons, which are encoded in the mathematical data of a modular tensor category. The three-dimensional topological quantum field theory describing the topological phase is then obtained from the modular tensor category by the Reshetikhin-Turaev construction \cite{bak-kir,rt-original,turaev-quantum-invariants}. Furthermore, twist defects are represented by bimodule categories. Let $\mathcal{C}$ be a modular tensor category.

\begin{itemize}
\item A $\mathcal{C}$-$\mathcal{C}$-bimodule category (called a $\mathcal{C}$-bimodule category for short) represents a surface defect in the topological phase $\mathcal{C}$, i.e. a domain wall separating the phase $\mathcal{C}$ from itself \cite{fsv1,kitaev-kong}.
\item The ``fusion product'' of surface defects is given by the relative Deligne product $\boxtimes_\mathcal{C}$. If $\mathcal{M}_1$ and $\mathcal{M}_2$ are $\mathcal{C}$-bimodule categories, then also $\mathcal{M}_1 \boxtimes_\mathcal{C} \mathcal{M}_2$ is a $\mathcal{C}$-bimodule category. This is interpreted as the fusion of the two surface defects into a new surface defect labeled by $\mathcal{M}_1 \boxtimes_\mathcal{C} \mathcal{M}_2$.
\item As a $\mathcal{C}$-bimodule category, the category $\Fun_\mathcal{C} (\mathcal{M}_1,\mathcal{M}_2)$ of $\mathcal{C}$-bimodule functors from $\mathcal{M}_1$ to $\mathcal{M}_2$ is interpreted as the category of line defects in the plane (``surface Wilson lines'') separating a surface defect with label $\mathcal{M}_1$ from a surface defect with label $\mathcal{M}_2$ \cite{fuchs-schweigert,fsv1}.
\item $\mathcal{C}^{\boxtimes n}$ corresponds to an $n$-layer system, i.e. $n$ copies of the topological phase $\mathcal{C}$ stacked on top of each other \cite{BaJQ1} (see also \cite{BaQ,BaWen}).
\item The topological phase $\mathcal{C}^{\boxtimes n}$ has an internal symmetry by permuting the $n$ layers. ``Permutation twist surface defects'' are $\mathcal{C}^{\boxtimes n}$-bimodule categories representing such permutations. In particular, there is the transparent surface defect $\mathcal{C}^{\boxtimes n}$, as trivial $\mathcal{C}^{\boxtimes n}$-bimodule category, representing the identity permutation.
\end{itemize}

Therefore, finding a representation of $S_n$ on $\mathcal{C}^{\boxtimes n}$ in the sense of Definition \ref{def:representation-on-cat} is related to identifying surface defects which describe permutations of layers in the topological $n$-layer phase $\mathcal{C}^{\boxtimes n}$.\medskip

The study of permutation twist defects is important for topological quantum computation, an approach to quantum computation which is based on anyonic statistics \cite{kitaev,anyon-tqc,pachos,rowell16}. The properties of a system of anyons are robust under local perturbations like, for example, interactions with the surrounding, and thus the system is protected from decoherence and error. This topological fault-tolerance is a key property of topological quantum computation, where quantum information is stored in the anyon state spaces employing the ground state degeneracy, and adiabatic motions of anyons play the role of logical gates. Adding twist defects allows to enrich the set of logical gates. Similar to anyons, twist defects can also be exchanged, and they can be thought to be labeled by a symmetry, which acts on the state of an anyon when the anyon is braided around the defect \cite{BaJQ1}. For a practical quantum computer one requires to have a ``universal'' set of gates to which any unitary transformation on the state space (which are the operations on the quantum computer) can be reduced. Remarkably, introducing twist defects to a topological phase can enable universal topological quantum computation, even when the non-abelian state of the topological phase itself is non-universal, as is shown to be the case for a bilayer Ising system in \cite[Section V.]{BaJQ1}. The results of this paper might allow a better mathematical understanding of topological multilayer phases with permutation twist surface defects as systems for quantum computers which are potentially universal. \medskip

We mention that the physical interpretation suggests to view all bimodule categories as 1-morphisms in the tricategory $\Bimod$, which has fusion categories as objects and bimodule categories, bimodule functors and bimodule natural transformations as 1-, 2- and 3-morphisms, respectively \cite{schaumann} (see Theorem \ref{thm:tricat-bimod} for a summary). In particular, this viewpoint is related to the problem of ``gauging'' the $S_n$-symmetry of the topological phase $\mathcal{C}^{\boxtimes n}$. We will comment on this issue in the final Section \ref{sec:outlook}.\medskip

The outline of the paper is as follows. In Section \ref{sec:prelim} some important notions, conventions and results are collected, and the formalism of graphical calculus is briefly introduced. We then verify in Section \ref{sec:modules-deligne-two} that the $\mathcal{C}\boxtimes\mathcal{C}$-bimodule category $\mathcal{P}$ considered in \cite{bfrs,fuchs-schweigert} is, up to equivalence, the unique one among all $\mathcal{C}\boxtimes\mathcal{C}$-bimodule structures on $\mathcal{C}$ whose actions just arrange three objects under the tensor product of $\mathcal{C}$ and whose module associativity isomorphisms contain a minimal number of braidings. Then the problem of constructing a representation of $S_n$ on $\mathcal{C}^{\boxtimes n}$ is formulated in Section \ref{sec:construction}, and we state the main theorem in Theorem \ref{thm:rep-Sn}. The proof is elaborated in detail, with some background on dual bases of Hom spaces presented in Appendix \ref{appendix:algebras-graphical-notation}. A few fundamental definitions and results from the theory of module categories are summarized in Appendix \ref{appendix:module-categories}.

\section{Preliminaries}
\label{sec:prelim}

To fix conventions we recall some notions from tensor category theory. Basic knowledge of ordinary category theory \cite{maclane} is assumed.

Let $\mathds{k}$ be an algebraically closed field of zero characteristic. A category is called \textit{$\mathds{k}$-linear} if all morphism spaces are $\mathds{k}$-vector spaces and the composition is $\mathds{k}$-linear. A $\mathds{k}$-linear category is called \textit{finitely semisimple} if there are finitely many isomorphism classes of simple objects, i.e. objects $X$ such that $\End(X) \cong \mathds{k}$, and every object can be written as a direct sum of finitely many simple objects. 

For two (locally finite and abelian, cf. \cite[Chapter 1]{egno}) $\mathds{k}$-linear categories $\mathcal{C}$ and $\mathcal{D}$, the \textit{Deligne product} $\mathcal{C} \boxtimes \mathcal{D}$ is the $\mathds{k}$-linear category defined by a universal property for bifunctors from $\mathcal{C} \times \mathcal{D}$, that are $\mathds{k}$-linear and right exact in both variables (see \cite{deligne} and \cite[Section 1.11]{egno} for details). One writes $X \boxtimes Y$ for the objects of $\mathcal{C} \boxtimes \mathcal{D}$. If $\mathcal{C}$ and $\mathcal{D}$ are finitely semisimple, then also $\mathcal{C}\boxtimes\mathcal{D}$ is finitely semisimple, and if $\{X_i\}_i$ and $\{Y_j\}_j$ are sets of representatives of the isomorphism classes of simple objects in $\mathcal{C}$ and $\mathcal{D}$, respectively, then $\{X_i \boxtimes Y_j\}_{i,j}$ is a set of representatives of the isomorphism classes of simple objects in $\mathcal{C}\boxtimes\mathcal{D}$ \cite[Proposition 1.11.2]{egno}.

A \textit{monoidal category} is a tuple consisting of a category $\mathcal{C}$, a bifunctor $\otimes\colon \mathcal{C}\times\mathcal{C}\to\mathcal{C}$, called the tensor product, a unit object $\mathds{1}_{\mathcal{C}} \in\mathcal{C}$, and natural isomorphisms, which describe the weak associativity and weak unitality of the tensor product and satisfy certain coherence axioms \cite[Definition 2.2.8]{egno}. By Mac Lane's strictness theorem \cite{maclane}, every monoidal category is monoidally equivalent to a so-called ``strict'' monoidal category, for which the tensor product is strictly associative and strictly unital with respect to $\mathds{1}_{\mathcal{C}}$. As long as there is no danger of confusion we will write $\mathds{1}$ instead of $\mathds{1}_{\mathcal{C}}$ from now on.

A \textit{fusion category} \cite[Chapter 4]{egno} is a $\mathds{k}$-linear, finitely semisimple and rigid monoidal category with simple monoidal unit. ``Rigid'' means that right and left duals exist: the right dual $X^\ast$ of an object $X$ is an object together with evaluation and coevaluation morphisms $\ev_X \colon X \otimes X^\ast \to \mathds{1}$ and $\coev_X \colon \mathds{1}\to X^\ast \otimes X$ satisfying the ``snake identites''(a term which becomes clear when they translated into string diagrams as introduced below):
\begin{gather}
(\ev_X \otimes \id_{X}) \circ (\id_{X} \otimes \coev_X) = \id_{X} \label{eq:snake-identities-right} \\
(\id_{X^\ast} \otimes \ev_X) \circ (\coev_X \otimes \id_{X^\ast}) = \id_{X^\ast} \nonumber
\end{gather}
Similarly, the left dual of $X$ is an object $ \tensor[^{\ast}]{X}{}$ together with morphisms $\widetilde{\ev}_X \colon \tensor[^{\ast}]{X}{} \otimes X \to \mathds{1}$ and $\widetilde{\coev}_X \colon \mathds{1}\to X \otimes \tensor[^{\ast}]{X}{}$, such that the corresponding snake identites hold:
\begin{gather}
(\id_{X} \otimes \widetilde{\ev}_X) \circ (\widetilde{\coev}_X \otimes \id_{X}) = \id_{X} \label{eq:snake-identities-left} \\
(\widetilde{\ev}_X \otimes \id_{\tensor[^{\ast}]{X}{}}) \circ (\id_{\tensor[^{\ast}]{X}{}} \otimes \widetilde{\coev}_X) = \id_{\tensor[^{\ast}]{X}{}} \nonumber
\end{gather}
To a morphism $f \in \Hom(X,Y)$ one defines the right dual morphism $f^\ast \in \Hom(Y^\ast, X^\ast)$ by
\begin{gather}
f^\ast = (\id_{X^\ast} \otimes \ev_Y) \circ (\id_{X^\ast} \otimes f \otimes \id_{Y^\ast}) \circ (\coev_X \otimes \id_{Y^\ast}) \, \, ,
\label{eq:cat-dual-morphism-right}
\end{gather}
and analogously the left dual $\tensor[^{\ast}]{f}{} \in \Hom(\tensor[^{\ast}]{Y}{},\tensor[^{\ast}]{X}{})$ by
\begin{gather}
\tensor[^{\ast}]{f}{} = (\widetilde{\ev}_Y \otimes \id_{\tensor[^{\ast}]{X}{}}) \circ (\id_{\tensor[^{\ast}]{Y}{}} \otimes f \otimes \id_{\tensor[^{\ast}]{X}{}}) \circ (\id_{\tensor[^{\ast}]{Y}{}} \otimes \widetilde{\coev}_X) \, \, .
\label{eq:cat-dual-morphism-left}
\end{gather}
For better distinction from other duality concepts considered in this paper, dual morphisms will be called \textit{``categorical duals''}.

A \textit{braiding} on a monoidal category $\mathcal{C}$ is a natural isomorphism with components $c_{X,Y} : X \otimes Y \to Y \otimes X$ for objects $X,Y$ such that a pair of hexagon diagrams, which express compatibility with the tensor product, commutes \cite[XIII.1.1]{kassel}. A braided monoidal category is a monoidal category together with a braiding. A \textit{twist} on a braided monoidal category is a natural isomorphism with components $\theta_X \colon X \to X$ such that $\theta_{X\otimes Y} = c_{Y,X} \circ c_{X,Y} \circ (\theta_X \otimes \theta_Y)$.

A \textit{ribbon category} is a braided rigid monoidal category with selfdual twists, i.e. $(\theta_X)^\ast = \theta_{X^\ast}$ for all objects $X$. It suffices to assume e.g. right duality, because the left duality can be constructed from right duality, braiding and twist, in a way that left and right duals are equal for objects and morphisms \cite[Proposition XIV.3.5]{kassel}. This implies that every ribbon category is \textit{pivotal}: there is a monoidal natural isomorphism (the ``pivotal structure'') with components $X\to(X^\ast)^\ast$ \cite[Definition 4.7.7]{egno}. As left and right duals are identified, we choose to denote them as right duals.

Finally, a \textit{modular tensor category} is a ribbon fusion category with invertible ``S-matrix'' \cite[Chapter 3]{bak-kir} (see also \cite{mueger-mtc} for a general overview). The latter is a non-degeneracy condition for the braiding, which can equivalently be described by the property that every simple object $U$, for which $c_{X,U} \circ c_{U,X} = \id_{U\otimes X}$ holds for all objects $X$, is isomorphic to the monoidal unit \cite{bruguieres}. The Deligne product of modular tensor categories is again a modular tensor category (in particular a fusion category, by \cite[Corollary 4.6.2]{egno}). \bigskip

The goal is to construct, for any integer $n\geq 2$, a non-trivial representation of the symmetric group $S_n$ on the $n$-fold Deligne product $\mathcal{C}^{\boxtimes n}$ of a modular tensor category $\mathcal{C}$, according to the following definition.

\begin{definition}[Representation, action, global symmetry, Picard group]
\label{def:representation-on-cat}
A \textit{representation} of a group $G$ on a braided fusion category $\mathcal{B}$ is a group homomorphism
\begin{gather*}
[\varrho]\colon G \to \mathsf{Pic}(\mathcal{B}) \, \, ,
\end{gather*}
where the \textit{Picard group} $\mathsf{Pic}(\mathcal{B})$ of $\mathcal{B}$ \cite{eno09} is the group of equivalence classes of invertible $\mathcal{B}$-bimodule categories (Definition \ref{def:bimodule-cat}) for which left and right action are identical and left and right module structure are related via the braiding (Proposition \ref{prop:from-module-to-bimodule-category}), with group multiplication induced by the relative Deligne product $\boxtimes_{\mathcal{B}}$ (Definition \ref{def:rel-Deligne-product}). This is also called an \textit{action} of $G$ on $\mathcal{B}$, and the pair $(G,[\varrho])$ is called a \textit{global symmetry} of $\mathcal{B}$ \cite{BBCW}.
\end{definition}

A $\mathcal{B}$-bimodule category $\mathcal{M}$ is \textit{invertible} if there exists a $\mathcal{B}$-bimodule category $\overline{\mathcal{M}}$ with a bimodule equivalence $\mathcal{M} \boxtimes_{\mathcal{B}} \overline{\mathcal{M}} \simeq \mathcal{B}$ \cite[Definition 4.1]{eno09}. If $\mathcal{B}$ is a modular tensor category representing a topological phase, $\overline{\mathcal{M}}$ has the physical interpretation of an ``anti'' surface defect of $\mathcal{M}$, which annihilates when fused with $\mathcal{M}$. We refer to Appendix \ref{appendix:module-categories} for an overview over the most important notions from the theory of module categories. \medskip

More conceptually, the Picard group is obtained from a sub-tricategory \cite[Definition 4.1]{gurski} of the tricategory $\Bimod$ of fusion categories, bimodule categories, bimodule functors and bimodule natural transformations \cite[Theorem 3.6.1]{schaumann}:

\begin{definition}
\label{def:picard}
The sub-tricategory $\dunderline{\mathsf{Pic}} \subset \Bimod$, called the \textit{Picard 3-groupoid}, has
\begin{itemize}
\item objects: braided fusion categories,
\item 1-morphisms: invertible bimodule categories for which left and right module structure are related via the braiding (Proposition \ref{prop:from-module-to-bimodule-category}),
\item 2-morphisms: equivalences of bimodule categories,
\item 3-morphisms: isomorphisms of bimodule equivalences.
\end{itemize}
The \textit{Picard 2-groupoid} $\underline{\mathsf{Pic}}$ has
\begin{itemize}
\item objects: braided fusion categories,
\item 1-morphisms: invertible bimodule categories for which left and right module structure are related via the braiding (Proposition \ref{prop:from-module-to-bimodule-category}),
\item 2-morphisms: isomorphism classes of bimodule equivalences.
\end{itemize}
The \textit{Picard groupoid} $\mathsf{Pic}$ has
\begin{itemize}
\item objects: braided fusion categories,
\item morphisms: equivalence classes of invertible bimodule categories for which left and right module structure are related via the braiding (Proposition \ref{prop:from-module-to-bimodule-category})
\end{itemize}
(i.e. it is obtained from $\dunderline{\mathsf{Pic}}$ by forgetting 2- and 3-morphisms and identifying isomorphic 1-morphisms). 
\end{definition}
Note that for a braided fusion category $\mathcal{B}$, $\mathsf{Pic}(\mathcal{B}) \equiv \mathsf{Pic}(\mathcal{B},\mathcal{B})$ is the above defined Picard group of $\mathcal{B}$. It makes sense to call a group homomorphism $G \to \mathsf{Pic}(\mathcal{B})$ a representation of $G$ on a braided fusion category $\mathcal{B}$: there is a monoidal functor $\underline{\mathsf{Pic}}(\mathcal{B}) \to \underline{\mathsf{EqBr}}(\mathcal{B})$ obtained from the so-called ``alpha-induction functors'' \cite[Section 5.4]{eno09}. Here, $\underline{\mathsf{EqBr}}(\mathcal{B})$ is the groupoid of braided autoequivalences of $\mathcal{B}$ and their isomorphisms, which comes from a 2-groupoid $\underline{\mathsf{EqBr}}$ of braided fusion categories, braided equivalences and their isomorphisms, and is the proper categorical analog of the automorphism space of a commutative ring. If $\mathcal{B}$ is a modular tensor category, this monoidal functor can be promoted to an equivalence $\underline{\mathsf{Pic}}(\mathcal{B}) \simeq \underline{\mathsf{EqBr}}(\mathcal{B})$ of groupoids by \cite[Theorem 5.2]{eno09}, which implies the equivalence between invertible topological defects and symmetry of the topological phase described by $\mathcal{B}$.

In this paper we will specify representations on modular tensor categories in terms of left module categories, which are then interpreted as 1-morphisms in $\underline{\mathsf{Pic}}$ by Proposition \ref{prop:from-module-to-bimodule-category}. Equivalences of these bimodule categories (Lemma \ref{lem:bimodule-eq}) are fixed by stating them as equivalences of left module categories (Definition \ref{def:appendix-equivalence-modulecats}, Proposition \ref{prop:from-module-functor-to-bimodule-functor}). Representations in terms of braided autoequivalences can be obtained from our bimodule categories by employing the equivalence of \cite[Theorem 5.2]{eno09}.\medskip

Our starting point is a representation of $S_2$ on $\mathcal{C}\boxtimes\mathcal{C}$ by Fuchs and Schweigert:
\begin{theorem}[{\cite{fuchs-schweigert}}]
\label{thm:fs-thm}
Let $\mathcal{C}$ be a modular tensor category, which is assumed to be monoidally strict without loss of generality by Mac Lane's strictness theorem, and let $\mathcal{P}=(\mathcal{C},\triangleright,\psi)$ be the $\mathcal{C}\boxtimes\mathcal{C}$-bimodule category with left action functor
\begin{gather*}
\triangleright\colon (\mathcal{C}\boxtimes\mathcal{C}) \times \mathcal{C} \to \mathcal{C} \, \, ,\\
(X\boxtimes Y)\triangleright C \defeq  X \otimes Y \otimes C \, \, ,
\end{gather*}
left module associativity isomorphism
\begin{gather}
\psi_{D,D',C} \colon (D\otimes D') \triangleright C \to D \triangleright (D' \triangleright C) \, \, ,\nonumber\\
\psi_{D,D',C} = \id_X \otimes c_{X',Y} \otimes \id_{Y'\otimes C}
\label{eq:module-assoc-deligne2}
\end{gather}
for $D=X\boxtimes Y$ and $D'=X'\boxtimes Y'$ (note that $D \otimes D' = (X \otimes X')\boxtimes(Y \otimes Y')$), and identity module unit isomorphism
\begin{gather*}
l_C = \id_C \colon \mathds{1}_{\mathcal{C}\boxtimes\mathcal{C}} \triangleright C \equiv (\mathds{1}_{\mathcal{C}} \boxtimes \mathds{1}_{\mathcal{C}}) \triangleright C = \mathds{1}_{\mathcal{C}} \otimes \mathds{1}_{\mathcal{C}} \otimes C = C \to C
\end{gather*}
(by assumed strictness of the monoidal structure). This defines a bimodule category (an object in $\underline{\mathsf{Pic}}(\mathcal{C}\boxtimes\mathcal{C})$) by Proposition \ref{prop:from-module-to-bimodule-category}. Furthermore, let $\mathcal{T}\defeq \mathcal{C}\boxtimes\mathcal{C}$ be the trivial $\mathcal{C}\boxtimes\mathcal{C}$-bimodule category (with action given by the tensor product in $\mathcal{C}\boxtimes\mathcal{C}$). Then there are bimodule equivalences
\begin{gather}
\mathcal{P} \boxtimes_{\mathcal{C}\boxtimes\mathcal{C}} \mathcal{T} \simeq \mathcal{P} \simeq \mathcal{T} \boxtimes_{\mathcal{C}\boxtimes\mathcal{C}} \mathcal{P} \, \, , \label{eq:equivalences1} \\
\mathcal{P} \boxtimes_{\mathcal{C}\boxtimes\mathcal{C}} \mathcal{P} \simeq \mathcal{T} \, \, . \label{eq:equivalence2}
\end{gather}
Therefore one obtains a representation $[\varrho]\colon S_2 \to \mathsf{Pic}(\mathcal{\mathcal{C}\boxtimes\mathcal{C}})$ of the symmetric group $S_2 = \{ e,\pi \}$ on $\mathcal{C}\boxtimes\mathcal{C}$ by $[\varrho](e) \defeq [\mathcal{T}]$ and $[\varrho](\pi) \defeq [\mathcal{P}]$ (the bimodule equivalence classes of $\mathcal{T}$ and $\mathcal{P}$, respectively).
\end{theorem}

The equivalences in Equation \eqref{eq:equivalences1} are immediate from the weak unitality of the relative Deligne product (Proposition \ref{prop:unitality}), which is part of the structure of the tricategory $\Bimod$. The non-trivial equivalence in Equation \eqref{eq:equivalence2} is established by realizing the bimodule category $\mathcal{P}$ as the category of modules over an algebra object $A_{\mathcal{P}} \in \mathcal{C}\boxtimes\mathcal{C}$ (Definition \ref{def:appendix-algebra-object}, Theorem \ref{thm:egno-mod-alg}). Then one shows that the algebra object $A_{\mathcal{P}} \otimes A_{\mathcal{P}}$ is Morita equivalent to the monoidal unit $\mathds{1}_{\mathcal{C}\boxtimes\mathcal{C}}$ (viewed as trivial algebra object), i.e. the module categories over these algebras are equivalent. By Proposition \ref{appendix-prop-relative-deligne}, $A_{\mathcal{P}} \otimes A_{\mathcal{P}}$ is the algebra representing $\mathcal{P} \boxtimes_{\mathcal{C}\boxtimes\mathcal{C}} \mathcal{P}$ in the sense of Theorem \ref{thm:egno-mod-alg}. Therefore, the bimodule category $\mathcal{P}$ indeed represents the non-trivial permutation in $S_2$ and thus describes the permutation twist defect in the topological bilayer phase $\mathcal{C}\boxtimes\mathcal{C}$ \cite{BaJQ1}.\medskip

We want to find a representation of $S_n$ on $\mathcal{C}^{\boxtimes n}$ for any $n\geq 2$, which coincides with the representation of Fuchs and Schweigert for $n=2$. Before doing so in Section \ref{sec:construction}, we show that $\mathcal{P}$ carries, under certain restrictions and up to equivalence, the unique $\mathcal{C}\boxtimes\mathcal{C}$-bimodule category structure (with left and right module structure related via the braiding).\medskip

For these tasks and in most cases it is easier to manipulate expressions in monoidal categories when they are translated into the graphical calculus of planar string diagrams. We want to close this section by briefly introducing the notation for the diagrams in this paper. A detailed introduction to graphical calculus can be found in \cite[Chapter 2]{turaev-virelizier}, which uses the same diagrammatic conventions apart from an opposite orientation of the arrows.

Let $\mathcal{C}$ be a pivotal and braided fusion category.
\begin{itemize}
\item Diagrams are read from bottom to top.
\item An object $X$ (or the identity morphism on $X$) is depicted by an upwards oriented vertical line and its dual object by a downwards oriented line, both labeled $X$. If the rigid structure of the category is not relevant in a diagram (as e.g. in the braiding diagrams below), all lines are drawn unoriented for simplicity.
\item Morphisms are labeled nodes (or vertices) sitting on these lines.
\item Composition of morphisms is vertical concatenation of morphisms.
\item The tensor product is horizontal juxtaposition from left to right.
\item The monoidal unit object $\mathds{1}$ and the monoidal coherence isomorphisms are not drawn. This means that one treats $\mathcal{C}$ as strict monoidal, which is justified by Mac Lane's strictness theorem \cite{maclane}.
\item Coevaluation and evaluation are represented by cups and caps, i.e.
\begin{gather*}
\coev_X \, \, \equiv \; \,
\begin{tikzpicture}[baseline=(current bounding box.center)]
     \draw[looseness=2] (-0.5,0.5) edge[out=270,in=270] node[above right,yshift=0.5cm] {$\,\mathsmaller{X}$} (0.5,0.5);
\end{tikzpicture} \; \; \; \; , \; \; \; \; \ev_X \, \, \equiv \; \,
\begin{tikzpicture}[baseline=(current bounding box.center)]
     \draw[looseness=2] (-0.5,0.5) edge[out=90,in=90] node[right,xshift=0.2cm] {$\,\mathsmaller{X}$} (0.5,0.5);
\end{tikzpicture}
\end{gather*}
for $\coev_X \colon \mathds{1}\to X^\ast \otimes X$ and $\ev_X \colon X \otimes X^\ast \to \mathds{1}$ in the case of right duality, and similarly with reversed arrows for $\widetilde{\coev}_X \colon \mathds{1}\to X \otimes \tensor[^{\ast}]{X}{}$ and $\widetilde{\ev}_X \colon \tensor[^{\ast}]{X}{} \otimes X \to \mathds{1}$ in the case of left duality. The pivotal structure isomorphisms are not drawn, i.e. left duals $\tensor[^{\ast}]{X}{}$ are strictly identified with right duals $X^{\ast}$.
\item The braidings $c_{X,Y}\colon X\otimes Y \to Y \otimes X$ and $c_{Y,X}^{-1} \colon X\otimes Y \to Y \otimes X$ are represented by
\begin{gather*}
c_{X,Y} = 
\scalebox{0.7}[-0.7]{
\begin{tikzpicture}[baseline=(current bounding box.center)]
\braid[number of strands=2] (braid) a_1;
\node[yshift=0.3cm] at (braid-1-s) {\scalebox{1}[-1]{$X$}};
\node[yshift=-0.3cm] at (braid-1-e) {\scalebox{1}[-1]{$X$}};
\node[yshift=0.3cm] at (braid-2-s) {\scalebox{1}[-1]{$Y$}};
\node[yshift=-0.3cm] at (braid-2-e) {\scalebox{1}[-1]{$Y$}};
\end{tikzpicture}} \; \; \; \; , \; \; \; \; \; \; \;
c_{Y,X}^{-1} = 
\scalebox{0.7}[-0.7]{
\begin{tikzpicture}[baseline=(current bounding box.center)]
\braid[number of strands=2] (braid) a_1^{-1};
\node[yshift=0.3cm] at (braid-1-s) {\scalebox{1}[-1]{$X$}};
\node[yshift=-0.3cm] at (braid-1-e) {\scalebox{1}[-1]{$X$}};
\node[yshift=0.3cm] at (braid-2-s) {\scalebox{1}[-1]{$Y$}};
\node[yshift=-0.3cm] at (braid-2-e) {\scalebox{1}[-1]{$Y$}};
\end{tikzpicture}} \; \, \, .
\end{gather*}
\end{itemize}
Each diagram represents a morphism in the category. Accordingly, an equality sign between diagrams means equality of the associated morphisms. By \cite{joyal-street}, diagrams which can be isotopically deformed into each other represent the same morphism. In particular, the braiding $c$ satisfies the second and third Reidemeister moves \cite[X.3]{kassel}: the second Reidemeister move
\begin{gather}
\scalebox{0.7}[-0.7]{
\begin{tikzpicture}[baseline=(current bounding box.center)]
\braid[number of strands=2] (braid) a_1 a_1^{-1};
\node[yshift=0.3cm] at (braid-1-s) {\scalebox{1}[-1]{$X$}};
\node[yshift=-0.3cm] at (braid-1-e) {\scalebox{1}[-1]{$X$}};
\node[yshift=0.3cm] at (braid-2-s) {\scalebox{1}[-1]{$Y$}};
\node[yshift=-0.3cm] at (braid-2-e) {\scalebox{1}[-1]{$Y$}};
\end{tikzpicture}} \, = \!\!\!\!\!\! \!\!\!\!\!\! \!\!\!\!
\scalebox{0.7}[-0.7]{
\begin{tikzpicture}[baseline=(current bounding box.center)]
\braid[number of strands=2] (braid) a_{-1} a_{-1};
\node[yshift=0.3cm] at (braid-1-s) {\scalebox{1}[-1]{$X$}};
\node[yshift=-0.3cm] at (braid-1-e) {\scalebox{1}[-1]{$X$}};
\node[yshift=0.3cm] at (braid-2-s) {\scalebox{1}[-1]{$Y$}};
\node[yshift=-0.3cm] at (braid-2-e) {\scalebox{1}[-1]{$Y$}};
\end{tikzpicture}}
\label{eq:reidemeister2}
\end{gather}
corresponds to the definition of the inverse braiding $c^{-1}$, and the third Reidemeister move,
\begin{gather}
\scalebox{0.7}[-0.7]{
\begin{tikzpicture}[baseline=(current bounding box.center)]
\braid[number of strands=3] (braid) a_1 a_2 a_1;
\node[yshift=0.3cm] at (braid-1-s) {\scalebox{1}[-1]{$X$}};
\node[yshift=-0.3cm] at (braid-1-e) {\scalebox{1}[-1]{$X$}};
\node[yshift=0.3cm] at (braid-2-s) {\scalebox{1}[-1]{$Y$}};
\node[yshift=-0.3cm] at (braid-2-e) {\scalebox{1}[-1]{$Y$}};
\node[yshift=0.3cm] at (braid-3-s) {\scalebox{1}[-1]{$Z$}};
\node[yshift=-0.3cm] at (braid-3-e) {\scalebox{1}[-1]{$Z$}};
\end{tikzpicture}} = 
\scalebox{0.7}[-0.7]{
\begin{tikzpicture}[baseline=(current bounding box.center)]
\braid[number of strands=3] (braid) a_2 a_1 a_2;
\node[yshift=0.3cm] at (braid-1-s) {\scalebox{1}[-1]{$X$}};
\node[yshift=-0.3cm] at (braid-1-e) {\scalebox{1}[-1]{$X$}};
\node[yshift=0.3cm] at (braid-2-s) {\scalebox{1}[-1]{$Y$}};
\node[yshift=-0.3cm] at (braid-2-e) {\scalebox{1}[-1]{$Y$}};
\node[yshift=0.3cm] at (braid-3-s) {\scalebox{1}[-1]{$Z$}};
\node[yshift=-0.3cm] at (braid-3-e) {\scalebox{1}[-1]{$Z$}};
\end{tikzpicture}}
\label{eq:reidemeister3}
\end{gather}
for all $X,Y,Z\in\mathcal{C}$, follows algebraically from the hexagon axioms and is also called the Yang-Baxter equation \cite[XIII.1.2]{kassel}.

\section{Module structures over the two-fold Deligne product}
\label{sec:modules-deligne-two}

Let $\mathcal{C}$ be a braided monoidal category. Without loss of generality we assume that the monoidal structure of $\mathcal{C}$ is strict, and we tacitly assume all categories to be $\mathds{k}$-linear and finitely semisimple. As mentioned above we will only care about left module structures, however all module categories will eventually be considered as bimodule categories, which is straightforward from the standard construction of Proposition \ref{prop:from-module-to-bimodule-category}. As a consequence, all module equivalences can be extended to bimodule equivalences by Proposition \ref{prop:from-module-functor-to-bimodule-functor}. \medskip

We now want to study left module structures on $\mathcal{C}$ over the two-fold Deligne product $\mathcal{C}\boxtimes\mathcal{C}$. Since $\mathcal{C}$ is braided monoidal, the category $\mathcal{C}\boxtimes\mathcal{C}$ is also braided monoidal. Consider the left $\mathcal{C}\boxtimes\mathcal{C}$-module category $\mathcal{P} \equiv (\mathcal{C},\triangleright,\psi)$ defined in Theorem \ref{thm:fs-thm}. Its module structure is just one member of a whole family of $\mathcal{C}\boxtimes\mathcal{C}$-module structures on $\mathcal{C}$. As the module associativity isomorphisms can be viewed as elements of a braid group \cite{artin-zopf}, their classification is out of reach. In particular, one could consider associativity isomorphisms involving arbitrarily high powers of the braiding. However, for the module action $\triangleright$ from Theorem \ref{thm:fs-thm}, higher powers of the braiding in Equation \eqref{eq:module-assoc-deligne2} lead to equivalent module categories if there exists a twist isomorphism (which is the case if $\mathcal{C}$ is ribbon or even modular, as will be assumed later) \cite[Theorem 2.4]{bfrs}. This motivates to only consider module structures with associativity isomorphisms containing a minimal number of braidings from now on. Here we call the number of braidings in a braid minimal if the image of the braid under the surjection $B_n \to S_n$ \cite[Lemma X.6.10]{kassel}, where $B_n$ is the braid group and $S_n$ is the symmetric group, consists of a minimal number of transpositions. 

It only remains the freedom to define module associativity isomorphisms using the inverse braiding $c^{-1}$ instead of $c$. Concerning the action, three objects $X,Y$ and $C$ can be arranged in different ways under the tensor product, each giving a candidate for a module action $(\mathcal{C}\boxtimes\mathcal{C})\times\mathcal{C} \to \mathcal{C}$.

\begin{lemma}
\label{lem:all-module-cats}
Let $\mathcal{C}$ be a braided monoidal category. There is a family of $\mathcal{C}\boxtimes\mathcal{C}$-module categories $\mathcal{P}^{xy,\varepsilon}$ for $x,y\in\{1,2,3\}$, $x\neq y$, and $\varepsilon\in\{\pm\}$, defined by 
\begin{gather*}
\mathcal{P}^{xy,\varepsilon} \defeq (\mathcal{C},\triangleright^{xy},\psi^{xy,\varepsilon})
\end{gather*}
with the identity morphism as module unit isomorphism. The action functors $(\mathcal{C}\boxtimes\mathcal{C}) \times \mathcal{C} \to \mathcal{C}$ are uniquely defined on $\boxtimes$-factorizable objects $X \boxtimes Y \in \mathcal{C}\boxtimes\mathcal{C}$ by
\begin{gather*}
(X\boxtimes Y)\triangleright^{12} C \defeq  X \otimes Y \otimes C \, \, ,\\
(X\boxtimes Y)\triangleright^{21} C \defeq  Y \otimes X \otimes C \, \, ,\\
(X\boxtimes Y)\triangleright^{13} C \defeq  X \otimes C \otimes Y \, \, ,\\
(X\boxtimes Y)\triangleright^{31} C \defeq  Y \otimes C \otimes X \, \, ,\\
(X\boxtimes Y)\triangleright^{23} C \defeq  C \otimes X \otimes Y \, \, ,\\
(X\boxtimes Y)\triangleright^{32} C \defeq  C \otimes Y \otimes X \, \, ,
\end{gather*}
where a pair of numbers $xy$ as an index means that $X$ and $Y$ appear at the $x$-th and $y$-th position, respectively, in the triple $- \otimes - \otimes -$. For objects $D=X\boxtimes Y$ and $D'=X'\boxtimes Y'$ in $\mathcal{C}\boxtimes\mathcal{C}$ the module natural isomorphisms $\psi^{xy,\varepsilon}$ are given by
\begin{gather*}
\psi^{12,\pm}_{D,D',C} \colon (D\otimes D') \triangleright^{12} C \to D \triangleright^{12} (D' \triangleright^{12} C) \, \, ,\\
\psi^{12,\pm}_{D,D',C} \defeq \id_X \otimes c^{\pm}_{X',Y} \otimes \id_{Y'\otimes C} \, \, ,
\end{gather*}
\begin{gather*}
\psi^{21,\pm}_{D,D',C} \colon (D\otimes D') \triangleright^{21} C \to D \triangleright^{21} (D' \triangleright^{21} C) \, \, ,\\
\psi^{21,\pm}_{D,D',C} \defeq \id_Y \otimes c^{\pm}_{Y',X} \otimes \id_{X'\otimes C} \, \, ,
\end{gather*}
\begin{gather*}
\psi^{13,\pm}_{D,D',C} \colon (D\otimes D') \triangleright^{13} C \to D \triangleright^{13} (D' \triangleright^{13} C) \, \, ,\\
\psi^{13,\pm}_{D,D',C} \defeq \id_{X\otimes X'\otimes C} \otimes c^{\pm}_{Y,Y'} \, \, ,
\end{gather*}
\begin{gather*}
\psi^{31,\pm}_{D,D',C} \colon (D\otimes D') \triangleright^{31} C \to D \triangleright^{31} (D' \triangleright^{31} C) \, \, ,\\
\psi^{31,\pm}_{D,D',C} \defeq \id_{Y\otimes Y'\otimes C} \otimes c^{\pm}_{X,X'} \, \, ,
\end{gather*}
\begin{gather*}
\psi^{23,\pm}_{D,D',C} \colon (D\otimes D') \triangleright^{23} C \to D \triangleright^{23} (D' \triangleright^{23} C) \, \, ,\\
\psi^{23,\pm}_{D,D',C} \defeq \id_C \otimes ((\id_{X'} \otimes c^{\pm}_{X\otimes Y,Y'}) \circ (c^{\pm}_{X,X'} \otimes \id_{Y \otimes Y'})) \, \, ,
\end{gather*}
\begin{gather*}
\psi^{32,\pm}_{D,D',C} \colon (D\otimes D') \triangleright^{32} C \to D \triangleright^{32} (D' \triangleright^{32} C) \, \, ,\\
\psi^{32,\pm}_{D,D',C} \defeq \id_C \otimes ((\id_{Y'} \otimes c^{\pm}_{Y\otimes X,X'}) \circ (c^{\pm}_{Y,Y'} \otimes \id_{X \otimes X'})) \, \, ,
\end{gather*}
where $c^{+}_{X,Y} \defeq c_{X,Y}$ and $c^{-}_{X,Y} \defeq c_{Y,X}^{-1}$.
\end{lemma}
Obviously these isomorphisms are natural (because the braiding is a natural isomorphism) and have minimal number of braidings. It is straightforward to see that $\mathcal{P}^{xy,\varepsilon}$ is indeed a $\mathcal{C}\boxtimes\mathcal{C}$-module category for every $x,y\in\{1,2,3\}$, $x\neq y$ and $\varepsilon\in\{\pm\}$. Note that we will not consider the situation where the module associativity isomorphisms associated to $\triangleright^{23}$ and $\triangleright^{32}$ contain mixed braidings.

Most of the time we will omit the plus sign and the index $12$ in the notation, so we write
\begin{gather*}
\psi^{xy} \defeq \psi^{xy,+} \, \, , \, \, \mathcal{P}^{xy} \defeq \mathcal{P}^{xy,+}
\end{gather*}
for all $x,y\in\{1,2,3\}$, $x\neq y$, and
\begin{gather}
\psi^{-} \defeq \psi^{12,-} \, \, , \, \, \mathcal{P}^{-} \defeq \mathcal{P}^{12,-} \, \, .
\label{eq:abbrev}
\end{gather}
From the definitions in Theorem \ref{thm:fs-thm} it follows that $\triangleright \equiv \triangleright^{12}$, $\psi \equiv \psi^{12,+}$ and $\mathcal{P} \equiv \mathcal{P}^{12,+}$.

Up to the choice of $c$ or $c^{-1}$ as braiding in the associativity isomorphisms, the following proposition shows that the module categories defined above are pairwise equivalent. The dependency on $\varepsilon$ (the choice of braiding) will be addressed afterwards.

\begin{proposition}
\label{prop:equivalent-2d}
Let $\mathcal{C}$ be a braided monoidal category. Then the left $\mathcal{C}\boxtimes\mathcal{C}$-module categories $\mathcal{P}$, $\mathcal{P}^{21,-}$, $\mathcal{P}^{13,-}$, $\mathcal{P}^{31}$, $\mathcal{P}^{23,-}$ and $\mathcal{P}^{32}$ are pairwise equivalent left $\mathcal{C}\boxtimes\mathcal{C}$-module categories.
\end{proposition}

\begin{proof}
Let $D=X\boxtimes Y$ and $D'=X'\boxtimes Y'$. Recall that an equivalence of module categories is given by a module functor which induces an equivalence of categories (Definition \ref{def:appendix-equivalence-modulecats}). As functors, all these equivalences will be $\id_{\mathcal{C}}$. Hence we just have to find module functor constraints satisfying the pentagon and triangle axioms of Definition \ref{def:module-functor} to establish module functors and thus module equivalences between the different module categories.\medskip

\begin{itemize}
\item We show the equivalence of $\mathcal{P}$ and $\mathcal{P}^{21,-}$ by specifying a module functor $(\id_{\mathcal{C}},f)\colon \mathcal{P} \to \mathcal{P}^{21,-}$. The morphism $f$ is defined as
\begin{gather*}
f_{X\boxtimes Y,C}\colon \id_{\mathcal{C}} ((X\boxtimes Y)\triangleright C) \equiv X \otimes Y \otimes C \to (X\boxtimes Y) \triangleright^{21} \id_{\mathcal{C}} (C) \equiv Y \otimes X \otimes C \, \, ,\\
f_{X\boxtimes Y,C} = c_{X,Y} \otimes \id_C \, \, ,
\end{gather*}
which is obviously a natural isomorphism. The triangle axiom is satisfied trivially, as we chose the monoidal structure of $\mathcal{C}$ to be strict. The pentagon axiom reads
\begin{gather*}
\psi^{21,-}_{D,D',C} \circ f_{D\otimes D',C} = (\id_D \triangleright^{21} f_{D',C}) \circ f_{D,D'\triangleright C} \circ \psi_{D,D',C} \, \, ,
\end{gather*}
that is, one has to verify that
\begin{equation}
\begin{gathered}
(\id_Y \otimes c_{X,Y'}^{-1} \otimes \id_{X'\otimes C}) \circ (c_{X\otimes X',Y\otimes Y'} \otimes \id_C) = \\ = (\id_{Y\otimes X} \otimes c_{X',Y'} \otimes \id_C) \circ (c_{X,Y} \otimes \id_{X'\otimes Y'\otimes C}) \circ (\id_X \otimes c_{X',Y} \otimes \id_{Y'\otimes C})
\end{gathered}
\label{eq:check-eq-string}
\end{equation}
is true. In graphical notation (Section \ref{sec:prelim}), Equation \eqref{eq:check-eq-string} looks like this:
\begin{gather*}
\scalebox{0.7}[-0.7]{
\begin{tikzpicture}[baseline=(current bounding box.center)]
\braid[number of strands=5] (braid) a_2 a_1-a_3 a_2 a_2^{-1};
\node[yshift=0.3cm] at (braid-1-s) {\scalebox{1}[-1]{$X$}};
\node[yshift=-0.3cm] at (braid-1-e) {\scalebox{1}[-1]{$X$}};
\node[yshift=0.3cm] at (braid-2-s) {\scalebox{1}[-1]{$X'$}};
\node[yshift=-0.3cm] at (braid-2-e) {\scalebox{1}[-1]{$X'$}};
\node[yshift=0.3cm] at (braid-3-s) {\scalebox{1}[-1]{$Y$}};
\node[yshift=-0.3cm] at (braid-3-e) {\scalebox{1}[-1]{$Y$}};
\node[yshift=0.3cm] at (braid-4-s) {\scalebox{1}[-1]{$Y'$}};
\node[yshift=-0.3cm] at (braid-4-e) {\scalebox{1}[-1]{$Y'$}};
\node[yshift=0.3cm] at (braid-5-s) {\scalebox{1}[-1]{$C$}};
\node[yshift=-0.3cm] at (braid-5-e) {\scalebox{1}[-1]{$C$}};
\end{tikzpicture}} =
\scalebox{0.7}[-0.7]{
\begin{tikzpicture}[baseline=(current bounding box.center)]
\braid[number of strands=5] (braid) a_2 a_1 a_3;
\node[yshift=0.3cm] at (braid-1-s) {\scalebox{1}[-1]{$X$}};
\node[yshift=-0.3cm] at (braid-1-e) {\scalebox{1}[-1]{$X$}};
\node[yshift=0.3cm] at (braid-2-s) {\scalebox{1}[-1]{$X'$}};
\node[yshift=-0.3cm] at (braid-2-e) {\scalebox{1}[-1]{$X'$}};
\node[yshift=0.3cm] at (braid-3-s) {\scalebox{1}[-1]{$Y$}};
\node[yshift=-0.3cm] at (braid-3-e) {\scalebox{1}[-1]{$Y$}};
\node[yshift=0.3cm] at (braid-4-s) {\scalebox{1}[-1]{$Y'$}};
\node[yshift=-0.3cm] at (braid-4-e) {\scalebox{1}[-1]{$Y'$}};
\node[yshift=0.3cm] at (braid-5-s) {\scalebox{1}[-1]{$C$}};
\node[yshift=-0.3cm] at (braid-5-e) {\scalebox{1}[-1]{$C$}};
\end{tikzpicture}}
\end{gather*}
This is evidently true by the second Reidemeister move (Equation \eqref{eq:reidemeister2}) and functoriality of $\otimes$. Hence $\mathcal{P}$ and $\mathcal{P}^{21,-}$ are equivalent module categories.

\item We now show that $(\id_{\mathcal{C}},g)\colon \mathcal{P} \to \mathcal{P}^{13,-}$ with
\begin{gather*}
g_{X\boxtimes Y,C}\colon \id_{\mathcal{C}} ((X\boxtimes Y)\triangleright C) \equiv X \otimes Y \otimes C \to (X\boxtimes Y) \triangleright^{13} \id_{\mathcal{C}} (C) \equiv X \otimes C \otimes Y \, \, ,\\
g_{X\boxtimes Y,C} = \id_X \otimes c_{C,Y}^{-1}
\end{gather*}
is a module functor. Again, the module functor constraint $g$ is obviously a natural isomorphism and the triangle axiom is automatically satisfied. The pentagon axiom is
\begin{gather*}
\psi^{13,-}_{D,D',C} \circ g_{D\otimes D',C} = (\id_D \triangleright^{13} \, g_{D',C}) \circ g_{D,D'\triangleright C} \circ \psi_{D,D',C} \, \, ,
\end{gather*}
that is
\begin{gather*}
(\id_{X\otimes X' \otimes C} \otimes c_{Y',Y}^{-1}) \circ (\id_{X\otimes X'} \otimes c_{C,Y\otimes Y'}^{-1}) = \\ = (\id_{X\otimes X'} \otimes c_{C,Y'}^{-1} \otimes \id_Y) \circ (\id_X \otimes c_{X'\otimes Y' \otimes C,Y}^{-1}) \circ (\id_X \otimes c_{X',Y} \otimes \id_{Y'\otimes C}) \, \, ,
\end{gather*}
which is translated into string diagrams
\begin{gather*}
\scalebox{0.7}[-0.7]{
\begin{tikzpicture}[baseline=(current bounding box.center)]
\braid[number of strands=5] (braid) a_4^{-1} a_3^{-1} a_4^{-1};
\node[yshift=0.3cm] at (braid-1-s) {\scalebox{1}[-1]{$X$}};
\node[yshift=-0.3cm] at (braid-1-e) {\scalebox{1}[-1]{$X$}};
\node[yshift=0.3cm] at (braid-2-s) {\scalebox{1}[-1]{$X'$}};
\node[yshift=-0.3cm] at (braid-2-e) {\scalebox{1}[-1]{$X'$}};
\node[yshift=0.3cm] at (braid-3-s) {\scalebox{1}[-1]{$Y$}};
\node[yshift=-0.3cm] at (braid-3-e) {\scalebox{1}[-1]{$Y$}};
\node[yshift=0.3cm] at (braid-4-s) {\scalebox{1}[-1]{$Y'$}};
\node[yshift=-0.3cm] at (braid-4-e) {\scalebox{1}[-1]{$Y'$}};
\node[yshift=0.3cm] at (braid-5-s) {\scalebox{1}[-1]{$C$}};
\node[yshift=-0.3cm] at (braid-5-e) {\scalebox{1}[-1]{$C$}};
\end{tikzpicture}} =
\scalebox{0.7}[-0.7]{
\begin{tikzpicture}[baseline=(current bounding box.center)]
\braid[number of strands=5] (braid) a_2 a_2^{-1} a_3^{-1} a_4^{-1} a_3^{-1};
\node[yshift=0.3cm] at (braid-1-s) {\scalebox{1}[-1]{$X$}};
\node[yshift=-0.3cm] at (braid-1-e) {\scalebox{1}[-1]{$X$}};
\node[yshift=0.3cm] at (braid-2-s) {\scalebox{1}[-1]{$X'$}};
\node[yshift=-0.3cm] at (braid-2-e) {\scalebox{1}[-1]{$X'$}};
\node[yshift=0.3cm] at (braid-3-s) {\scalebox{1}[-1]{$Y$}};
\node[yshift=-0.3cm] at (braid-3-e) {\scalebox{1}[-1]{$Y$}};
\node[yshift=0.3cm] at (braid-4-s) {\scalebox{1}[-1]{$Y'$}};
\node[yshift=-0.3cm] at (braid-4-e) {\scalebox{1}[-1]{$Y'$}};
\node[yshift=0.3cm] at (braid-5-s) {\scalebox{1}[-1]{$C$}};
\node[yshift=-0.3cm] at (braid-5-e) {\scalebox{1}[-1]{$C$}};
\end{tikzpicture}} \, \, .
\end{gather*}
The validity of this equation is evident by the second and third Reidemeister move (Equation \eqref{eq:reidemeister3}). This proves the equivalence of the module categories $\mathcal{P}$ and $\mathcal{P}^{13,-}$.

\item To show the equivalence $\mathcal{P} \simeq \mathcal{P}^{31}$ we prove that $(\id_{\mathcal{C}},l)\colon \mathcal{P} \to \mathcal{P}^{31}$ with
\begin{gather*}
l_{X\boxtimes Y,C}\colon \id_{\mathcal{C}} ((X\boxtimes Y)\triangleright C) \equiv X \otimes Y \otimes C \to (X\boxtimes Y) \triangleright^{31} \id_{\mathcal{C}} (C) \equiv Y \otimes C \otimes X \, \, ,\\
l_{X\boxtimes Y,C} = c_{X,Y \otimes C}
\end{gather*}
is a module functor. The pentagon axiom is
\begin{gather*}
\psi^{31}_{D,D',C} \circ l_{D\otimes D',C} = (\id_D \triangleright^{31} \, l_{D',C}) \circ l_{D,D'\triangleright C} \circ \psi_{D,D',C} \, \, ,
\end{gather*}
i.e.
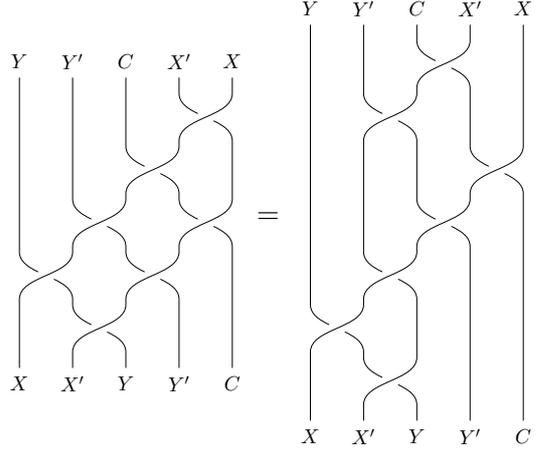
\begin{gather*}
\scalebox{0.7}[-0.7]{
\begin{tikzpicture}[baseline=(current bounding box.center)]
\braid[number of strands=5] (braid) a_2 a_1-a_3 a_2-a_4 a_3 a_4;
\node[yshift=0.3cm] at (braid-1-s) {\scalebox{1}[-1]{$X$}};
\node[yshift=-0.3cm] at (braid-1-e) {\scalebox{1}[-1]{$X$}};
\node[yshift=0.3cm] at (braid-2-s) {\scalebox{1}[-1]{$X'$}};
\node[yshift=-0.3cm] at (braid-2-e) {\scalebox{1}[-1]{$X'$}};
\node[yshift=0.3cm] at (braid-3-s) {\scalebox{1}[-1]{$Y$}};
\node[yshift=-0.3cm] at (braid-3-e) {\scalebox{1}[-1]{$Y$}};
\node[yshift=0.3cm] at (braid-4-s) {\scalebox{1}[-1]{$Y'$}};
\node[yshift=-0.3cm] at (braid-4-e) {\scalebox{1}[-1]{$Y'$}};
\node[yshift=0.3cm] at (braid-5-s) {\scalebox{1}[-1]{$C$}};
\node[yshift=-0.3cm] at (braid-5-e) {\scalebox{1}[-1]{$C$}};
\end{tikzpicture}} =
\scalebox{0.7}[-0.7]{
\begin{tikzpicture}[baseline=(current bounding box.center)]
\braid[number of strands=5] (braid) a_2 a_1 a_2 a_3 a_4 a_2 a_3;
\node[yshift=0.3cm] at (braid-1-s) {\scalebox{1}[-1]{$X$}};
\node[yshift=-0.3cm] at (braid-1-e) {\scalebox{1}[-1]{$X$}};
\node[yshift=0.3cm] at (braid-2-s) {\scalebox{1}[-1]{$X'$}};
\node[yshift=-0.3cm] at (braid-2-e) {\scalebox{1}[-1]{$X'$}};
\node[yshift=0.3cm] at (braid-3-s) {\scalebox{1}[-1]{$Y$}};
\node[yshift=-0.3cm] at (braid-3-e) {\scalebox{1}[-1]{$Y$}};
\node[yshift=0.3cm] at (braid-4-s) {\scalebox{1}[-1]{$Y'$}};
\node[yshift=-0.3cm] at (braid-4-e) {\scalebox{1}[-1]{$Y'$}};
\node[yshift=0.3cm] at (braid-5-s) {\scalebox{1}[-1]{$C$}};
\node[yshift=-0.3cm] at (braid-5-e) {\scalebox{1}[-1]{$C$}};
\end{tikzpicture}} \, \, .
\end{gather*}
This equation is obviously true by the third Reidemeister move, therefore we have an equivalence $\mathcal{P} \simeq \mathcal{P}^{31}$ of module categories.

\item Next, we verify that $(\id_{\mathcal{C}},h)\colon \mathcal{P} \to \mathcal{P}^{23,-}$ with
\begin{gather*}
h_{X\boxtimes Y,C}\colon \id_{\mathcal{C}} ((X\boxtimes Y)\triangleright C) \equiv X \otimes Y \otimes C \to (X\boxtimes Y) \triangleright^{23} \id_{\mathcal{C}} (C) \equiv C \otimes X \otimes Y \, \, ,\\
h_{X\boxtimes Y,C} = c_{C,X\otimes Y}^{-1}
\end{gather*}
is a module functor. The pentagon axiom is
\begin{gather*}
\psi^{23,-}_{D,D',C} \circ h_{D\otimes D',C} = (\id_D \triangleright^{23} h_{D',C}) \circ h_{D,D'\triangleright C} \circ \psi_{D,D',C} \, \, ,
\end{gather*}
which is diagrammatically
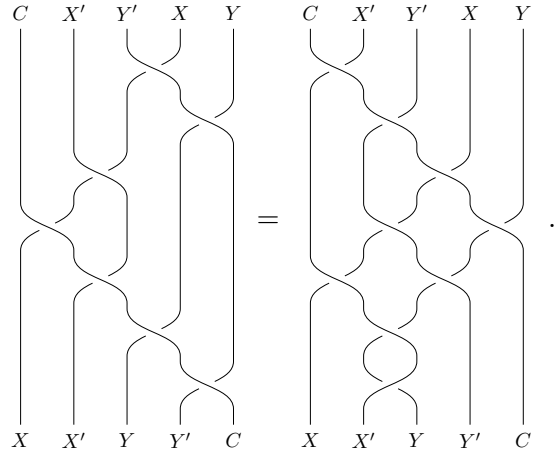
\begin{gather*}
\scalebox{0.7}[-0.7]{
\begin{tikzpicture}[baseline=(current bounding box.center)]
\braid[number of strands=5] (braid) a_4^{-1} a_3^{-1} a_2^{-1} a_1^{-1} a_2^{-1} a_4^{-1} a_3^{-1};
\node[yshift=0.3cm] at (braid-1-s) {\scalebox{1}[-1]{$X$}};
\node[yshift=-0.3cm] at (braid-1-e) {\scalebox{1}[-1]{$X$}};
\node[yshift=0.3cm] at (braid-2-s) {\scalebox{1}[-1]{$X'$}};
\node[yshift=-0.3cm] at (braid-2-e) {\scalebox{1}[-1]{$X'$}};
\node[yshift=0.3cm] at (braid-3-s) {\scalebox{1}[-1]{$Y$}};
\node[yshift=-0.3cm] at (braid-3-e) {\scalebox{1}[-1]{$Y$}};
\node[yshift=0.3cm] at (braid-4-s) {\scalebox{1}[-1]{$Y'$}};
\node[yshift=-0.3cm] at (braid-4-e) {\scalebox{1}[-1]{$Y'$}};
\node[yshift=0.3cm] at (braid-5-s) {\scalebox{1}[-1]{$C$}};
\node[yshift=-0.3cm] at (braid-5-e) {\scalebox{1}[-1]{$C$}};
\end{tikzpicture}} =
\scalebox{0.7}[-0.7]{
\begin{tikzpicture}[baseline=(current bounding box.center)]
\braid[number of strands=5] (braid) a_2 a_2^{-1} a_1^{-1}-a_3^{-1} a_2^{-1}-a_4^{-1} a_3^{-1} a_2^{-1} a_1^{-1};
\node[yshift=0.3cm] at (braid-1-s) {\scalebox{1}[-1]{$X$}};
\node[yshift=-0.3cm] at (braid-1-e) {\scalebox{1}[-1]{$X$}};
\node[yshift=0.3cm] at (braid-2-s) {\scalebox{1}[-1]{$X'$}};
\node[yshift=-0.3cm] at (braid-2-e) {\scalebox{1}[-1]{$X'$}};
\node[yshift=0.3cm] at (braid-3-s) {\scalebox{1}[-1]{$Y$}};
\node[yshift=-0.3cm] at (braid-3-e) {\scalebox{1}[-1]{$Y$}};
\node[yshift=0.3cm] at (braid-4-s) {\scalebox{1}[-1]{$Y'$}};
\node[yshift=-0.3cm] at (braid-4-e) {\scalebox{1}[-1]{$Y'$}};
\node[yshift=0.3cm] at (braid-5-s) {\scalebox{1}[-1]{$C$}};
\node[yshift=-0.3cm] at (braid-5-e) {\scalebox{1}[-1]{$C$}};
\end{tikzpicture}} \, \, .
\end{gather*}
Starting from the diagram on the right hand side we obtain, using the second Reidemeister move for the two lowest braidings in the first step and the third Reidemeister move and functoriality of $\otimes$ in the second step,
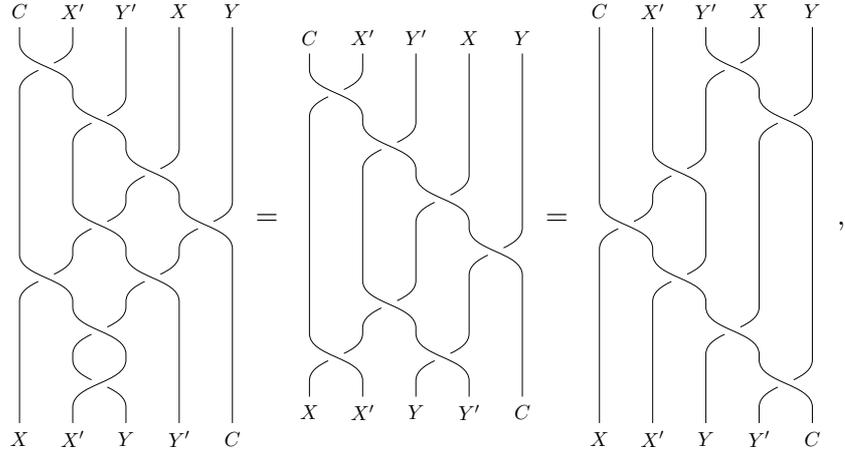
\begin{gather*}
\scalebox{0.7}[-0.7]{
\begin{tikzpicture}[baseline=(current bounding box.center)]
\braid[number of strands=5] (braid) a_2 a_2^{-1} a_1^{-1}-a_3^{-1} a_2^{-1}-a_4^{-1} a_3^{-1} a_2^{-1} a_1^{-1};
\node[yshift=0.3cm] at (braid-1-s) {\scalebox{1}[-1]{$X$}};
\node[yshift=-0.3cm] at (braid-1-e) {\scalebox{1}[-1]{$X$}};
\node[yshift=0.3cm] at (braid-2-s) {\scalebox{1}[-1]{$X'$}};
\node[yshift=-0.3cm] at (braid-2-e) {\scalebox{1}[-1]{$X'$}};
\node[yshift=0.3cm] at (braid-3-s) {\scalebox{1}[-1]{$Y$}};
\node[yshift=-0.3cm] at (braid-3-e) {\scalebox{1}[-1]{$Y$}};
\node[yshift=0.3cm] at (braid-4-s) {\scalebox{1}[-1]{$Y'$}};
\node[yshift=-0.3cm] at (braid-4-e) {\scalebox{1}[-1]{$Y'$}};
\node[yshift=0.3cm] at (braid-5-s) {\scalebox{1}[-1]{$C$}};
\node[yshift=-0.3cm] at (braid-5-e) {\scalebox{1}[-1]{$C$}};
\end{tikzpicture}} =
\scalebox{0.7}[-0.7]{
\begin{tikzpicture}[baseline=(current bounding box.center)]
\braid[number of strands=5] (braid) a_1^{-1}-a_3^{-1} a_2^{-1} a_4^{-1} a_3^{-1} a_2^{-1} a_1^{-1};
\node[yshift=0.3cm] at (braid-1-s) {\scalebox{1}[-1]{$X$}};
\node[yshift=-0.3cm] at (braid-1-e) {\scalebox{1}[-1]{$X$}};
\node[yshift=0.3cm] at (braid-2-s) {\scalebox{1}[-1]{$X'$}};
\node[yshift=-0.3cm] at (braid-2-e) {\scalebox{1}[-1]{$X'$}};
\node[yshift=0.3cm] at (braid-3-s) {\scalebox{1}[-1]{$Y$}};
\node[yshift=-0.3cm] at (braid-3-e) {\scalebox{1}[-1]{$Y$}};
\node[yshift=0.3cm] at (braid-4-s) {\scalebox{1}[-1]{$Y'$}};
\node[yshift=-0.3cm] at (braid-4-e) {\scalebox{1}[-1]{$Y'$}};
\node[yshift=0.3cm] at (braid-5-s) {\scalebox{1}[-1]{$C$}};
\node[yshift=-0.3cm] at (braid-5-e) {\scalebox{1}[-1]{$C$}};
\end{tikzpicture}} =
\scalebox{0.7}[-0.7]{
\begin{tikzpicture}[baseline=(current bounding box.center)]
\braid[number of strands=5] (braid) a_4^{-1} a_3^{-1} a_2^{-1} a_1^{-1} a_2^{-1} a_4^{-1} a_3^{-1};
\node[yshift=0.3cm] at (braid-1-s) {\scalebox{1}[-1]{$X$}};
\node[yshift=-0.3cm] at (braid-1-e) {\scalebox{1}[-1]{$X$}};
\node[yshift=0.3cm] at (braid-2-s) {\scalebox{1}[-1]{$X'$}};
\node[yshift=-0.3cm] at (braid-2-e) {\scalebox{1}[-1]{$X'$}};
\node[yshift=0.3cm] at (braid-3-s) {\scalebox{1}[-1]{$Y$}};
\node[yshift=-0.3cm] at (braid-3-e) {\scalebox{1}[-1]{$Y$}};
\node[yshift=0.3cm] at (braid-4-s) {\scalebox{1}[-1]{$Y'$}};
\node[yshift=-0.3cm] at (braid-4-e) {\scalebox{1}[-1]{$Y'$}};
\node[yshift=0.3cm] at (braid-5-s) {\scalebox{1}[-1]{$C$}};
\node[yshift=-0.3cm] at (braid-5-e) {\scalebox{1}[-1]{$C$}};
\end{tikzpicture}} \, \, ,
\end{gather*}
hence the pentagon axiom is satisfied.

\item Finally, there is a module functor $(\id_{\mathcal{C}},k)\colon \mathcal{P} \to \mathcal{P}^{32}$ with
\begin{gather*}
k_{X\boxtimes Y,C}\colon \id_{\mathcal{C}} ((X\boxtimes Y)\triangleright C) \equiv X \otimes Y \otimes C \to (X\boxtimes Y) \triangleright^{32} \id_{\mathcal{C}} (C) \equiv C \otimes Y \otimes X \, \, ,\\
k_{X\boxtimes Y,C} = (\id_C \otimes c_{X,Y}) \circ c_{X \otimes Y,C} \, \, ,
\end{gather*}
establishing a module equivalence. The pentagon axiom is
\begin{gather*}
\psi^{32}_{D,D',C} \circ k_{D\otimes D',C} = (\id_D \triangleright^{32} k_{D',C}) \circ k_{D,D'\triangleright C} \circ \psi_{D,D',C} \, \, ,
\end{gather*}
i.e.
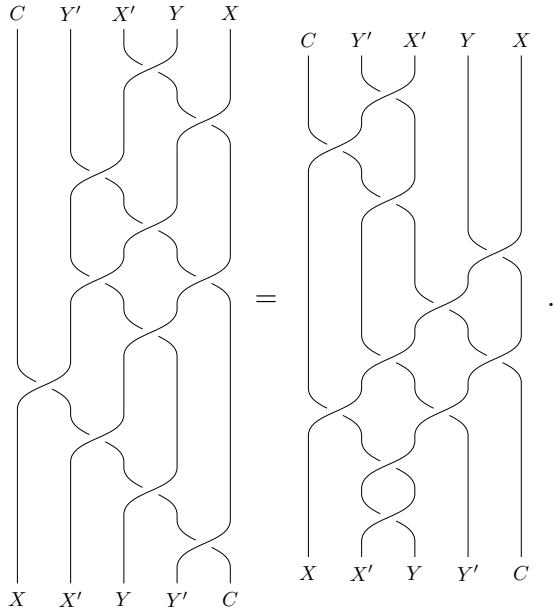
\begin{gather*}
\scalebox{0.7}[-0.7]{
\begin{tikzpicture}[baseline=(current bounding box.center)]
\braid[number of strands=5] (braid) a_4 a_3 a_2 a_1 a_3 a_2-a_4 a_3 a_2 a_4 a_3;
\node[yshift=0.3cm] at (braid-1-s) {\scalebox{1}[-1]{$X$}};
\node[yshift=-0.3cm] at (braid-1-e) {\scalebox{1}[-1]{$X$}};
\node[yshift=0.3cm] at (braid-2-s) {\scalebox{1}[-1]{$X'$}};
\node[yshift=-0.3cm] at (braid-2-e) {\scalebox{1}[-1]{$X'$}};
\node[yshift=0.3cm] at (braid-3-s) {\scalebox{1}[-1]{$Y$}};
\node[yshift=-0.3cm] at (braid-3-e) {\scalebox{1}[-1]{$Y$}};
\node[yshift=0.3cm] at (braid-4-s) {\scalebox{1}[-1]{$Y'$}};
\node[yshift=-0.3cm] at (braid-4-e) {\scalebox{1}[-1]{$Y'$}};
\node[yshift=0.3cm] at (braid-5-s) {\scalebox{1}[-1]{$C$}};
\node[yshift=-0.3cm] at (braid-5-e) {\scalebox{1}[-1]{$C$}};
\end{tikzpicture}} =
\scalebox{0.7}[-0.7]{
\begin{tikzpicture}[baseline=(current bounding box.center)]
\braid[number of strands=5] (braid) a_2 a_2 a_1-a_3 a_2-a_4 a_3 a_4 a_2 a_1 a_2;
\node[yshift=0.3cm] at (braid-1-s) {\scalebox{1}[-1]{$X$}};
\node[yshift=-0.3cm] at (braid-1-e) {\scalebox{1}[-1]{$X$}};
\node[yshift=0.3cm] at (braid-2-s) {\scalebox{1}[-1]{$X'$}};
\node[yshift=-0.3cm] at (braid-2-e) {\scalebox{1}[-1]{$X'$}};
\node[yshift=0.3cm] at (braid-3-s) {\scalebox{1}[-1]{$Y$}};
\node[yshift=-0.3cm] at (braid-3-e) {\scalebox{1}[-1]{$Y$}};
\node[yshift=0.3cm] at (braid-4-s) {\scalebox{1}[-1]{$Y'$}};
\node[yshift=-0.3cm] at (braid-4-e) {\scalebox{1}[-1]{$Y'$}};
\node[yshift=0.3cm] at (braid-5-s) {\scalebox{1}[-1]{$C$}};
\node[yshift=-0.3cm] at (braid-5-e) {\scalebox{1}[-1]{$C$}};
\end{tikzpicture}} \, \, .
\end{gather*}
If we start from the right hand side, this equation is verified after a few rearrangements of braidings:
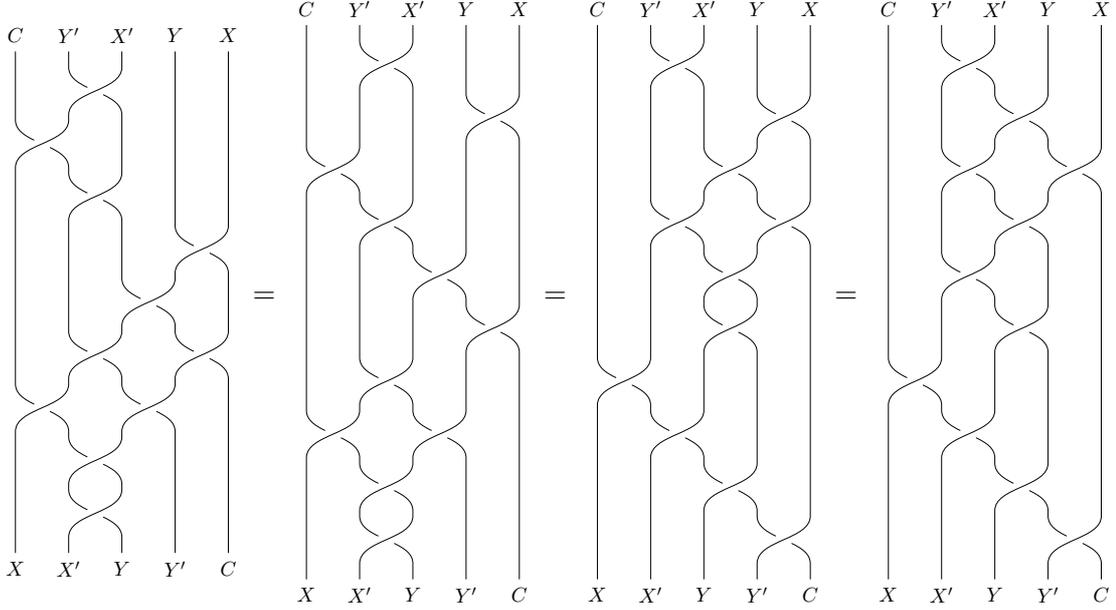
\begin{gather*}
\scalebox{0.7}[-0.7]{
\begin{tikzpicture}[baseline=(current bounding box.center)]
\braid[number of strands=5] (braid) a_2 a_2 a_1-a_3 a_2-a_4 a_3 a_4 a_2 a_1 a_2;
\node[yshift=0.3cm] at (braid-1-s) {\scalebox{1}[-1]{$X$}};
\node[yshift=-0.3cm] at (braid-1-e) {\scalebox{1}[-1]{$X$}};
\node[yshift=0.3cm] at (braid-2-s) {\scalebox{1}[-1]{$X'$}};
\node[yshift=-0.3cm] at (braid-2-e) {\scalebox{1}[-1]{$X'$}};
\node[yshift=0.3cm] at (braid-3-s) {\scalebox{1}[-1]{$Y$}};
\node[yshift=-0.3cm] at (braid-3-e) {\scalebox{1}[-1]{$Y$}};
\node[yshift=0.3cm] at (braid-4-s) {\scalebox{1}[-1]{$Y'$}};
\node[yshift=-0.3cm] at (braid-4-e) {\scalebox{1}[-1]{$Y'$}};
\node[yshift=0.3cm] at (braid-5-s) {\scalebox{1}[-1]{$C$}};
\node[yshift=-0.3cm] at (braid-5-e) {\scalebox{1}[-1]{$C$}};
\end{tikzpicture}} =
\scalebox{0.7}[-0.7]{
\begin{tikzpicture}[baseline=(current bounding box.center)]
\braid[number of strands=5] (braid) a_2 a_2 a_1-a_3 a_2 a_4 a_3 a_2 a_1 a_4 a_2;
\node[yshift=0.3cm] at (braid-1-s) {\scalebox{1}[-1]{$X$}};
\node[yshift=-0.3cm] at (braid-1-e) {\scalebox{1}[-1]{$X$}};
\node[yshift=0.3cm] at (braid-2-s) {\scalebox{1}[-1]{$X'$}};
\node[yshift=-0.3cm] at (braid-2-e) {\scalebox{1}[-1]{$X'$}};
\node[yshift=0.3cm] at (braid-3-s) {\scalebox{1}[-1]{$Y$}};
\node[yshift=-0.3cm] at (braid-3-e) {\scalebox{1}[-1]{$Y$}};
\node[yshift=0.3cm] at (braid-4-s) {\scalebox{1}[-1]{$Y'$}};
\node[yshift=-0.3cm] at (braid-4-e) {\scalebox{1}[-1]{$Y'$}};
\node[yshift=0.3cm] at (braid-5-s) {\scalebox{1}[-1]{$C$}};
\node[yshift=-0.3cm] at (braid-5-e) {\scalebox{1}[-1]{$C$}};
\end{tikzpicture}} = 
\scalebox{0.7}[-0.7]{
\begin{tikzpicture}[baseline=(current bounding box.center)]
\braid[number of strands=5] (braid) a_4 a_3 a_2 a_1 a_3 a_3 a_2-a_4 a_3 a_4 a_2;
\node[yshift=0.3cm] at (braid-1-s) {\scalebox{1}[-1]{$X$}};
\node[yshift=-0.3cm] at (braid-1-e) {\scalebox{1}[-1]{$X$}};
\node[yshift=0.3cm] at (braid-2-s) {\scalebox{1}[-1]{$X'$}};
\node[yshift=-0.3cm] at (braid-2-e) {\scalebox{1}[-1]{$X'$}};
\node[yshift=0.3cm] at (braid-3-s) {\scalebox{1}[-1]{$Y$}};
\node[yshift=-0.3cm] at (braid-3-e) {\scalebox{1}[-1]{$Y$}};
\node[yshift=0.3cm] at (braid-4-s) {\scalebox{1}[-1]{$Y'$}};
\node[yshift=-0.3cm] at (braid-4-e) {\scalebox{1}[-1]{$Y'$}};
\node[yshift=0.3cm] at (braid-5-s) {\scalebox{1}[-1]{$C$}};
\node[yshift=-0.3cm] at (braid-5-e) {\scalebox{1}[-1]{$C$}};
\end{tikzpicture}} =
\scalebox{0.7}[-0.7]{
\begin{tikzpicture}[baseline=(current bounding box.center)]
\braid[number of strands=5] (braid) a_4 a_3 a_2 a_1 a_3 a_2 a_3 a_2-a_4 a_3 a_2;
\node[yshift=0.3cm] at (braid-1-s) {\scalebox{1}[-1]{$X$}};
\node[yshift=-0.3cm] at (braid-1-e) {\scalebox{1}[-1]{$X$}};
\node[yshift=0.3cm] at (braid-2-s) {\scalebox{1}[-1]{$X'$}};
\node[yshift=-0.3cm] at (braid-2-e) {\scalebox{1}[-1]{$X'$}};
\node[yshift=0.3cm] at (braid-3-s) {\scalebox{1}[-1]{$Y$}};
\node[yshift=-0.3cm] at (braid-3-e) {\scalebox{1}[-1]{$Y$}};
\node[yshift=0.3cm] at (braid-4-s) {\scalebox{1}[-1]{$Y'$}};
\node[yshift=-0.3cm] at (braid-4-e) {\scalebox{1}[-1]{$Y'$}};
\node[yshift=0.3cm] at (braid-5-s) {\scalebox{1}[-1]{$C$}};
\node[yshift=-0.3cm] at (braid-5-e) {\scalebox{1}[-1]{$C$}};
\end{tikzpicture}}
\end{gather*}
The first step in this continued equality is evident from functoriality of $\otimes$, whereas the other steps are achieved using the third Reidemeister move. By another third Reidemeister move we arrive at the left hand side of the pentagon axiom. 
\end{itemize}

To wrap up, all stated module categories are pairwise equivalent.
\end{proof}

The definition of the module associativity isomorphisms for the module categories considered in this proposition does not rely on a consistent choice of braiding. For example, the module associativity $\psi^{21,-}$ of $\mathcal{P}^{21,-}$ is defined using the braiding $c^{-1}$, and $\psi$ of $\mathcal{P}$ contains the braiding $c$; but the actions of $X\boxtimes Y \in \mathcal{C}\boxtimes\mathcal{C}$ on $C\in\mathcal{C}$ under $\triangleright^{21}$ and $\triangleright$ only differ by a transposition of $X$ and $Y$. It is therefore desirable to determine the relation between $\mathcal{P}$ and $\mathcal{P}^{21}$, and in turn also the other $\mathcal{C}\boxtimes\mathcal{C}$-module categories $\mathcal{P}^{xy} \defeq \mathcal{P}^{xy,+}$ of Lemma \ref{lem:all-module-cats}.

However, it turns out that $\mathcal{P}$ and $\mathcal{P}^{21}$ are inequivalent $\mathcal{C}\boxtimes\mathcal{C}$-module categories, unless $\mathcal{C}$ is endowed with a twist, i.e. a natural isomorphism $\theta \colon 1_{\mathcal{C}} \to 1_{\mathcal{C}}$ with components $\theta_X \colon X \to X$ satisfying $\theta_{X\otimes Y} = c_{Y,X} \circ c_{X,Y} \circ (\theta_X \otimes \theta_Y)$ for all objects $X,Y$. Similarly, a twist is needed to establish equivalences between all module categories $\mathcal{P}^{xy}$. This follows from \cite[Lemma 2.5]{bfrs}, which is restated here in a simplified form.

\begin{lemma}
\label{lem:equivalence-braidings}
Let $\mathcal{C}$ be a braided strict monoidal category with twist $\theta$. Then the module categories $\mathcal{P}$ and $\mathcal{P}^{-}$ of Lemma \ref{lem:all-module-cats} (with abbreviations from Equation \eqref{eq:abbrev}) are equivalent via the module functor $(\id_{\mathcal{C}},p)\colon \mathcal{P} \to \mathcal{P}^{-}$, with
\begin{gather*}
p_{X\boxtimes Y, C} = \id_X \otimes (\theta_{Y \otimes C}^{-1} \circ (\id_Y \otimes \theta_C)) \, \equiv \, \,
\scalebox{0.7}[-0.7]{
\begin{tikzpicture}[trim left=0.993cm,baseline=(current bounding box.center)]
\braid[number of strands=3] (braid) a_{-1} a_{-1} a_{-1};
\node[yshift=0.3cm] at (braid-1-s) {\scalebox{1}[-1]{$X$}};
\node[yshift=-0.3cm] at (braid-1-e) {\scalebox{1}[-1]{$X$}};
\node[yshift=0.3cm] at (braid-2-s) {\scalebox{1}[-1]{$Y$}};
\node[yshift=-0.3cm] at (braid-2-e) {\scalebox{1}[-1]{$Y$}};
\node[yshift=0.3cm] at (braid-3-s) {\scalebox{1}[-1]{$C$}};
\node[yshift=-0.3cm] at (braid-3-e) {\scalebox{1}[-1]{$C$}};
\filldraw[draw=black,dashed,fill=white] (3,-0.9) circle (0.4);
\filldraw[draw=black,dashed,fill=white] (2.5,-2.6) ellipse (0.9cm and 0.4cm);
\node at (3,-0.9) {\scalebox{1}[-1]{$\theta$}};
\node at (2.5,-2.6) {\scalebox{1}[-1]{$\theta^{-1}$}};
\end{tikzpicture}} \, \, .
\end{gather*}
\end{lemma}
\begin{proof}
Since $\theta$ is a twist, the relation
\begin{gather*}
\theta_{Y \otimes C}^{-1} = (\theta_{Y}^{-1} \otimes \theta_{C}^{-1}) \circ c_{Y,C}^{-1} \circ c_{C,Y}^{-1}
\end{gather*}
allows to write
\begin{gather*}
p_{X\boxtimes Y, C} = \id_X \otimes ( (\theta_Y^{-1} \otimes \id_C) \circ c_{Y,C}^{-1} \circ c_{C,Y}^{-1} ) \, \equiv \, \,
\scalebox{0.7}[-0.7]{
\begin{tikzpicture}[trim left=0.993cm,baseline=(current bounding box.center)]
\braid[number of strands=3] (braid) a_2^{-1} a_2^{-1} a_{-1};
\node[yshift=0.3cm] at (braid-1-s) {\scalebox{1}[-1]{$X$}};
\node[yshift=-0.3cm] at (braid-1-e) {\scalebox{1}[-1]{$X$}};
\node[yshift=0.3cm] at (braid-2-s) {\scalebox{1}[-1]{$Y$}};
\node[yshift=-0.3cm] at (braid-2-e) {\scalebox{1}[-1]{$Y$}};
\node[yshift=0.3cm] at (braid-3-s) {\scalebox{1}[-1]{$C$}};
\node[yshift=-0.3cm] at (braid-3-e) {\scalebox{1}[-1]{$C$}};
\filldraw[draw=black,dashed,fill=white] (2,-2.7) circle (0.4);
\node at (2,-2.7) {\scalebox{1}[-1]{$\theta^{-1}$}};
\end{tikzpicture}} \, \, .
\end{gather*}
The pentagon axiom for $(\id_{\mathcal{C}},p)$ is
\begin{gather*}
\psi^{-}_{D,D',C} \circ p_{D\otimes D',C} = (\id_D \triangleright p_{D',C}) \circ p_{D,D'\triangleright C} \circ \psi_{D,D',C} \, \, ,
\end{gather*}
where
\begin{gather*}
p_{D \otimes D',C} = \id_{X \otimes X'} \otimes ( (\theta_{Y \otimes Y'}^{-1} \otimes \id_C) \circ c_{Y \otimes Y',C}^{-1} \circ c_{C,Y \otimes Y'}^{-1} ) = \\ = \id_{X \otimes X'} \otimes \left\lbrace \left[ \left( (\theta_Y^{-1} \otimes \theta_{Y'}^{-1}) \circ c_{Y,Y'}^{-1} \circ c_{Y',Y}^{-1} \right) \otimes \id_C \right] \circ c_{Y \otimes Y',C}^{-1} \circ c_{C,Y \otimes Y'}^{-1} \right\rbrace \, \, .
\end{gather*}
We verify that $(\id_{\mathcal{C}},p)$ satisfies this axiom. By naturality, the twist morphisms can be eliminated by composing both sides of the pentagon axiom with $\theta_Y \otimes \theta_{Y'}$ from the left, resulting in
\begin{gather*}
\scalebox{0.7}[-0.7]{
\begin{tikzpicture}[baseline=(current bounding box.center)]
\braid[number of strands=5] (braid) a_4^{-1} a_3^{-1} a_3^{-1} a_4^{-1} a_3^{-1} a_3^{-1} a_2^{-1};
\node[yshift=0.3cm] at (braid-1-s) {\scalebox{1}[-1]{$X$}};
\node[yshift=-0.3cm] at (braid-1-e) {\scalebox{1}[-1]{$X$}};
\node[yshift=0.3cm] at (braid-2-s) {\scalebox{1}[-1]{$X'$}};
\node[yshift=-0.3cm] at (braid-2-e) {\scalebox{1}[-1]{$X'$}};
\node[yshift=0.3cm] at (braid-3-s) {\scalebox{1}[-1]{$Y$}};
\node[yshift=-0.3cm] at (braid-3-e) {\scalebox{1}[-1]{$Y$}};
\node[yshift=0.3cm] at (braid-4-s) {\scalebox{1}[-1]{$Y'$}};
\node[yshift=-0.3cm] at (braid-4-e) {\scalebox{1}[-1]{$Y'$}};
\node[yshift=0.3cm] at (braid-5-s) {\scalebox{1}[-1]{$C$}};
\node[yshift=-0.3cm] at (braid-5-e) {\scalebox{1}[-1]{$C$}};
\end{tikzpicture}} =
\scalebox{0.7}[-0.7]{
\begin{tikzpicture}[baseline=(current bounding box.center)]
\braid[number of strands=5] (braid) a_2 a_2^{-1} a_3^{-1} a_4^{-1} a_4^{-1} a_3^{-1} a_2^{-1} a_4^{-1} a_4^{-1};
\node[yshift=0.3cm] at (braid-1-s) {\scalebox{1}[-1]{$X$}};
\node[yshift=-0.3cm] at (braid-1-e) {\scalebox{1}[-1]{$X$}};
\node[yshift=0.3cm] at (braid-2-s) {\scalebox{1}[-1]{$X'$}};
\node[yshift=-0.3cm] at (braid-2-e) {\scalebox{1}[-1]{$X'$}};
\node[yshift=0.3cm] at (braid-3-s) {\scalebox{1}[-1]{$Y$}};
\node[yshift=-0.3cm] at (braid-3-e) {\scalebox{1}[-1]{$Y$}};
\node[yshift=0.3cm] at (braid-4-s) {\scalebox{1}[-1]{$Y'$}};
\node[yshift=-0.3cm] at (braid-4-e) {\scalebox{1}[-1]{$Y'$}};
\node[yshift=0.3cm] at (braid-5-s) {\scalebox{1}[-1]{$C$}};
\node[yshift=-0.3cm] at (braid-5-e) {\scalebox{1}[-1]{$C$}};
\end{tikzpicture}} \, \, .
\end{gather*}
To prove this equation, it is instructive to perform one intermediate step, starting from the right hand side:
\begin{gather*}
\scalebox{0.7}[-0.7]{
\begin{tikzpicture}[baseline=(current bounding box.center)]
\braid[number of strands=5] (braid) a_2 a_2^{-1} a_3^{-1} a_4^{-1} a_4^{-1} a_3^{-1} a_2^{-1} a_4^{-1} a_4^{-1};
\node[yshift=0.3cm] at (braid-1-s) {\scalebox{1}[-1]{$X$}};
\node[yshift=-0.3cm] at (braid-1-e) {\scalebox{1}[-1]{$X$}};
\node[yshift=0.3cm] at (braid-2-s) {\scalebox{1}[-1]{$X'$}};
\node[yshift=-0.3cm] at (braid-2-e) {\scalebox{1}[-1]{$X'$}};
\node[yshift=0.3cm] at (braid-3-s) {\scalebox{1}[-1]{$Y$}};
\node[yshift=-0.3cm] at (braid-3-e) {\scalebox{1}[-1]{$Y$}};
\node[yshift=0.3cm] at (braid-4-s) {\scalebox{1}[-1]{$Y'$}};
\node[yshift=-0.3cm] at (braid-4-e) {\scalebox{1}[-1]{$Y'$}};
\node[yshift=0.3cm] at (braid-5-s) {\scalebox{1}[-1]{$C$}};
\node[yshift=-0.3cm] at (braid-5-e) {\scalebox{1}[-1]{$C$}};
\draw[black,dashed] (4.5,-7.75) circle (0.5);
\draw[->,loosely dashed] plot [smooth, tension=0.3] coordinates { (4.5,-7.2) (4.5,-6) (3.5,-4.4) (4.5,-2.1) (4.5,-0.9)};
\end{tikzpicture}} =
\scalebox{0.7}[-0.7]{
\begin{tikzpicture}[baseline=(current bounding box.center)]
\braid[number of strands=5] (braid) a_4^{-1} a_3^{-1} a_4^{-1} a_4^{-1} a_3^{-1} a_4^{-1} a_2^{-1};
\node[yshift=0.3cm] at (braid-1-s) {\scalebox{1}[-1]{$X$}};
\node[yshift=-0.3cm] at (braid-1-e) {\scalebox{1}[-1]{$X$}};
\node[yshift=0.3cm] at (braid-2-s) {\scalebox{1}[-1]{$X'$}};
\node[yshift=-0.3cm] at (braid-2-e) {\scalebox{1}[-1]{$X'$}};
\node[yshift=0.3cm] at (braid-3-s) {\scalebox{1}[-1]{$Y$}};
\node[yshift=-0.3cm] at (braid-3-e) {\scalebox{1}[-1]{$Y$}};
\node[yshift=0.3cm] at (braid-4-s) {\scalebox{1}[-1]{$Y'$}};
\node[yshift=-0.3cm] at (braid-4-e) {\scalebox{1}[-1]{$Y'$}};
\node[yshift=0.3cm] at (braid-5-s) {\scalebox{1}[-1]{$C$}};
\node[yshift=-0.3cm] at (braid-5-e) {\scalebox{1}[-1]{$C$}};
\draw[black,dashed] (4.5,-0.75) circle (0.5);
\draw[black,dashed] (4.5,-3.25) ellipse (0.7cm and 1.3cm);
\draw[->,loosely dashed] (4.5,-4.55) to (3.5,-5.55);
\end{tikzpicture}} =
\scalebox{0.7}[-0.7]{
\begin{tikzpicture}[baseline=(current bounding box.center)]
\braid[number of strands=5] (braid) a_4^{-1} a_3^{-1} a_3^{-1} a_4^{-1} a_3^{-1} a_3^{-1} a_2^{-1};
\node[yshift=0.3cm] at (braid-1-s) {\scalebox{1}[-1]{$X$}};
\node[yshift=-0.3cm] at (braid-1-e) {\scalebox{1}[-1]{$X$}};
\node[yshift=0.3cm] at (braid-2-s) {\scalebox{1}[-1]{$X'$}};
\node[yshift=-0.3cm] at (braid-2-e) {\scalebox{1}[-1]{$X'$}};
\node[yshift=0.3cm] at (braid-3-s) {\scalebox{1}[-1]{$Y$}};
\node[yshift=-0.3cm] at (braid-3-e) {\scalebox{1}[-1]{$Y$}};
\node[yshift=0.3cm] at (braid-4-s) {\scalebox{1}[-1]{$Y'$}};
\node[yshift=-0.3cm] at (braid-4-e) {\scalebox{1}[-1]{$Y'$}};
\node[yshift=0.3cm] at (braid-5-s) {\scalebox{1}[-1]{$C$}};
\node[yshift=-0.3cm] at (braid-5-e) {\scalebox{1}[-1]{$C$}};
\draw[black,dashed] (3.5,-5.2) ellipse (0.7cm and 1.3cm);
\end{tikzpicture}} \, \, .
\end{gather*}
Hence $(\id_{\mathcal{C}},p)\colon \mathcal{P} \to \mathcal{P}^{-}$ is a module functor.
\end{proof}

From \cite[Theorem 2.4]{bfrs} it also follows that the twist is necessary: the module categories $\mathcal{P}$ and $\mathcal{P}^{-}$ (and in turn also a whole family of module structures on $\mathcal{C}$ involving higher powers of the braiding) are equivalent if and only if $\mathcal{C}$ is endowed with a twist.\medskip

We now assume $\mathcal{C}$ to be a modular tensor category, but the following corollary still holds if $\mathcal{C}$ is just a braided monoidal category with a twist.
\begin{corollary}
\label{cor:modules2}
Let $\mathcal{C}$ be a modular tensor category. Then all module categories $\mathcal{P}^{xy,\varepsilon}$ of Lemma \ref{lem:all-module-cats} (for $x,y\in\{1,2,3\}$, $x\neq y$, and $\varepsilon\in\{\pm\}$) are pairwise equivalent left $\mathcal{C}\boxtimes\mathcal{C}$-module categories.
\end{corollary}
\begin{proof}
In Proposition \ref{prop:equivalent-2d} it is shown that $\mathcal{P}$ is equivalent to $\mathcal{P}^{21,-}$, $\mathcal{P}^{13,-}$, $\mathcal{P}^{31}$, $\mathcal{P}^{23,-}$ and $\mathcal{P}^{32}$. There is an analogous result if all braidings are reversed throughout Proposition \ref{prop:equivalent-2d} (replacing $c$ by $c^{-1}$ and vice versa), i.e. the module categories $\mathcal{P}^{-}$, $\mathcal{P}^{21}$, $\mathcal{P}^{13}$, $\mathcal{P}^{31,-}$, $\mathcal{P}^{23}$ and $\mathcal{P}^{32,-}$ are pairwise equivalent. By Lemma \ref{lem:equivalence-braidings}, the left $\mathcal{C}\boxtimes\mathcal{C}$-module categories $\mathcal{P}$ and $\mathcal{P}^{-}$ are equivalent, and the claim follows.
\end{proof}

We emphasize that the corollary does not provide a complete classification of $\mathcal{C}\boxtimes\mathcal{C}$-module structures on $\mathcal{C}$. However, we showed that certain $\mathcal{C}\boxtimes\mathcal{C}$-module structures on $\mathcal{C}$, which are restricted by a set of reasonable assumptions (module actions which arrange three objects under the tensor product, a minimal number of braidings, and no mixed braidings in the associativity isomorphisms corresponding to the actions $\triangleright^{23}$ and $\triangleright^{32}$), are equivalent. Consequently, different choices for the $S_2$ permutation action in the literature are equivalent: for example, $\mathcal{P}$ appears in \cite{fuchs-schweigert}, whereas $\mathcal{P}^{13}$ is considered in \cite{ejp2018} in the same context. From the physical point of view, our partial classification relates several possible candidates for the permutation twist surface defect of a bilayer system.

\section{Construction of a representation}
\label{sec:construction}

In this section we construct a representation of $S_n$ on $\mathcal{C}^{\boxtimes n}$ for a modular tensor category $\mathcal{C}$ for any $n\geq 2$, which extends the representation of Fuchs and Schweigert for $n=2$ (Theorem \ref{thm:fs-thm}). Let $n \geq 2$ be an integer and let $\tau_j \defeq  (j \, (j+1))\in S_n$ be the $j$-th adjacent transposition for $j=1,\ldots,n-1$. We choose the Coxeter presentation $S_n = \langle S | R \rangle$ of the symmetric group, where
\begin{gather*}
S \defeq  \{ \tau_j, \, j=1,\ldots,n-1 \}
\end{gather*}
is the set of generators, and the subset
\begin{align*}
R \defeq  & \{ \tau_j \tau_j, \, j=1,\ldots,n \} \cup \{ \tau_i \tau_j \tau_i^{-1} \tau_j^{-1}, \, |i-j|>1 \} \cup \\ & \cup \{ \tau_j \tau_{j+1} \tau_j \tau_{j+1}^{-1} \tau_j^{-1} \tau_{j+1}^{-1}, \, j=1,\ldots,n-2 \}
\end{align*}
of the free group of $S$ is the subset of relations. In order to construct a representation of $S_n$ on $\mathcal{C}^{\boxtimes n}$ in the sense of Definition \ref{def:representation-on-cat}, it is sufficient to specify a map $[\bar{\varrho}]\colon S\to\mathsf{Pic}(\mathcal{C}^{\boxtimes n})$ such that $[\bar{\varrho}](R)=\{ [\mathcal{C}^{\boxtimes n}] \}$ (the singleton containing the bimodule equivalence class of the trivial $\mathcal{C}^{\boxtimes n}$-bimodule category $\mathcal{C}^{\boxtimes n}$). The universal property of a presentation then yields a unique group homomorphism $[\varrho]\colon S_n \to\mathsf{Pic}(\mathcal{C}^{\boxtimes n})$. Therefore, we want to detect non-trivial candidates for invertible bimodule categories describing adjacent transpositions. Their bimodule equivalence classes are the elements of $\mathsf{Pic}(\mathcal{C}^{\boxtimes n})$ in the image of the map $[\bar{\varrho}]$. Concatenating the bimodule categories of adjacent transpositions under the relative Deligne product leads to representatives of general permutations. 

By definition, the bimodule categories of adjacent transpositions are supposed to be $\mathcal{C}^{\boxtimes n}$-bimodule categories $\varPi_j$, $j=1,\ldots,n-1$, satisfying the relations
\begin{enumerate}
\item $\varPi_j \boxtimes_n \varPi_j \simeq \mathcal{C}^{\boxtimes n}$ for all $j=1,\ldots,n-1$,
\item $\varPi_i \boxtimes_n \varPi_j \simeq \varPi_j \boxtimes_n \varPi_i$ for all $i,j$ with $|i-j|>1$,
\item $\varPi_j \boxtimes_n \varPi_{j+1} \boxtimes_n \varPi_j \simeq \varPi_{j+1} \boxtimes_n \varPi_{j} \boxtimes_n \varPi_{j+1}$ for all $j = 1,\ldots,n-2$
\end{enumerate}
with respect to the relative Deligne product $\boxtimes_n \defeq  \boxtimes_{\mathcal{C}^{\boxtimes n}}$, which are equivalences of $\mathcal{C}^{\boxtimes n}$-bimodule categories. In the third relation as well as in the following, we suppress the associativity equivalences for $\boxtimes_n$ (Proposition \ref{prop:assoc-rel-deligne}) by omitting particular bracketings in the relative Deligne product of more than two categories. This simplifies the treatment and the notation, however to obtain a precise expression for the equivalence one has to choose a bracketing and insert the canonical equivalences in the proper places. For further simplification we write $\mathcal{C}^{n}$ instead of $\mathcal{C}^{\boxtimes n}$ from now on. As in Section \ref{sec:modules-deligne-two} it suffices to consider $\boxtimes$-factorizable objects $\bigbox_{i=1}^{n} X_i$ of $\mathcal{C}^{n}$.

\begin{definition}
\label{def:generators-Sn}
Let $\mathcal{C}$ be a modular tensor category (strict monoidal without loss of generality) and $n \geq 2$ a natural number. Then $\{ \varPi_j \}_{j=1,\ldots,n-1}$ is a family of $\mathcal{C}^{n}$-bimodule categories defined by $\varPi_j = (\mathcal{C}^{n-1},\triangleright_j,\varphi^{(j)})$ (with bimodule structure obtained from Proposition \ref{prop:from-module-to-bimodule-category}), where the left action $\triangleright_j$ is defined by
\begin{gather*}
\left( \bigbox\limits_{i=1}^{n} X_i \right) \triangleright_j \left( \bigbox\limits_{k=1}^{n-1} C_k \right) \defeq  \bigbox\limits_{k=1}^{j-1} (X_k \otimes C_k) \boxtimes (X_j \otimes X_{j+1} \otimes C_j) \boxtimes \bigbox\limits_{k=j+1}^{n-1} (X_{k+1} \otimes C_k) \, \, ,
\end{gather*}
and, for objects $D\defeq \bigbox_{i=1}^{n} X_i$, $D'\defeq \bigbox_{i=1}^{n} X_i'$ and $C\defeq \bigbox\limits_{k=1}^{n-1} C_k$,
\begin{gather*}
\varphi^{(j)}_{D,D',C} \colon (D \otimes D') \triangleright_j \left( \bigbox\limits_{k=1}^{n-1} C_k \right) \to D \triangleright_j \left( D' \triangleright_j \left( \bigbox\limits_{k=1}^{n-1} C_k \right)\right)\\
\varphi^{(j)}_{D,D',C} = \bigbox\limits_{k=1}^{j-1} \id_{X_k \otimes X_k' \otimes C_k} \boxtimes \left(\id_{X_j} \otimes c_{X_j',X_{j+1}} \otimes \id_{X_{j+1}'\otimes C_j} \right) \boxtimes \bigbox\limits_{k=j+1}^{n-1} \id_{X_{k+1} \otimes X_{k+1}' \otimes C_k}
\end{gather*}
is the module associativity isomorphism . The module unit isomorphism is given by the identity morphism.
\end{definition}
We use the convention to omit factors in a Deligne product if they are not defined for a specific index (e.g. $\bigbox_{k=j+1}^{n-1}$ for $j=n-1$). In the module associativity isomorphism $\boxtimes$ is identified with the tensor product of $\mathds{k}$-vector spaces (for $\mathds{k}$-linear $\mathcal{C}$), since e.g. for morphisms in $\mathcal{C}^2$ one has $\Hom_{\mathcal{C}^2} (X \boxtimes Y,X' \boxtimes Y') = \Hom(X,X') \otimes_{\mathds{k}} \Hom(Y,Y')$ \cite[Proposition 1.11.2]{egno}.

The $\mathcal{C}^{n}$-bimodule category $\varPi_j$ can be written as the Deligne product 
\begin{gather*}
\varPi_j = \mathcal{C}^{j-1} \boxtimes \mathcal{P} \boxtimes \mathcal{C}^{n-j-1} \, \, ,
\end{gather*}
where $\mathcal{C}^{j-1}$ and $\mathcal{C}^{n-j-1}$ are the trivial $\mathcal{C}^{j-1}$- and $\mathcal{C}^{n-j-1}$-bimodule categories, respectively, and $\mathcal{P}\equiv(\mathcal{C},\triangleright,\psi)$ is the $\mathcal{C}^{2}$-bimodule category with underlying category $\mathcal{C}$ from Theorem \ref{thm:fs-thm}. The module associativity isomorphism $\varphi^{(j)}$ is of the form
\begin{gather*}
\varphi^{(j)} = \id^{\otimes_{\mathds{k}} (j-1)} \, \otimes_{\mathds{k}} \, \psi \, \otimes_{\mathds{k}} \, \id^{\otimes_{\mathds{k}} (n-j-1)} \, \, ,
\end{gather*}
where the morphism $\psi$ is the module associativity isomorphism of $\mathcal{P}$ from Equation \eqref{eq:module-assoc-deligne2}. For example, $\varPi_1 = \mathcal{P} \boxtimes \mathcal{C}^{n-2}$ with
\begin{gather*}
\left(\bigbox\limits_{i=1}^{n} X_i \right) \triangleright_1 \left(\bigbox\limits_{k=1}^{n-1} C_k \right) \defeq  (X_1 \otimes X_2 \otimes C_1) \boxtimes \bigbox\limits_{k=2}^{n-1} (X_{k+1} \otimes C_k) \, \, .
\end{gather*}
For $n=2$ the definition of $\varPi_1$ collapses to $\mathcal{P}$, as desired.

The main result of this paper is the following theorem.

\begin{theorem}[Representation of $S_n$]
\label{thm:rep-Sn}
Let $\mathcal{C}$ be a (monoidally strict) modular tensor category. For each natural number $n\geq 2$, the map $[\bar{\varrho}]\colon S\to\mathsf{Pic}(\mathcal{C}^{\boxtimes n})$, $\tau_j \mapsto [\varPi_j]$ with $\varPi_j = \mathcal{C}^{j-1} \boxtimes \mathcal{P} \boxtimes \mathcal{C}^{n-j-1}$, $j=1,\ldots,n-1$, satisfies $[\bar{\varrho}](R)=\{ [\mathcal{C}^{\boxtimes n}] \}$ and hence induces a representation $[\varrho]\colon S_n \to \mathsf{Pic}(\mathcal{C}^{\boxtimes n})$ of the symmetric group $S_n = \langle S | R \rangle$ on $\mathcal{C}^{n}$.
\end{theorem}
\textbf{Proof.}\\
The proof consists of three parts.

\proofpart{1}{For all $j=1,\ldots,n-1$, there is a bimodule equivalence $\varPi_j \boxtimes_n \varPi_j \simeq \mathcal{C}^n$.}
\label{proofpart:rep1}

For simplicity we write $\mathcal{C}^{k}$ for the trivial $\mathcal{C}^{k}$-bimodule category. The category $\mathcal{C}^{n}$ can be decomposed into $\mathcal{C}^{j-1} \boxtimes \mathcal{C}^{2} \boxtimes \mathcal{C}^{n-j-1}$, thus $\boxtimes_n = \boxtimes_{\mathcal{C}^{j-1} \boxtimes \mathcal{C}^{2} \boxtimes \mathcal{C}^{n-j-1}}$. (If $j=1$ respectively $j=n-1$ then $\mathcal{C}^{j-1}$ respectively $\mathcal{C}^{n-j-1}$ is omitted.) Accordingly,
\begin{gather*}
\varPi_j \boxtimes_n \varPi_j = (\mathcal{C}^{j-1} \boxtimes \mathcal{P} \boxtimes \mathcal{C}^{n-j-1}) \boxtimes_n (\mathcal{C}^{j-1} \boxtimes \mathcal{P} \boxtimes \mathcal{C}^{n-j-1}) \simeq \\ \simeq (\mathcal{C}^{j-1} \boxtimes_{j-1} \mathcal{C}^{j-1}) \boxtimes (\mathcal{P} \boxtimes_2 \mathcal{P}) \boxtimes (\mathcal{C}^{n-j-1} \boxtimes_{n-j-1} \mathcal{C}^{n-j-1}) \simeq \\ \simeq \mathcal{C}^{j-1} \boxtimes \mathcal{C}^{2} \boxtimes \mathcal{C}^{n-j-1} = \mathcal{C}^{n} \, \, ,
\end{gather*}
where we use that the relative Deligne product factorizes under the Deligne product, and that there exists an equivalence $\mathcal{P} \boxtimes_2 \mathcal{P} \simeq \mathcal{C}^{2}$ of bimodule categories by the results of \cite{fuchs-schweigert}. Hence all $\varPi_j$ are self-inverse; in particular they are invertible bimodule categories.

\proofpart{2}{For all $i,j = 1,\ldots,n-1$ with $|i-j|>1$ there is a bimodule equivalence $\varPi_i \boxtimes_n \varPi_j \simeq \varPi_j \boxtimes_n \varPi_i$.}
\label{proofpart:rep2}

Without loss of generality we choose $i>j+1$. Then we have
\begin{gather*}
\varPi_i = \mathcal{C}^{i-1} \boxtimes \mathcal{P} \boxtimes \mathcal{C}^{n-i-1} = \mathcal{C}^{j-1} \boxtimes \mathcal{C}^{2} \boxtimes \mathcal{C}^{i-j-2} \boxtimes \mathcal{P} \boxtimes \mathcal{C}^{n-i-1} \, \, ,
\end{gather*}
since $i-j \geq 2$. Similarly,
\begin{gather*}
\varPi_j = \mathcal{C}^{j-1} \boxtimes \mathcal{P} \boxtimes \mathcal{C}^{n-j-1} = \mathcal{C}^{j-1} \boxtimes \mathcal{P} \boxtimes \mathcal{C}^{i-j-2} \boxtimes \mathcal{C}^{2} \boxtimes \mathcal{C}^{n-i-1} \, \, ,
\end{gather*}
since $n-j-1 > n-i-1$. (In both decompositions $\mathcal{C}^{i-j-2}$ is omitted if $i-j=2$.) Note that this makes sense, as from $1 \leq i \leq n-1$ we have $j \leq i-2 \leq n-3$, and hence $n-j-1 \geq n-n+3-1=2$. Now we can decompose the relative Deligne product into several parts,
\begin{gather*}
\boxtimes_n = \boxtimes_{\mathcal{C}^{j-1} \boxtimes \mathcal{C}^{2} \boxtimes \mathcal{C}^{i-j-2} \boxtimes \mathcal{C}^{2} \boxtimes \mathcal{C}^{n-i-1}} \, \, ,
\end{gather*}
and therefore
\begin{gather*}
\varPi_i \boxtimes_n \varPi_j = (\mathcal{C}^{j-1} \boxtimes_{j-1} \mathcal{C}^{j-1}) \boxtimes (\mathcal{C}^{2} \boxtimes_2 \mathcal{P}) \boxtimes \\ \boxtimes (\mathcal{C}^{i-j-2} \boxtimes_{i-j-2} \mathcal{C}^{i-j-2}) \boxtimes (\mathcal{P} \boxtimes_2 \mathcal{C}^{2}) \boxtimes (\mathcal{C}^{n-i-1} \boxtimes_{n-i-1} \mathcal{C}^{n-i-1}) \simeq \\ \simeq \mathcal{C}^{j-1} \boxtimes \mathcal{P} \boxtimes \mathcal{C}^{i-j-2} \boxtimes \mathcal{P} \boxtimes \mathcal{C}^{n-i-1} \simeq \varPi_j \boxtimes_n \varPi_i \, \, .
\end{gather*}
Here we used the fact that $\mathcal{P}$ is a $\mathcal{C}^{2}$-bimodule category and hence there are equivalences $\mathcal{C}^{2} \boxtimes_2 \mathcal{P} \simeq \mathcal{P} \simeq \mathcal{P} \boxtimes_2 \mathcal{C}^{2}$ of bimodule categories (Proposition \ref{prop:unitality}).

\proofpart{3}{For all $j = 1,\ldots,n-2$ there is a bimodule equivalence $\varPi_j \boxtimes_n \varPi_{j+1} \boxtimes_n \varPi_j \simeq \varPi_{j+1} \boxtimes_n \varPi_{j} \boxtimes_n \varPi_{j+1}$.}

By definition we can write
\begin{gather*}
\varPi_j = \mathcal{C}^{j-1} \boxtimes \mathcal{P} \boxtimes \mathcal{C}^{n-j-1} = \mathcal{C}^{j-1} \boxtimes \mathcal{P} \boxtimes \mathcal{C} \boxtimes \mathcal{C}^{n-j-2} \, \, , \\
\varPi_{j+1} = \mathcal{C}^{j} \boxtimes \mathcal{P} \boxtimes \mathcal{C}^{n-j-2} = \mathcal{C}^{j-1} \boxtimes \mathcal{C} \boxtimes \mathcal{P} \boxtimes \mathcal{C}^{n-j-2} \, \, ,
\end{gather*}
thus
\begin{gather*}
\varPi_j \boxtimes_n \varPi_{j+1} \boxtimes_n \varPi_j \simeq \mathcal{C}^{j-1} \boxtimes [(\mathcal{P}\boxtimes\mathcal{C}) \boxtimes_3 (\mathcal{C}\boxtimes\mathcal{P}) \boxtimes_3 (\mathcal{P}\boxtimes\mathcal{C})] \boxtimes \mathcal{C}^{n-j-2} \, \, , \\
\varPi_{j+1} \boxtimes_n \varPi_{j} \boxtimes_n \varPi_{j+1} \simeq \mathcal{C}^{j-1} \boxtimes [(\mathcal{C}\boxtimes\mathcal{P}) \boxtimes_3 (\mathcal{P}\boxtimes\mathcal{C}) \boxtimes_3 (\mathcal{C}\boxtimes\mathcal{P})] \boxtimes \mathcal{C}^{n-j-2} \, \, .
\end{gather*}
The theorem is proven if we show that the two $\mathcal{C}^3$-bimodule categories enclosed in square brackets are equivalent. To summarize, it remains to show the following assertion:
\begin{proposition}
\label{prop:rep-Sn-3rd}
Let $\mathcal{P}$ be the $\mathcal{C}\boxtimes\mathcal{C}$-bimodule category from Theorem \ref{thm:fs-thm}, view $\mathcal{C}$ as the trivial $\mathcal{C}$-bimodule category, and define the $\mathcal{C}^{3}$-bimodule categories $\mathcal{P}_1 \defeq  \mathcal{P}\boxtimes\mathcal{C}$ and $\mathcal{P}_2 \defeq  \mathcal{C}\boxtimes\mathcal{P}$. Then there is an equivalence of $\mathcal{C}^{3}$-bimodule categories
\begin{gather*}
\mathcal{P}_1 \boxtimes_3 \mathcal{P}_2 \boxtimes_3 \mathcal{P}_1 \simeq \mathcal{P}_2 \boxtimes_3 \mathcal{P}_1 \boxtimes_3 \mathcal{P}_2 \, \, .
\end{gather*}
\end{proposition}
We reemphasize that different bracketings of $\mathcal{P}_1 \boxtimes_3 \mathcal{P}_2 \boxtimes_3 \mathcal{P}_1$ lead to equivalent bimodule categories (Proposition \ref{prop:assoc-rel-deligne}), hence it does not matter -- up to canonical equivalence -- if we interpret it as $(\mathcal{P}_1 \boxtimes_3 \mathcal{P}_2) \boxtimes_3 \mathcal{P}_1$ or $\mathcal{P}_1 \boxtimes_3 (\mathcal{P}_2 \boxtimes_3 \mathcal{P}_1)$. Proving this proposition is quite extensive and needs some preparation. \bigskip

From the results of \cite{fuchs-schweigert,bfrs} we know that $\mathcal{P}$ is equivalent to the category $\Mod[r]{A_{\mathcal{P}}} (\mathcal{C}^2)$ of modules over an algebra object $A_{\mathcal{P}} \in \mathcal{C}^2$ (Theorem \ref{thm:egno-mod-alg}). Expressed in terms of representatives $\{X_i\}_{i\in\mathcal{I}}$ of the isomorphism classes of simple objects in $\mathcal{C}$ (where $\mathcal{I}$ is some index set), $A_{\mathcal{P}}$ takes the form
\begin{gather*}
A_{\mathcal{P}} = \bigoplus\limits_{i\in\mathcal{I}} X_i^\ast \boxtimes X_i \, \, .
\end{gather*}
(For the basic background on algebra and module objects in monoidal categories we refer to Appendix \ref{appendix:cats-of-modules-algebras}.) Furthermore, one has $\mathcal{C} \simeq \Mod[r]{\mathds{1}_{\mathcal{C}}} (\mathcal{C})$ as $\mathcal{C}$-bimodule categories. Then by \cite[Proposition 3.7]{dsps} there are equivalences
\begin{gather*}
\mathcal{P}\boxtimes\mathcal{C} \simeq \Mod[r]{(A_{\mathcal{P}} \boxtimes \mathds{1}_{\mathcal{C}})} (\mathcal{C}^3) \, \, , \, \, \mathcal{C}\boxtimes\mathcal{P} \simeq \Mod[r]{(\mathds{1}_{\mathcal{C}} \boxtimes A_{\mathcal{P}})} (\mathcal{C}^3)
\end{gather*}
of $\mathcal{C}^3$-bimodule categories. From now on, let us write
\begin{gather*}
A_1 \defeq  A_{\mathcal{P}} \boxtimes \mathds{1}_{\mathcal{C}} \, \, , \, \, A_2 \defeq  \mathds{1}_{\mathcal{C}} \boxtimes A_{\mathcal{P}}
\end{gather*}
for short. The algebra structures of these algebras will be discussed in the next subsection. Realizing the relative Deligne product via bimodules (Proposition \ref{appendix-prop-relative-deligne}) one obtains
\begin{gather*}
\mathcal{P}_1 \boxtimes_3 \mathcal{P}_2 \simeq A_1\mh\Bimod\mh A_2 (\mathcal{C}^3) \, \, , 
\end{gather*}
and therefore
\begin{gather*}
(\mathcal{P}_1 \boxtimes_3 \mathcal{P}_2) \boxtimes_3 \mathcal{P}_1 \simeq A_1\mh\Bimod\mh (A_2 \otimes A_1) (\mathcal{C}^3) \simeq \Mod[r]{(A_1^{\op} \otimes A_2 \otimes A_1)} (\mathcal{C}^3)
\end{gather*}
(see Proposition \ref{def:op-alg} and Proposition \ref{prop:alg-structure-on-tensor-product} for the definition of the opposite algebra and the tensor product of algebras). Similarly,
\begin{gather*}
(\mathcal{P}_2 \boxtimes_3 \mathcal{P}_1) \boxtimes_3 \mathcal{P}_2 \simeq \Mod[r]{(A_2^{\op} \otimes A_1 \otimes A_2)} (\mathcal{C}^3) \, \, .
\end{gather*}
To show the desired equivalence of these two bimodule categories we need to prove that $A_1^{\op} \otimes A_2 \otimes A_1$ and $A_2^{\op} \otimes A_1 \otimes A_2$ are Morita equivalent algebras (Definition \ref{def:morita-eq}) in $\mathcal{C}^3$. In fact, we will show that these two algebras are even isomorphic.

\subsection{Algebra structures}
\label{sec:algebra-structures}

To accomplish this goal it is beneficial to employ the graphical calculus of string diagrams, using the conventions of Section \ref{sec:prelim} and Appendix \ref{appendix:algebras-graphical-notation}. In addition we graphically separate the factors of a Deligne product of objects by thickened vertical lines. Using that the Deligne product is associative up to equivalence (and hence associative up to isomorphism on the level of algebras), we can forget about bracketings of multiple factors. For morphisms between objects in $\mathcal{C}^n$, $n\in\mathds{N}$, whose factors are graphically separated by thickened vertical lines, these lines just denote the tensor product of $\mathds{k}$-vector spaces (for $\mathds{k}$-linear $\mathcal{C}$).\medskip

First we specify the algebra structures of all appearing algebra objects. For the rest of this section, let $\{X_i\}_{i\in\mathcal{I}}$ be a collection of representatives of the isomorphism classes of simple objects in $\mathcal{C}$ with index set $\mathcal{I}$. To shorten the notation we will write $i,j,k,\ldots$ instead of $X_i,X_j,X_k,\ldots$. We choose a basis $\{ \alpha \}$ of the space $\Hom(i\otimes j,k)$ as in Equation \eqref{eq:basis-alpha}, where the different basis elements are notationally identified with the labels $\alpha=1,\ldots,N_{ij}^k$ for $N_{ij}^k = \dim(\Hom(i\otimes j,k)) \in \mathds{N}$. A basis of $\Hom(k,i \otimes j)$, which is dual to $\{ \alpha \}$, is given by $\{ \hat{\alpha} \}$ as defined in Equation \eqref{eq:dual-kappa}. For details and notation regarding the different incarnations of dual bases appearing in this section (which are all equivalent according to Lemma \ref{lem:duals-are-equivalent}) we refer to Appendix \ref{appendix:algebras-graphical-notation}. Primes $'$ and tildes $\tilde{\,}$ distinguish different simple objects and basis vectors, whereas stars $^\ast$, hats $\hat{\,}$ and checks $\check{\,}$ refer to the different notions of duality.

The algebra $A_{\mathcal{P}} \in \mathcal{C}^2$ introduced above is
\begin{gather}
A_{\mathcal{P}} = \bigoplus\limits_{i\in\mathcal{I}} X_i^\ast \boxtimes X_i \equiv \; \, \bigoplus\limits_{i\in\mathcal{I}} \; \, \, \begin{tikzpicture}[baseline=(current bounding box.center)]
\draw[-<-=.5] (-0.5,-0.5) node[right] {$\mathsmaller{i}\,$} to (-0.5,0.5);
\draw[line width=0.5mm] (0,-0.6) -- (0,0.6);
\draw[->-=.5] (0.5,-0.5) node[right] {$\mathsmaller{i}\,$} to (0.5,0.5);
\end{tikzpicture}
\label{eq:def-A_P}
\end{gather}
with multiplication $m_{A_{\mathcal{P}}}\colon A_{\mathcal{P}} \otimes A_{\mathcal{P}} \to A_{\mathcal{P}}$ given by \cite{bfrs}
\begin{gather*}
m_{A_{\mathcal{P}}} \, = \; \, \bigoplus\limits_{i,j,k\in\mathcal{I}} \sum\limits_{\alpha=1}^{N_{ij}^k} \; \, \, \begin{tikzpicture}[baseline=(current bounding box.center)] 
\draw[->-=.9,-<-=.1] (2,1.5) node[below left] (i2) {$\mathsmaller{i}$} .. controls +(0,-1) and +(0,-1) .. (2.8,1.5) node[below right] (j2) {$\mathsmaller{j}$};
\draw[-<-=.5] (3.6,0.75) to (3.6,2) node[below left] (k2) {$\mathsmaller{k}$};
\draw (2.4,0.75) .. controls +(0,-1.2) and +(0,-1.2) .. (3.6,0.75);
\draw[black,fill=white] (2.4,0.75) circle (0.285cm) node (hatalpha) {$\hat{\alpha}$};
\draw[-<-=.5] (0.4,-0.2) node[above left] (i3) {$\mathsmaller{i}$} to (0.4,0.75);
\draw (0.4,0.75) .. controls +(0,1) and +(0,1) .. (2,1.5);
\draw[-<-=.5] (1.2,-0.2) node[above left] (j3) {$\mathsmaller{j}$} to (1.2,0.75);
\draw[line width=1mm,white] (1.2,0.75) .. controls +(0,1) and +(0,1) .. (2.8,1.5);
\draw (1.2,0.75) .. controls +(0,1) and +(0,1) .. (2.8,1.5);
\draw[line width=0.5mm] (4.2,-0.3) -- (4.2,2.1);
\draw[-<-=.8,->-=.2] (5,-0.2) node[above left] (i) {$\mathsmaller{i}$} .. controls +(0,1.2) and +(0,1.2) .. (5.8,-0.2) node[above right] (j) {$\mathsmaller{j}$};
\draw[->-=.7] (5.4,0.75) to (5.4,2) node[below right] (k) {$\mathsmaller{k}$}; 
\draw[black,fill=white] (5.4,0.75) circle (0.285cm) node (alpha) {$\alpha$};
\end{tikzpicture} \; \, \, \equiv \\
\equiv \; \, \bigoplus\limits_{i,j,k\in\mathcal{I}} \sum\limits_{\alpha=1}^{N_{ij}^k} \; \, \,
\begin{tikzpicture}[baseline=(current bounding box.center)] 
\draw[->-=.9,-<-=.1] (0,0) node[above left] (j2) {$\mathsmaller{j}$} .. controls +(0,1) and +(0,1) .. (0.8,0) node[above right] (i2) {$\mathsmaller{i}$};
\draw[-<-=.6] (0.4,0.75) to (0.4,1.5) node[below right] (k2) {$\mathsmaller{k}$}; 
\draw[black,fill=white] (0.4,0.75) circle (0.285cm) node (hatalphadual) {$\hat{\alpha}^\ast$};
\draw (0,-0.8) .. controls +(0,0.25) and +(-0.15,0) .. 
       (0.4,-0.4) .. controls +(0.15,0) and +(0,-0.25) .. (0.8,0);
\draw[line width=1mm,white] (0.8,-0.8) .. controls +(0,0.25) and +(0.15,0) ..
       (0.4,-0.4) .. controls +(-0.15,0) and +(0,-0.25) .. (0,0);
\draw (0.8,-0.8) .. controls +(0,0.25) and +(0.15,0) ..
       (0.4,-0.4) .. controls +(-0.15,0) and +(0,-0.25) .. (0,0);
\draw[line width=0.5mm] (1.6,-0.9) -- (1.6,1.6);
\draw[-<-=.8,->-=.2] (2.4,-0.8) node[above left] (i) {$\mathsmaller{i}$} .. controls +(0,1.5) and +(0,1.5) .. (3.2,-0.8) node[above right] (j) {$\mathsmaller{j}$};
\draw[->-=.7] (2.8,0.65) to (2.8,1.5) node[below right] (k) {$\mathsmaller{k}$}; 
\draw[black,fill=white] (2.8,0.35) circle (0.285cm) node (alpha) {$\alpha$};
\end{tikzpicture} \; \, \, .
\end{gather*}
Here, $\hat{\alpha}^\ast$ is in fact $(\hat{\alpha})^\ast$, the categorical dual morphism (Equation \eqref{eq:cat-dual-morphism-left}) of $\hat{\alpha}$, but by Lemma \ref{lem:cat-dual-kappa-dual} this is equal to $\widehat{\alpha^\ast}$, the dual basis vector of the categorical dual morphism $\alpha^\ast$ of $\alpha$. One can show that $m_{A_{\mathcal{P}}}\colon A_{\mathcal{P}} \otimes A_{\mathcal{P}} \to A_{\mathcal{P}}$ is independent of the choice of basis. The unit morphism $e_{A_{\mathcal{P}}}\colon \mathds{1}_{\mathcal{C}^2} \to A_{\mathcal{P}}$ is the embedding morphism of $\mathds{1}_{\mathcal{C}} \boxtimes \mathds{1}_{\mathcal{C}}$ as a subobject of $A_{\mathcal{P}}$. By the results of \cite{bfrs}, $A_{\mathcal{P}}$ indeed is an algebra object in $\mathcal{C}^2$. Consequently, $A_1 \defeq  A_{\mathcal{P}} \boxtimes \mathds{1}_{\mathcal{C}}$ and $A_2 \defeq  \mathds{1}_{\mathcal{C}} \boxtimes A_{\mathcal{P}}$ are algebra objects in $\mathcal{C}^3$, with obvious algebra structures: for example, the multiplication morphism for $A_1$ is given by $m_{A_{\mathcal{P}}} \otimes_{\mathds{k}} \id_{\mathds{1}_{\mathcal{C}}}$ (here we implicitly use the imposed strictness for $\mathds{1}_{\mathcal{C}} \otimes \mathds{1}_{\mathcal{C}} = \mathds{1}_{\mathcal{C}}$), and the unit morphism is the embedding morphism of $(\mathds{1}_{\mathcal{C}} \boxtimes \mathds{1}_{\mathcal{C}}) \boxtimes \mathds{1}_{\mathcal{C}}$ as a subobject of $A_1$.\medskip

Now we are ready to specify the algebra structures of $A \defeq  A_1^{\op} \otimes A_2 \otimes A_1$ and $B \defeq  A_2^{\op} \otimes A_1 \otimes A_2$ in $\mathcal{C}^3$ (using Proposition \ref{def:op-alg} and Proposition \ref{prop:alg-structure-on-tensor-product}). It is evident from Equation \eqref{eq:def-A_P} that
\begin{gather}
A = \bigoplus\limits_{i,j,k\in\mathcal{I}} \; \, \, 
 \; \, \, ,
\label{eq:def-A}
\end{gather}
where by strictness we can omit the dashed lines representing $\mathds{1}_{\mathcal{C}}$. Similarly, one finds
\begin{gather}
B = \bigoplus\limits_{i,j,k\in\mathcal{I}} \; \, \,
%
}$ for the multiplication of an algebra, the induced multiplication (via Proposition \ref{prop:alg-structure-on-tensor-product}) for $A$ is given by
\begin{gather*}
m_A = \, \,
%
\end{gather*}
is the multiplication of the opposite algebra (Proposition \ref{def:op-alg}). Written out in string diagrams,
\begin{gather*}
m_A \, = \; \, \bigoplus \sum \; \, \,
%
 \; \, \, ,
\end{gather*}
with the sums running over all appearing labels. From now on the symbols $\bigoplus$ and $\sum$ imply the sum over all line labels (simple objects) and vertex labels (basis elements), respectively, appearing in a diagram. Analogously, the multiplication morphism for $B$ is
\begin{gather*}
m_B = \, \,
%
 \; \, \, .
\end{gather*}

\subsection{The isomorphism}

In this subsection we show that $A$ and $B$ are isomorphic as algebra objects in $\mathcal{C}^3$.

We define the morphism $f\colon A\to B$ by
\begin{gather}
f \defeq  \, \bigoplus \sum \; \, \, 
%
 \; \, \, ,
\label{eq:def-f}
\end{gather}
which is basis independent\footnote{Expressions of the form $\sum \alpha \otimes_{\mathds{k}} \hat{\alpha}$, where $\{ \hat{\alpha} \}$ is dual to $\{ \alpha \}$ with respect to some non-degenerate pairing, are basis independent. This can be shown directly by mimicking the proof of the statement from basic linear algebra: for a $\mathds{k}$-vector space $V$, a basis $\{ \alpha_i \}_i$ of $V$ and its dual basis $\{ \alpha^i \}_i$ of $V^\ast \defeq  \Hom_{\mathds{k}} (V,\mathds{k})$, the vector $\sum_i \alpha_i \otimes_{\mathds{k}} \alpha^i \in V \otimes_{\mathds{k}} V^\ast$ is basis independent.} since the basis $\{ \alpha \}$ is dual to $\{ \check{\alpha} \}$ (and $\{ \beta \}$ is dual to $\{ \check{\beta} \}$) with respect to the pairing $\vartheta$ defined in Equation \eqref{eq:def-pairing-theta}.

To subdivide the problem, we set
\begin{gather}
f_1 \defeq  \, \bigoplus \sum \; \, \,
\begin{tikzpicture}[baseline=(current bounding box.center)]
\draw[->-=.9,-<-=.1] (0,0) node[below] {$\mathsmaller{i}$} .. controls +(0,1) and +(0,1) .. (0.8,0) node[below] {$\mathsmaller{k}$};
\draw[-<-=.6] (0.4,0.75) to (0.4,1.5) node[below right] {$\mathsmaller{\tilde{k}}$}; 
\draw[black,fill=white] (0.4,0.75) circle (0.285cm) node {$\alpha$};
\draw[line width=0.5mm] (1.6,-0.4) -- (1.6,1.6);
\draw[->-=.5] (2.4,0) node[below] {$\mathsmaller{i}$} to (2.4,1.5);
\draw[-<-=.5] (3.2,0) node[below] {$\mathsmaller{j}$} to (3.2,1.5);
\draw[-<-=.8,-<-=.1] (4,1.5) node[right] {$\mathsmaller{\tilde{k}}$} .. controls +(0,-1) and +(0,-1) .. (4.8,1.5) node[right] {$\mathsmaller{i}$};
\draw[-<-=.5] (4.4,0.75) to (4.4,0) node[below] (k2) {$\mathsmaller{k}$}; 
\draw[black,fill=white] (4.4,0.75) circle (0.285cm) node {$\check{\alpha}$};
\draw[line width=0.5mm] (5.6,-0.4) -- (5.6,1.6);
\draw[->-=.5] (6.4,0) node[below] {$\mathsmaller{j}$} to (6.4,1.5);
\end{tikzpicture}
\label{eq:def-f1}
\end{gather}
and
\begin{gather}
f_2 \defeq  \, \bigoplus \sum \; \, \,
\begin{tikzpicture}[baseline=(current bounding box.center)]
\draw[-<-=.5] (0.8,0) node[below] {$\mathsmaller{\tilde{k}}$} to (0.8,1.5);
\draw[line width=0.5mm] (1.6,-0.4) -- (1.6,1.6);
\draw[->-=.9,->-=.2] (2.4,0) node[below] {$\mathsmaller{i}$} .. controls +(0,1) and +(0,1) .. (3.2,0) node[below] {$\mathsmaller{j}$};
\draw[-<-=.6] (2.8,0.75) to (2.8,1.5) node[below right] {$\mathsmaller{\tilde{j}}$}; 
\draw[black,fill=white] (2.8,0.75) circle (0.285cm) node {$\beta$};
\draw[->-=.5] (4,0) node[below] {$\mathsmaller{\tilde{k}}$} to (4,1.5);
\draw[-<-=.5] (4.8,0) node[below] {$\mathsmaller{i}$} to (4.8,1.5);
\draw[line width=0.5mm] (5.6,-0.4) -- (5.6,1.6);
\draw[->-=.9,-<-=.1] (6.4,1.5) node[right] {$\mathsmaller{\tilde{j}}$} .. controls +(0,-1) and +(0,-1) .. (7.2,1.5) node[right] {$\mathsmaller{i}$};
\draw[-<-=.5] (6.8,0.75) to (6.8,0) node[below] (k2) {$\mathsmaller{j}$}; 
\draw[black,fill=white] (6.8,0.75) circle (0.285cm) node {$\check{\beta}$};
\end{tikzpicture} \; \, \, .
\label{eq:def-f2}
\end{gather}
This yields a decomposition of the morphism $f$ into $f_1 \colon A \to C$ and $f_2 \colon C \to A$, such that $f=f_2 \circ f_1$, for
\begin{gather}
C \defeq  \bigoplus\limits_{i,j,\tilde{k}\in\mathcal{I}} \; \, \,
\begin{tikzpicture}[baseline=(current bounding box.center)]
\draw[-<-=.5] (-2.5,-0.5) node[right] {$\mathsmaller{\tilde{k}}$} to (-2.5,0.5);
\draw[line width=0.5mm] (-2,-0.6) -- (-2,0.6);
\draw[->-=.5] (-1.5,-0.5) node[right] {$\mathsmaller{i}$} to (-1.5,0.5);
\draw[-<-=.5] (-1,-0.5) node[right] {$\mathsmaller{j}$} to (-1,0.5);
\draw[->-=.5] (-0.5,-0.5) node[right] {$\mathsmaller{\tilde{k}}$} to (-0.5,0.5);
\draw[-<-=.5] (0,-0.5) node[right] {$\mathsmaller{i}$} to (0,0.5);
\draw[line width=0.5mm] (0.5,-0.6) -- (0.5,0.6);
\draw[->-=.5] (1,-0.5) node[right] {$\mathsmaller{j}$} to (1,0.5);
\end{tikzpicture} \; \, \, .
\label{eq:def-C}
\end{gather}
Now we use the following strategy to prove that $f\colon A\to B$ is an algebra isomorphism:
\begin{itemize}
\item Check that $f$ is an ordinary isomorphism between the objects $A,B \in \mathcal{C}^3$.
\item Equip $C$ with an algebra structure.
\item Show that $f_1 \colon A \to C$ and $f_2 \colon C \to B$ are algebra homomorphisms (Definition \ref{def:algebra-homomorphism}). Then also $f=f_2 \circ f_1$ is an algebra homomorphism, and thus an isomorphism $A \cong B$ of algebra objects is established by the first step.
\end{itemize}

Indeed, using Equation \eqref{eq:dual-kappa} and Equation \eqref{eq:appendix-completeness-relation}, it is immediate that $f$ is an isomorphism with inverse
\begin{gather*}
f^{-1} = \, \bigoplus \sum \; \, \, \begin{tikzpicture}[baseline=(current bounding box.center)]
\draw[-<-=.8,->-=.2] (0,1.5) node[right] {$\mathsmaller{i}$} .. controls +(0,-1) and +(0,-1) .. (0.8,1.5) node[right] {$\mathsmaller{k}$};
\draw[->-=.8] (0.4,0.75) to (0.4,0) node[below] (k2) {$\mathsmaller{\tilde{k}}$}; 
\draw[black,fill=white] (0.4,0.75) circle (0.285cm) node {$\hat{\alpha}$};
\draw[line width=0.5mm] (1.6,-0.4) -- (1.6,1.6);
\draw[-<-=.8,-<-=.1] (2.4,1.5) node[right] {$\mathsmaller{i}$} .. controls +(0,-1) and +(0,-1) .. (3.2,1.5) node[right] {$\mathsmaller{j}$};
\draw[->-=.8] (2.8,0.75) to (2.8,0) node[below] (k2) {$\mathsmaller{\tilde{j}}$}; 
\draw[black,fill=white] (2.8,0.75) circle (0.285cm) node {$\hat{\beta}$};
\draw[->-=.9,->-=.2] (4,0) node[below] {$\mathsmaller{\tilde{k}}$} .. controls +(0,1) and +(0,1) .. (4.8,0) node[below] {$\mathsmaller{i}$};
\draw[->-=.8] (4.4,0.75) to (4.4,1.5) node[below right] {$\mathsmaller{k}$}; 
\draw[black,fill=white] (4.4,0.75) circle (0.285cm) node {$\widehat{\check{\alpha}}$};
\draw[line width=0.5mm] (5.6,-0.4) -- (5.6,1.6);
\draw[-<-=.8,->-=.2] (6.4,0) node[below] {$\mathsmaller{\tilde{j}}$} .. controls +(0,1) and +(0,1) .. (7.2,0) node[below] {$\mathsmaller{i}$};
\draw[->-=.8] (6.8,0.75) to (6.8,1.5) node[right] {$\mathsmaller{j}$}; 
\draw[black,fill=white] (6.8,0.75) circle (0.31cm) node {$\widehat{\check{\beta}}$};
\end{tikzpicture} \; \, \, .
\end{gather*}

For the second point of the above list, we equip $C$ with a ``reasonable'' algebra structure. The object $C$, decomposed into simple objects as in Equation \eqref{eq:def-C}, consists of invariant factors: the objects $\tilde{k}^\ast$, $\tilde{k}$ and $i^\ast$ are the same as in $B$ (Equation \eqref{eq:def-B}) in the sense that $f_2$ acts as identity on these objects (where we view $C$ as the source object of the morphism $f_2 \colon C\to B$). Similarly, $i$, $j^\ast$ and $j$ are the same as in $A$ (Equation \eqref{eq:def-A}) in the sense that $f_1$ acted as identity on these objects (where we view $C$ as the target object of the morphism $f_1 \colon A\to C$). Hence a ``reasonable'' algebra structure for $C$ should respect the algebra structures of $A$ and $B$ in the decomposition of $C$ into simple objects. This heuristic observation motivates to define a morphism $m_C \colon C\otimes C \to C$ as
\begin{gather*}
m_C \,=\bigoplus \sum \; \, \, 
\begin{tikzpicture}[baseline=(current bounding box.center)]
\draw[-<-=.1] (1.6,0) node[below] {$\mathsmaller{\tilde{k}}$} to (1.6,0.8);
\draw[-<-=.1] (2.4,0) node[below] {$\mathsmaller{\tilde{k}'}$} to (2.4,0.8);
\draw (1.6,0.8) .. controls +(0,0.25) and +(-0.15,0) .. 
       (2,1.2) .. controls +(0.15,0) and +(0,-0.25) .. (2.4,1.6);
\draw[line width=1mm,white] (2.4,0.8) .. controls +(0,0.25) and +(0.15,0) .. 
       (2,1.2) .. controls +(-0.15,0) and +(0,-0.25) .. (1.6,1.6);
\draw (2.4,0.8) .. controls +(0,0.25) and +(0.15,0) .. 
       (2,1.2) .. controls +(-0.15,0) and +(0,-0.25) .. (1.6,1.6);
\draw[->-=.9,-<-=.1] (1.6,1.6) .. controls +(0,1) and +(0,1) .. (2.4,1.6);
\draw[-<-=.6] (2,2.35) to (2,4.7) node[below right] {$\mathsmaller{\tilde{k}''}$}; 
\draw[black,fill=white] (2,2.35) circle (0.285cm) node {$\hat{\beta}^{\ast}$};
\draw[line width=0.5mm] (3.2,-0.4) -- (3.2,4.8);
\draw[->-=.1] (4,0) node[below] {$\mathsmaller{i}$} to (4,2.4);
\draw[-<-=.1] (4.8,0) node[below] {$\mathsmaller{j}$} to (4.8,1.6);
\draw[->-=.2] (5.6,0) node[below] {$\mathsmaller{\tilde{k}}$} to (5.6,0.8);
\draw[-<-=.2] (8,0) node[below] {$\mathsmaller{j'}$} to (8,0.8);
\draw[->-=.1] (8.8,0) node[below] {$\mathsmaller{\tilde{k}'}$} to (8.8,1.6);
\draw[-<-=.2] (9.6,0) node[below] {$\mathsmaller{i'}$} to (9.6,2.4);
\draw[->-=.2] (7.2,0) node[below] {$\mathsmaller{i'}$} .. controls +(0,0.25) and +(0.15,0) ..
       (6.8,0.4) .. controls +(-0.15,0) and +(0,-0.25) .. (6.4,0.8);
\draw[line width=1mm,white] (6.4,0) .. controls +(0,0.25) and +(-0.15,0) .. 
       (6.8,0.4) .. controls +(0.15,0) and +(0,-0.25) .. (7.2,0.8);
\draw[-<-=.1] (6.4,0) node[below] {$\mathsmaller{i}$} .. controls +(0,0.25) and +(-0.15,0) .. 
       (6.8,0.4) .. controls +(0.15,0) and +(0,-0.25) .. (7.2,0.8);
\draw (6.4,0.8) .. controls +(0,0.25) and +(0.15,0) ..
       (6,1.2) .. controls +(-0.15,0) and +(0,-0.25) .. (5.6,1.6);
\draw[line width=1mm,white] (5.6,0.8) .. controls +(0,0.25) and +(-0.15,0) .. 
       (6,1.2) .. controls +(0.15,0) and +(0,-0.25) .. (6.4,1.6);
\draw (5.6,0.8) .. controls +(0,0.25) and +(-0.15,0) .. 
       (6,1.2) .. controls +(0.15,0) and +(0,-0.25) .. (6.4,1.6);
\draw (8,0.8) .. controls +(0,0.25) and +(0.15,0) ..
       (7.6,1.2) .. controls +(-0.15,0) and +(0,-0.25) .. (7.2,1.6);
\draw[line width=1mm,white] (7.2,0.8) .. controls +(0,0.25) and +(-0.15,0) .. 
       (7.6,1.2) .. controls +(0.15,0) and +(0,-0.25) .. (8,1.6);
\draw (7.2,0.8) .. controls +(0,0.25) and +(-0.15,0) .. 
       (7.6,1.2) .. controls +(0.15,0) and +(0,-0.25) .. (8,1.6);
\draw (8.8,1.6) .. controls +(0,0.25) and +(0.15,0) ..
       (8.4,2) .. controls +(-0.15,0) and +(0,-0.25) .. (8,2.4);
\draw[line width=1mm,white] (8,1.6) .. controls +(0,0.25) and +(-0.15,0) .. 
       (8.4,2) .. controls +(0.15,0) and +(0,-0.25) .. (8.8,2.4);
\draw (8,1.6) .. controls +(0,0.25) and +(-0.15,0) .. 
       (8.4,2) .. controls +(0.15,0) and +(0,-0.25) .. (8.8,2.4);
\draw (5.6,1.6) .. controls +(0,0.25) and +(0.15,0) ..
       (5.2,2) .. controls +(-0.15,0) and +(0,-0.25) .. (4.8,2.4);
\draw[line width=1mm,white] (4.8,1.6) .. controls +(0,0.25) and +(-0.15,0) .. 
       (5.2,2) .. controls +(0.15,0) and +(0,-0.25) .. (5.6,2.4);
\draw (4.8,1.6) .. controls +(0,0.25) and +(-0.15,0) .. 
       (5.2,2) .. controls +(0.15,0) and +(0,-0.25) .. (5.6,2.4);
\draw (7.2,1.6) .. controls +(0,0.25) and +(0.15,0) ..
       (6.8,2) .. controls +(-0.15,0) and +(0,-0.25) .. (6.4,2.4);
\draw[line width=1mm,white] (6.4,1.6) .. controls +(0,0.25) and +(-0.15,0) .. 
       (6.8,2) .. controls +(0.15,0) and +(0,-0.25) .. (7.2,2.4);
\draw (6.4,1.6) .. controls +(0,0.25) and +(-0.15,0) .. 
       (6.8,2) .. controls +(0.15,0) and +(0,-0.25) .. (7.2,2.4);
\draw (4.8,2.4) .. controls +(0,0.25) and +(0.15,0) ..
       (4.4,2.8) .. controls +(-0.15,0) and +(0,-0.25) .. (4,3.2);
\draw[line width=1mm,white] (4,2.4) .. controls +(0,0.25) and +(-0.15,0) .. 
       (4.4,2.8) .. controls +(0.15,0) and +(0,-0.25) .. (4.8,3.2);
\draw (4,2.4) .. controls +(0,0.25) and +(-0.15,0) .. 
       (4.4,2.8) .. controls +(0.15,0) and +(0,-0.25) .. (4.8,3.2);
\draw (5.6,2.4) .. controls +(0,0.25) and +(-0.15,0) .. 
       (6,2.8) .. controls +(0.15,0) and +(0,-0.25) .. (6.4,3.2);
\draw[line width=1mm,white] (6.4,2.4) .. controls +(0,0.25) and +(0.15,0) ..
       (6,2.8) .. controls +(-0.15,0) and +(0,-0.25) .. (5.6,3.2);
\draw (6.4,2.4) .. controls +(0,0.25) and +(0.15,0) ..
       (6,2.8) .. controls +(-0.15,0) and +(0,-0.25) .. (5.6,3.2);
\draw[-<-=.9,->-=.1] (4,3.2) .. controls +(0,1) and +(0,1) .. (4.8,3.2);
\draw[->-=.7] (4.4,3.95) to (4.4,4.7) node[below right] {$\mathsmaller{i''}$}; 
\draw[black,fill=white] (4.4,3.95) circle (0.285cm) node {$\gamma$};
\draw[->-=.9,-<-=.1] (5.6,3.2) .. controls +(0,1) and +(0,1) .. (6.4,3.2);
\draw[-<-=.7] (6,3.95) to (6,4.7) node[below right] {$\mathsmaller{j''}$}; 
\draw[black,fill=white] (6,3.95) circle (0.285cm) node {$\hat{\delta}^\ast$};
\draw[-<-=.9,->-=.1] (7.2,2.4) .. controls +(0,1) and +(0,1) .. (8,2.4);
\draw[->-=.6] (7.6,3.15) to (7.6,4.7) node[below right] {$\mathsmaller{\tilde{k}''}$};
\draw[black,fill=white] (7.6,3.15) circle (0.285cm) node {$\beta$};
\draw[->-=.9,-<-=.1] (8.8,2.4) .. controls +(0,1) and +(0,1) .. (9.6,2.4);
\draw[-<-=.6] (9.2,3.15) to (9.2,4.7) node[below right] {$\mathsmaller{i''}$};
\draw[black,fill=white] (9.2,3.15) circle (0.285cm) node {$\hat{\gamma}^\ast$};
\draw[line width=0.5mm] (10.4,-0.4) -- (10.4,4.8);
\draw[->-=.6] (11.2,0) node[below] {$\mathsmaller{j}$} to (11.2,1.5);
\draw[->-=.6] (12,0) node[below] {$\mathsmaller{j'}$} to (12,1.5);
\draw (11.2,1.5) .. controls +(0,1) and +(0,1) .. (12,1.5);
\draw[->-=.6] (11.6,2.25) to (11.6,4.7) node[below right] {$\mathsmaller{j''}$}; 
\draw[black,fill=white] (11.6,2.25) circle (0.285cm) node {$\delta$};
\end{tikzpicture} \; \, \, ,
\end{gather*}
and a morphism $e_C$ as the embedding morphism of $(\mathds{1}_{\mathcal{C}} \boxtimes \mathds{1}_{\mathcal{C}}) \boxtimes \mathds{1}_{\mathcal{C}}$ as a subobject of $C$. It is a direct calculation to check that $(C,m_C,e_C)$ is an algebra object in $\mathcal{C}^3$; loosely speaking, as $m_C$ contains parts of $m_A$ and $m_B$, this follows from the fact that $A$ and $B$ are algebra objects.\medskip

We are now ready to prove that $f_1 \colon A \to C$ is an algebra homomorphism (Definition \ref{def:algebra-homomorphism}). As the unit morphisms $e_A$ and $e_C$ of the algebras $A$ and $C$, respectively, are just embedding morphisms of units, it evidently follows that $f_1 \circ e_A = e_C$. Thus it remains to verify that
\begin{gather}

\label{eq:expr2}
\end{gather}
are equal.\medskip

First observe that we can ignore all lines labeled $j$, $j'$ and $j''$, since they are the same in both of the above diagrams \eqref{eq:expr1} and \eqref{eq:expr2}: the third component is identical in both diagrams, and in the second component the braidings of $j$ and $j'$ with the other objects coincide in both diagrams (after sliding $\check{\alpha}$ past $i'$ and $j'^\ast$ in diagram \eqref{eq:expr1}). Furthermore, it turns out to be useful to change the bases
\begin{gather*}
\left\{\raisebox{0.1cm}{%
\right\}
\end{gather*}
of $\Hom(\tilde{k} \otimes \tilde{k}',\tilde{k}'')$ (in diagram \eqref{eq:expr1}) and $\Hom(k \otimes k',k'')$ (in diagram \eqref{eq:expr2}), respectively. Note that here and from now on we abuse notation and use $\beta$ and $\chi$ as labels for a basis of $\Hom(\tilde{k}' \otimes \tilde{k},\tilde{k}'')$ and $\Hom(k' \otimes k, k'')$, respectively. Applying these changes, one is left with showing that
\begin{subequations}
\begin{gather}
\label{eq:showthis-leftside}
\bigoplus \sum \; \, \, %
 \; \, \, .
\end{gather}
\end{subequations}
The idea for the proof of this equation is to apply basis changes on certain parts of the diagram \eqref{eq:showthis-leftside}, utilizing the fact that expressions of the form $\sum \alpha \otimes_{\mathds{k}} \hat{\alpha}$ are basis independent. \medskip

The morphisms
\begin{gather*}
\xi \; \equiv \; \, \, %
\end{gather*}
form a basis of $\Hom(i^\ast \otimes k^\ast \otimes i'^\ast \otimes k'^\ast,\tilde{k}''^\ast)$, since $\alpha$, $\alpha'$ and $\beta$ are basis vectors of their corresponding Hom spaces. (Recall we identify basis vectors with their integer label, as is indicated in the above definition of $\xi$. This means that $\alpha$, $\alpha'$ and $\beta$ are actually integer-valued labels for basis vectors, and hence $\xi$ is a label for a basis of $\Hom(i^\ast \otimes k^\ast \otimes i'^\ast \otimes k'^\ast,\tilde{k}''^\ast)$, depending on the labels $\alpha$, $\alpha'$ and $\beta$.) The morphism
\begin{gather*}
\hat{\xi} \; \equiv \; \, \, \begin{tikzpicture}[baseline=(current bounding box.center)]
\draw[-<-=.9,->-=.1] (0,0) node[above] {$\mathsmaller{i}$} .. controls +(0,-1.5) and +(0,-1.5) .. (2.4,0) node[above] {$\mathsmaller{k'}$};
\draw[->-=.8] (1.2,-1.2) to (1.2,-2.25) node[above right] {$\mathsmaller{\tilde{k}''}$}; 
\draw[->-=.25] (0.8,0) node[above] {$\mathsmaller{k}$} .. controls +(0,-0.5) and +(0,-0.5) .. (1.2,-1.2);
\draw[->-=.25] (1.6,0) node[above] {$\mathsmaller{i'}$} .. controls +(0,-0.5) and +(0,-0.5) .. (1.2,-1.2);
\draw[black,fill=white] (1.2,-1.2) circle (0.285cm) node {$\hat{\xi}$};
\end{tikzpicture} \; \, \, \defeq  \; \, \, \begin{tikzpicture}[baseline=(current bounding box.center)]
\draw[-<-=.9,->-=.1] (0,0) node[above] {$\mathsmaller{i}$} .. controls +(0,-1) and +(0,-1) .. (0.8,0) node[above] {$\mathsmaller{k}$};
\draw (0.4,-0.75) to (0.4,-1); 
\draw[-<-=.9,->-=.1] (1.6,0) node[above] {$\mathsmaller{i'}$} .. controls +(0,-1) and +(0,-1) .. (2.4,0) node[above] {$\mathsmaller{k'}$};
\draw (2,-0.75) to (2,-1); 
\draw[-<-=.8,->-=.2] (0.4,-1) node[left,yshift=-0.5cm] {$\mathsmaller{\tilde{k}}$}  .. controls +(0,-1.25) and +(0,-1.25) .. (2,-1) node[right,yshift=-0.5cm] {$\mathsmaller{\tilde{k}'}$};
\draw[-<-=.2] (1.2,-2.75) node[above right] {$\mathsmaller{\tilde{k}''}$} to (1.2,-1.9); 
\draw[black,fill=white] (1.2,-1.9) circle (0.285cm) node {$\beta^{\ast}$};
\draw[black,fill=white] (0.4,-0.75) circle (0.285cm) node {$\hat{\alpha}$};
\draw[black,fill=white] (2,-0.75) circle (0.285cm) node {$\hat{\alpha'}$};
\end{tikzpicture}
\end{gather*}
is dual to $\xi$ with respect to the pairing $\kappa$ from Equation \eqref{eq:def-pairing-kappa} as $\hat{\alpha}$, $\hat{\alpha'}$ and $\hat{\beta}^\ast$ are dual to $\alpha$, $\alpha'$ and $\beta^\ast$ with respect to $\kappa$, respectively (cf. Equation \eqref{eq:dual-kappa}, and Lemma \ref{lem:cat-dual-kappa-dual} that the categorical dual morphism $\beta^\ast$ of $\beta \in \Hom(\tilde{k}' \otimes \tilde{k}, \tilde{k}'')$ is $\kappa$-dual to $\hat{\beta}^\ast$):
\begin{gather*}
\begin{tikzpicture}[baseline=(current bounding box.center)]
\draw (0,0) node[left] {$\mathsmaller{i}$} .. controls +(0,1) and +(0,1) .. (0.8,0) node[right] {$\mathsmaller{k}$};
\draw (0.4,0.75) to (0.4,1); 
\draw (1.6,0) node[left] {$\mathsmaller{i'}$} .. controls +(0,1) and +(0,1) .. (2.4,0) node[right] {$\mathsmaller{k'}$};
\draw (2,0.75) to (2,1); 
\draw[->-=.9,-<-=.1] (0.4,1) node[left,yshift=0.5cm] {$\mathsmaller{\tilde{k}}$}  .. controls +(0,1.25) and +(0,1.25) .. (2,1) node[right,yshift=0.5cm] {$\mathsmaller{\tilde{k}'}$};
\draw[-<-=.6] (1.2,1.9) to (1.2,2.75) node[below right] {$\mathsmaller{\tilde{k}''}$}; 
\draw[black,fill=white] (1.2,1.9) circle (0.285cm) node {$\hat{\beta}^{\ast}$};
\draw[black,fill=white] (0.4,0.75) circle (0.285cm) node {$\alpha$};
\draw[black,fill=white] (2,0.75) circle (0.285cm) node {$\alpha'$};
\draw[-<-=.9,->-=.1] (0,0) .. controls +(0,-1) and +(0,-1) .. (0.8,0);
\draw (0.4,-0.75) to (0.4,-1); 
\draw[-<-=.9,->-=.1] (1.6,0) .. controls +(0,-1) and +(0,-1) .. (2.4,0);
\draw (2,-0.75) to (2,-1); 
\draw[-<-=.8,->-=.2] (0.4,-1) node[left,yshift=-0.5cm] {$\mathsmaller{\tilde{k}}$}  .. controls +(0,-1.25) and +(0,-1.25) .. (2,-1) node[right,yshift=-0.5cm] {$\mathsmaller{\tilde{k}'}$};
\draw[-<-=.2] (1.2,-2.75) node[above right] {$\mathsmaller{\tilde{k}''}$} to (1.2,-1.9); 
\draw[black,fill=white] (1.2,-1.9) circle (0.285cm) node {$\nu^{\ast}$};
\draw[black,fill=white] (0.4,-0.75) circle (0.285cm) node {$\hat{\lambda}$};
\draw[black,fill=white] (2,-0.75) circle (0.285cm) node {$\hat{\lambda'}$};
\end{tikzpicture} \; \, \, = \, \, \; \delta_{\alpha\lambda}\delta_{\alpha'\lambda'}\delta_{\beta\nu} \id_{\tilde{k}''^\ast}
\end{gather*}
Using the canonical isomorphisms of Hom spaces (Equation \eqref{eq:hom-isomorphisms}) we obtain
\begin{gather}
\begin{tikzpicture}[baseline=(current bounding box.center)]
\draw[-<-=.9,->-=.1] (0,0) node[above] {$\mathsmaller{i}$} .. controls +(0,-1.5) and +(0,-1.5) .. (2.4,0) node[above] {$\mathsmaller{k'}$};
\draw[->-=.8] (1.2,-1.2) to (1.2,-2.25) node[above right] {$\mathsmaller{\tilde{k}''}$}; 
\draw[->-=.25] (0.8,0) node[above] {$\mathsmaller{k}$} .. controls +(0,-0.5) and +(0,-0.5) .. (1.2,-1.2);
\draw[->-=.25] (1.6,0) node[above] {$\mathsmaller{i'}$} .. controls +(0,-0.5) and +(0,-0.5) .. (1.2,-1.2);
\draw[black,fill=white] (1.2,-1.2) circle (0.285cm) node {$\hat{\xi}$};
\end{tikzpicture} \; \, \, \mapsto \; \, \, \begin{tikzpicture}[baseline=(current bounding box.center)]
\draw[-<-=.9,->-=.1] (0,0) node[below] {$\mathsmaller{k'}$} .. controls +(0,1.5) and +(0,1.5) .. (2.4,0) node[below] {$\mathsmaller{i}$};
\draw[->-=.8] (1.2,1.2) to (1.2,2.25) node[below right] {$\mathsmaller{\tilde{k}''}$}; 
\draw[->-=.15] (0.8,0) node[below] {$\mathsmaller{i'}$} .. controls +(0,0.5) and +(0,0.5) .. (1.2,1.2);
\draw[->-=.15] (1.6,0) node[below] {$\mathsmaller{k}$} .. controls +(0,0.5) and +(0,0.5) .. (1.2,1.2);
\draw[black,fill=white] (1.2,1.2) circle (0.285cm) node {$\hat{\xi}^\ast$};
\end{tikzpicture} \; \, \, = \, \, \; \begin{tikzpicture}[baseline=(current bounding box.center)]
\draw[-<-=.9,->-=.1] (0,0) node[below] {$\mathsmaller{k'}$} .. controls +(0,1) and +(0,1) .. (0.8,0) node[below] {$\mathsmaller{i'}$};
\draw (0.4,0.75) to (0.4,1); 
\draw[-<-=.9,->-=.1] (1.6,0) node[below] {$\mathsmaller{k}$} .. controls +(0,1) and +(0,1) .. (2.4,0) node[below] {$\mathsmaller{i}$};
\draw (2,0.75) to (2,1); 
\draw[-<-=.8,->-=.2] (0.4,1) node[left,yshift=0.5cm] {$\mathsmaller{\tilde{k}'}$}  .. controls +(0,1.25) and +(0,1.25) .. (2,1) node[right,yshift=0.5cm] {$\mathsmaller{\tilde{k}}$};
\draw[->-=.8] (1.2,1.9) to (1.2,2.75) node[below right] {$\mathsmaller{\tilde{k}''}$}; 
\draw[black,fill=white] (1.2,1.9) circle (0.285cm) node {$\beta$};
\draw[black,fill=white] (0.4,0.75) circle (0.33cm) node {$\hat{\alpha'}^\ast$};
\draw[black,fill=white] (2,0.75) circle (0.285cm) node {$\hat{\alpha}^\ast$};
\end{tikzpicture} \; \, \, \mapsto \nonumber \\ \mapsto \; \, \, \begin{tikzpicture}[baseline=(current bounding box.center)]
\draw[-<-=.9,->-=.1] (0,0) node[below] {$\mathsmaller{k'}$} .. controls +(0,1.5) and +(0,0.8) .. (2,0.5);
\draw[->-=.8] (1.2,1.2) to (1.2,2.25) node[below right] {$\mathsmaller{\tilde{k}''}$}; 
\draw[->-=.15] (0.8,0) node[below] {$\mathsmaller{i'}$} .. controls +(0,0.5) and +(0,0.5) .. (1.2,1.2);
\draw[->-=.15] (1.6,0) node[below] {$\mathsmaller{k}$} .. controls +(0,0.5) and +(0,0.5) .. (1.2,1.2);
\draw[black,fill=white] (1.2,1.2) circle (0.285cm) node {$\hat{\xi}^\ast$};
\draw[-<-=.7] (2.8,1.2) to (2.8,2.25) node[below right] {$\mathsmaller{i}$};
\draw (2,0.5) .. controls +(0,-1) and +(0,-1) .. (2.8,1.2);
\end{tikzpicture} \; \, \, \equaltext{\eqref{eq:appendix-check-hatstar}} \; \, \, \begin{tikzpicture}[baseline=(current bounding box.center)]
\draw[-<-=.9,->-=.1] (0,0) node[below] {$\mathsmaller{k'}$} .. controls +(0,1) and +(0,1) .. (0.8,0) node[below] {$\mathsmaller{i'}$};
\draw (0.4,0.75) to (0.4,1); 
\draw (2,1) .. controls +(0,-0.5) and +(0,-1) .. (2.8,1.5);
\draw[->-=.2] (2.4,0) node[below] {$\mathsmaller{k}$} to (2.4,0.75);
\draw[-<-=.75] (2.8,1.5) to (2.8,2.75) node[below right] {$\mathsmaller{i}$}; 
\draw[-<-=.8,->-=.2] (0.4,1) node[left,yshift=0.5cm] {$\mathsmaller{\tilde{k}'}$}  .. controls +(0,1.25) and +(0,1.25) .. (2,1) node[right,yshift=0.5cm] {$\mathsmaller{\tilde{k}}$};
\draw[->-=.8] (1.2,1.9) to (1.2,2.75) node[below right] {$\mathsmaller{\tilde{k}''}$}; 
\draw[black,fill=white] (1.2,1.9) circle (0.285cm) node {$\beta$};
\draw[black,fill=white] (0.4,0.75) circle (0.33cm) node {$\hat{\alpha'}^\ast$};
\draw[black,fill=white] (2.4,0.75) circle (0.285cm) node {$\check{\alpha}$};
\end{tikzpicture} \; \, \, \defeqinv  \; \, \, \begin{tikzpicture}[baseline=(current bounding box.center)]
\draw[-<-=.9,->-=.1] (0.6,-2.25) node[below] {$\mathsmaller{k'}$} .. controls +(0,1.5) and +(0,1.5) .. (1.8,-2.25) node[below] {$\mathsmaller{k}$};
\draw[-<-=.8] (1.2,-1.2) to (1.2,-2.25) node[below] {$\mathsmaller{i'}$}; 
\draw[-<-=.15] (0.8,0) node[below left] {$\mathsmaller{\tilde{k}''}$} .. controls +(0,-0.5) and +(0,-0.5) .. (1.2,-1.2);
\draw[->-=.25] (1.6,0) node[below right] {$\mathsmaller{i}$} .. controls +(0,-0.5) and +(0,-0.5) .. (1.2,-1.2);
\draw[black,fill=white] (1.2,-1.15) circle (0.285cm) node {$\bar{\xi}$};
\end{tikzpicture} \; \, \, ,
\label{eq:xi-bar}
\end{gather}
so the morphism $\bar{\xi}$ is dual to $\xi$ with respect to the non-degenerate pairing $(\omega,\bar{\omega'})\mapsto(\omega,\hat{\omega'}^\ast)\mapsto\eta(\omega,\hat{\omega'}^\ast)$ for basis vectors $\omega\in\Hom(i^\ast \otimes k^\ast \otimes i'^\ast \otimes k'^\ast,\tilde{k}''^\ast)$ and $\bar{\omega'}\in\Hom(k' \otimes i' \otimes k,\tilde{k}'' \otimes i^\ast)$, where $\hat{\omega'}^\ast \in \Hom(k' \otimes i' \otimes k \otimes i,\tilde{k}'')$ is obtained from $\bar{\omega'}$ by composing with $\widetilde{\ev}_i$, and $\eta$ is the pairing given by Equation \eqref{eq:def-pairing-eta}. The plan now is to extract $\bar{\xi}$ from the second component of \eqref{eq:showthis-leftside} and then change the basis $\{ \xi \}$ of $\Hom(i^\ast \otimes k^\ast \otimes i'^\ast \otimes k'^\ast,\tilde{k}''^\ast)$ to the basis formed by the morphisms
\begin{gather*}
\zeta \; \defeq  \; \, \, 
 \; \, \, .
\end{gather*}
As we will see, this basis change leads to \eqref{eq:showthis-rightside}, thus completing the proof that $f_1$ is an algebra homomorphism. The dual $\hat{\zeta}$ of $\zeta$ (with respect to the pairing $\kappa$ from Equation \eqref{eq:def-pairing-kappa}) is
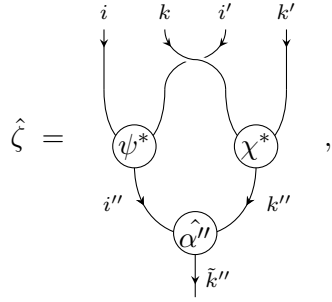
\begin{gather*}
\hat{\zeta} \; = \; \, \, %
 \; \, \, ,
\end{gather*}
and the dual $\bar{\zeta}$, which is $\hat{\zeta}^\ast$ with $i$ dualized by composing with $\widetilde{\coev}_i$, is given by
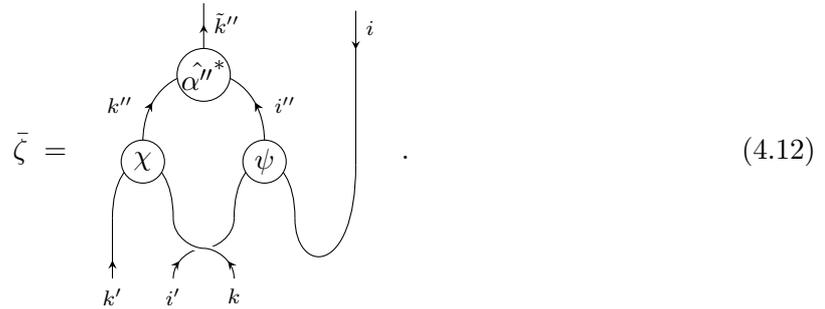
\begin{gather}
\bar{\zeta} \; = \; \, \, %
 \; \, \, .
\label{eq:zeta-bar}
\end{gather}
Note that $\bar{\zeta}$ is dual to $\zeta$ with respect to the same pairing as $\bar{\xi}$ is dual to $\xi$.

Let us start with the calculation. First we reformulate the second component of \eqref{eq:showthis-leftside} (keeping all labels fixed):
\begin{gather*}
%
\end{gather*}
Naturality of the braiding is used in all non-marked equations. Therefore, starting from expression \eqref{eq:showthis-leftside} (in the following we only emphasize relevant labels under the summation symbol, however it is summed over all appearing labels),
\begin{gather*}
\bigoplus \sum \sum\limits_{\substack{\gamma\\ \xi=\{ \alpha,\alpha',\beta \}}} \; \, \, \xi \; \; %
 \; \, \, .
\end{gather*}
This is exactly expression \eqref{eq:showthis-rightside}, thus we completed the proof that $f_1$ from Equation \eqref{eq:def-f1} is an algebra homomorphism.\medskip

We proceed analogously to show that $f_2$ from Equation \eqref{eq:def-f2} is an algebra homomorphism. We have to show that
\begin{gather*}
%
\end{gather*}
are equal. Observe that we can ignore all lines labeled $\tilde{k}$, $\tilde{k}'$ and $\tilde{k}''$, since they are the same in both of the above diagrams. Furthermore, we change bases according to
\begin{gather*}
\left\{\raisebox{0.1cm}{%
\right\} \; \, .
\end{gather*}
From now on we use $\varphi$ and $\psi$ as labels for a basis of $\Hom(j' \otimes j,j'')$ and $\Hom(i \otimes i', i'')$, respectively, by abuse of notation. Then it remains to show
\begin{gather}
\bigoplus \sum \; \, \,  %
 \; \, \, . \nonumber
\end{gather}

These expressions contain diagrams which are similar to the ones appearing in the calculation for $f_1$ (cf. \eqref{eq:showthis-leftside} and \eqref{eq:showthis-rightside}). In analogy to the basis $\{ \xi \}$ and its ``bar'' dual $\{ \bar{\xi} \}$ above (Equation \eqref{eq:xi-bar}) we define the basis of $\Hom(i \otimes j^\ast \otimes i' \otimes j'^\ast,\tilde{j}''^\ast)$ formed by
\begin{gather*}
\sigma \; \defeq  \; \, \, %
\end{gather*}
(for basis vectors $\alpha$, $\alpha'$ and $\delta$), with its dual basis formed by
\begin{gather*}
\bar{\sigma} \; \defeq  \; \, \, %
 \, \, .
\end{gather*}
Now we can proceed by expressing the second component of \eqref{eq:f2-left-hand-side} in terms of $\bar{\sigma}$, i.e.
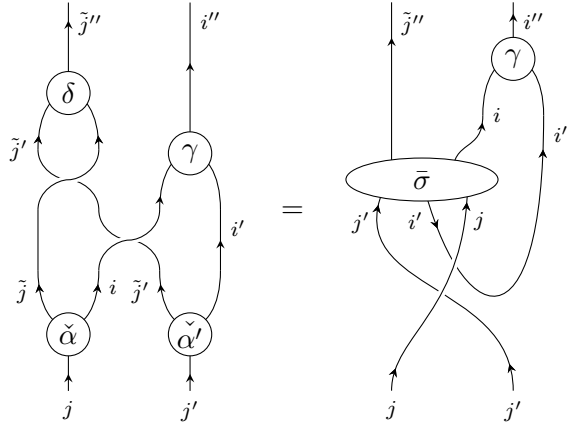
\begin{gather*}
%
\end{gather*}
(cf. the transformations succeeding Equation \eqref{eq:zeta-bar}) and then change the basis $\{ \sigma \}$ to the basis formed by
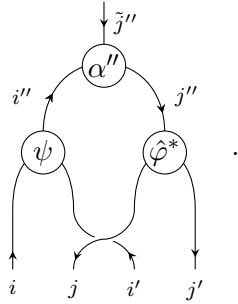
\begin{gather*}
%
 \; \, \, .
\end{gather*}
Again it is allowed to perform such a basis change since $\sum_\sigma \sigma \otimes_{\mathds{k}} \bar{\sigma}$ is basis independent. One immediately verifies that the resulting calculation is obtained by appropriate relabeling of the lines and vertices in parts of the diagrams in the calculation for $f_1$. Thus $f_2$ is an algebra homomorphism.\medskip

In summary, $f$ (Equation \eqref{eq:def-f}) is an algebra isomorphism between $A$ (Equation \eqref{eq:def-A}) and $B$ (Equation \eqref{eq:def-B}), which proves Proposition \ref{prop:rep-Sn-3rd} and therefore completes the proof of Theorem \ref{thm:rep-Sn}.

\begin{flushright}
$\Box$
\end{flushright}

\section{Outlook}
\label{sec:outlook}

In this paper we constructed a representation of the symmetric group $S_n$ on the Deligne power $\mathcal{C}^{\boxtimes n}$ of a modular tensor category $\mathcal{C}$ by explicitly specifying $\mathcal{C}^{\boxtimes n}$-bimodule categories, which correspond to the adjacent transpositions generating $S_n$. We want to comment on open problems emerging from this construction.\medskip

Let $\mathcal{B}$ be a modular tensor category and let $G$ be a finite group. A representation $[\varrho]\colon G \to \mathsf{Pic}(\mathcal{B})$ (Definition \ref{def:representation-on-cat}) is related to the specification of surface defects corresponding to a global $G$-symmetry of the topological phase described by $\mathcal{B}$. ``Gauging the $G$-symmetry'' is the process of promoting the global symmetry to a local one, so that the extrinsic defects are turned into deconfined quasiparticles in a new topological phase \cite{BBCW}. On the categorical side, gauging allows to produce a new modular tensor category from a given input modular tensor category together with an action of a group: first, $\mathcal{B}$ is extended to a ``$G$-crossed braided extension'', and then the $G$-action is ``equivariantized'' to obtain a new modular tensor category \cite{CGPW}. This two-step process amounts to lifting the group representation $[\varrho]$ to a representation on the level of tricategories, in a sense which we now want to make precise.

\begin{definition}
Let $G$ be a group. 
\begin{itemize}
\item $\underline{G}$ is the monoidal category with the elements of $G$ as objects, only identity morphisms, and monoidal structure induced by the group multiplication.
\item $\dunderline{G}$ is the monoidal bicategory \cite[Definition 2.20]{greenough} with the elements of $G$ as objects, only identity 1- and 2-morphisms, and monoidal structure induced by the group multiplication.
\item $B\dunderline{G}$ is the tricategory \cite[Definition 4.1]{gurski} with a single object $\ast$ and $\Hom(\ast,\ast)=\dunderline{G}$.
\end{itemize}
\end{definition}

The definitions of a $G$-crossed braided extension and equivariantization will not be stated here (see \cite{eno09,turaev-hqft} for details). We consider the following characterization: a $G$-crossed braided extension of $\mathcal{B}$ defines a group homomorphism $G \to \mathsf{Pic}(\mathcal{B})$, plus additional data from the tensor product and the associator of the extension, which can be encoded into a 2-functor $\dunderline{G} \to \dunderline{\mathsf{Pic}}(\mathcal{B})$ of monoidal bicategories (recall Definition \ref{def:picard} of the Picard 3-groupoid). Equivalently, such a 2-functor of monoidal bicategories is a 3-functor $B\dunderline{G} \to \dunderline{\mathsf{Pic}}$, which maps the single object of $B\dunderline{G}$ (denoted $\ast$) to $\mathcal{B}\in\dunderline{\mathsf{Pic}}$. We refer to \cite[Definition 4.10]{gurski} for the definition of a 3-functor, a morphism between tricategories.

\begin{theorem}[{\cite[Theorem 7.12]{eno09}}]
Equivalence classes of $G$-crossed extensions of $\mathcal{B}$ are in bijection with 3-functors $B\dunderline{G} \to \dunderline{\mathsf{Pic}}$, which map the single object of $B\dunderline{G}$ to $\mathcal{B}\in\dunderline{\mathsf{Pic}}$.
\end{theorem}
Thus ``gauging'' can be defined as the process of passing from a group homomorphism $G \to \mathsf{Pic}(\mathcal{B})$ to a 3-functor \cite{CGPW}:

\begin{definition}[Gauging]
\label{def:gauging}
A global symmetry $(G,[\varrho])$ of $\mathcal{B}$ (cf. Definition \ref{def:representation-on-cat}) can be \textit{gauged}, if there exists a 3-functor $\varrho\colon B\dunderline{G} \to \dunderline{\mathsf{Pic}}$ with $\varrho(\ast)=\mathcal{B}$, such that $\varrho$ is equivalent to $[\varrho]$ under the truncation map $\dunderline{\mathsf{Pic}} \to \mathsf{Pic}$.
\end{definition}

Gauging a global symmetry $(G,[\varrho])$ of a modular tensor category $\mathcal{B}$ is not always possible. By \cite{eno09}, certain obstruction classes in the third and fourth cohomology of $G$ have to vanish to ensure the existence of a $G$-crossed braided extension. However, once a $G$-crossed braided extension is found, equivariantization is always possible. In the case of permutation actions the following problems arise.

\begin{problem}
Is it possible to lift the group homomorphism $[\varrho]\colon S_n \to \mathsf{Pic}(\mathcal{C}^{\boxtimes n})$ of Theorem \ref{thm:rep-Sn} to a 3-functor $\varrho\colon B\dunderline{S_n} \to \dunderline{\mathsf{Pic}}$ with $\varrho(\ast)=\mathcal{C}^{\boxtimes n}$, and if yes, how many different such 3-functors are there? Put differently, do $S_n$-crossed braided extensions of $\mathcal{C}^{\boxtimes n}$ exist for all $n$?
\end{problem}

In a recent paper by Gannon and Jones \cite{gj2018} the question of existence is answered affirmatively: both obstruction classes vanish for the permutation action, hence gauging is always possible. For any subgroup $G \subseteq S_n$, the equivalence classes of $G$-crossed braided extensions then form a torsor over $H^3 (G,\mathds{k}^{\times})$ (for $\mathds{k}$-linear $\mathcal{C}$). By \cite{ejp2018}, there are precisely $n$ distinct $\mathds{Z}/n\mathds{Z}$-crossed braided extensions of $\mathcal{C}^{\boxtimes n}$; in particular there are precisely two distinct $S_2$-crossed braided extensions of $\mathcal{C}\boxtimes\mathcal{C}$.

\begin{problem}
\label{problem}
Specify the data of an $S_n$-crossed braided extension of $\mathcal{C}^{\boxtimes n}$ corresponding to $[\varrho]$ by constructing a 3-functor $B\dunderline{S_n} \to \dunderline{\mathsf{Pic}}$.
\end{problem}

For $S_2$ this problem is partially solved in \cite{ejp2018}, where the fusion rules of the two $S_2$-crossed braided extensions of $\mathcal{C}\boxtimes\mathcal{C}$ are specified. In \cite{gj2018} a canonical lift of $G \to \mathsf{Pic}(\mathcal{C}^{\boxtimes n})$ to a monoidal functor $\underline{G} \to \underline{\mathsf{EqBr}}(\mathcal{C}^{\boxtimes n}) \simeq \underline{\mathsf{Pic}}(\mathcal{C}^{\boxtimes n})$ for any subgroup $G \subseteq S_n$ is constructed, and in the case $G=S_n$, $n \geq 3$ this monoidal functor is unique if and only if the group of invertible objects of $\mathcal{C}$ has odd order. However, to the best of the author's knowledge the problem of constructing the explicit data for an $S_n$-crossed braided extension of the permutation action by specifying a 3-functor $B\dunderline{S_n} \to \dunderline{\mathsf{Pic}}$ is still open. The monoidal functor $\underline{S_n} \to \underline{\mathsf{Pic}}(\mathcal{C}^{\boxtimes n})$ from \cite{gj2018} may help in the general case of $n>2$. A future goal is to construct (at least one) such a 3-functor based on the representation of Theorem \ref{thm:rep-Sn}. The equivariantization of the associated $S_n$-crossed braided extension could lead to a family of potentially new modular tensor categories.\medskip

Our results could also provide a deeper insight into the properties of topological quantum computers based on topological multilayer phases. As mentioned in the introductory section, adding permutation twist defects to a non-universal non-abelian state of a topological phase can lead to universal topological quantum computing. This has been demonstrated to be the case for twist defects in a bilayer Ising phase \cite[Section V.]{BaJQ1}. It would be interesting to verify if universality also arises when permutation defects are added to multilayer Ising phases and to multilayer systems of other models like the three-state Potts model \cite{potts-model} and weakly integral anyon models (which are believed to be non-universal by the ``property \textbf{F} conjecture'' of \cite{naidu-rowell}).


\newpage
\appendix

\section{Dual bases of Hom spaces}
\label{appendix:algebras-graphical-notation}

Let $\mathcal{C}$ be a $\mathds{k}$-linear pivotal fusion category (for an algebraically closed field $\mathds{k}$ of zero characteristic) and let $\{X_i\}_{i\in\mathcal{I}}$ be a collection of representatives of the isomorphism classes of simple objects, where $\mathcal{I}$ is some index set (see Section \ref{sec:prelim} for definitions of these notions, and \cite{bak-kir,egno,kassel} for details). We will write $i,j,k,\ldots$ instead of $X_i,X_j,X_k,\ldots$ for short. Choose a basis $\{ \alpha_m \}_m$ of $\Hom(i\otimes j,k)$, with $m=1,\ldots,N_{ij}^k$ for $N_{ij}^k = \dim(\Hom(i\otimes j,k)) \in \mathds{N}$. To avoid lengthy notation we will simply denote the basis $\{ \alpha_m \}_m$ by the labels, $\{ \alpha \}$, leaving the range $\alpha=1,\ldots,N_{ij}^k$ for the different basis elements implicit. A specific basis element $\alpha$ is graphically represented by a vertex labeled $\alpha$ having two incoming edges and one outgoing edge:
\begin{gather}
\label{eq:basis-alpha}
\alpha \; \equiv \; \, \, \begin{tikzpicture}[baseline=(current bounding box.center)]
\draw[-<-=.8,->-=.2] (2,0) node[above left] (i) {$\mathsmaller{i}$} .. controls +(0,1) and +(0,1) .. (2.8,0) node[above right] (j) {$\mathsmaller{j}$};
\draw[->-=.7] (2.4,0.75) to (2.4,1.5) node[below right] (k) {$\mathsmaller{k}$}; 
\draw[black,fill=white] (2.4,0.75) circle (0.285cm) node (alpha) {$\alpha$};
\end{tikzpicture}
\end{gather}

We remark that writing $\{ \alpha \}$ instead of $\{ \alpha_m \}_m$ enables us to drop indices in the notation, but doing so requires to distinguish different bases graphically since a particular basis is only denoted by its index label. Most of the time we will consider only one basis for each Hom space, thus our notation convention is unambiguous. (If the basis is changed one could emphasize this in the diagrams by choosing a different vertex layout, e.g. squares instead of circles.)

Whenever two parallel morphisms $g$ and $h$ only differ by a constant, we write $\langle g \rangle_h \in \mathds{k}$ for the element of $\mathds{k}$ such that $g=\langle g \rangle_h \cdot h$.

We define the dual basis $\{ \hat{\alpha} \}$ of $\Hom(k,i\otimes j)$ via the non-degenerate pairing
\begin{gather}
\kappa\colon \Hom(i\otimes j,k) \otimes \Hom(k,i\otimes j) \to \mathds{k} \nonumber\\
\kappa\left(
\begin{tikzpicture}[baseline=(current bounding box.center)]
\draw[-<-=.8,->-=.2] (2,0) node[above left] (i) {$\mathsmaller{i}$} .. controls +(0,1) and +(0,1) .. (2.8,0) node[above right] (j) {$\mathsmaller{j}$};
\draw[->-=.7] (2.4,0.75) to (2.4,1.5) node[below right] (k) {$\mathsmaller{k}$}; 
\draw[black,fill=white] (2.4,0.75) circle (0.285cm) node {$\alpha$};
\end{tikzpicture} \; \, \raisebox{-0.7cm}{,} \, \;
\begin{tikzpicture}[baseline=(current bounding box.center)]
\draw[->-=.9,-<-=.1] (2,0) node[below left] (i2) {$\mathsmaller{i}$} .. controls +(0,-1) and +(0,-1) .. (2.8,0) node[below right] (j2) {$\mathsmaller{j}$};
\draw[-<-=.5] (2.4,-0.75) to (2.4,-1.5) node[above right] (k2) {$\mathsmaller{k}$}; 
\draw[black,fill=white] (2.4,-0.75) circle (0.285cm) node {$\hat{\beta}$};
\end{tikzpicture}
\right) \, \defeq  \,
{\scaleleftright[3ex]{\Biggl\langle}{\begin{tikzpicture}[baseline=(current bounding box.center)]
\draw[-<-=.99,->-=.01] (2,0) node[left] (i) {$\mathsmaller{i}$} .. controls +(0,1) and +(0,1) .. (2.8,0) node[right] (j) {$\mathsmaller{j}$};
\draw[->-=.7] (2.4,0.75) to (2.4,1.5) node[below right] (k) {$\mathsmaller{k}$}; 
\draw[black,fill=white] (2.4,0.75) circle (0.285cm) node {$\alpha$};
\draw (2,0) .. controls +(0,-1) and +(0,-1) .. (2.8,0);
\draw[-<-=.5] (2.4,-0.75) to (2.4,-1.5) node[above right] (k2) {$\mathsmaller{k}$}; 
\draw[black,fill=white] (2.4,-0.75) circle (0.285cm) node {$\hat{\beta}$};
\end{tikzpicture}}{\Biggr\rangle}}_{\id_k}
\, = \, \delta_{\alpha\beta} \, \, ,
\label{eq:def-pairing-kappa}
\end{gather}
i.e. $\hat{\beta}$ is defined by
\begin{gather}
\begin{tikzpicture}[baseline=(current bounding box.center)]
\draw[-<-=.99,->-=.01] (2,0) node[left] {$\mathsmaller{i}$} .. controls +(0,1) and +(0,1) .. (2.8,0) node[right] {$\mathsmaller{j}$};
\draw[->-=.7] (2.4,0.75) to (2.4,1.5) node[below right] {$\mathsmaller{k}$}; 
\draw[black,fill=white] (2.4,0.75) circle (0.285cm) node {$\alpha$};
\draw (2,0) .. controls +(0,-1) and +(0,-1) .. (2.8,0);
\draw[-<-=.5] (2.4,-0.75) to (2.4,-1.5) node[above right] {$\mathsmaller{k}$}; 
\draw[black,fill=white] (2.4,-0.75) circle (0.285cm) node {$\hat{\beta}$};
\end{tikzpicture} \; \, \defeq  \delta_{\alpha\beta} \id_k \, \, .
\label{eq:dual-kappa}
\end{gather}
Note that the argument of $\left\langle . \right\rangle_{\id_k}$ above is an element of $\End(k) \cong \mathds{k}$ (since $k$ is simple) and hence indeed is a multiple of $\id_k$.

Furthermore, semisimplicity of $\mathcal{C}$ and duality imply the ``completeness relation''
\begin{gather}
\bigoplus\limits_k \sum\limits_{\alpha} \; \, \, \begin{tikzpicture}[baseline=(current bounding box.center)]
\draw[-<-=.8,->-=.2] (6.4,0) node[below] {$\mathsmaller{i}$} .. controls +(0,1) and +(0,1) .. (7.2,0) node[below] {$\mathsmaller{j}$};
\draw (6.8,0.75) to (6.8,1.25) node[right] {$\mathsmaller{k}$}; 
\draw[black,fill=white] (6.8,0.75) circle (0.285cm) node {$\alpha$};
\draw[->-=.9,-<-=.1] (6.4,2.75) node[left] {$\mathsmaller{i}$} .. controls +(0,-1) and +(0,-1) .. (7.2,2.75) node[right] {$\mathsmaller{j}$};
\draw[->-=.2] (6.8,1.25) to (6.8,2); 
\draw[black,fill=white] (6.8,2) circle (0.285cm) node {$\hat{\alpha}$};
\end{tikzpicture} \; \, \, = \, \id_{i \otimes j} \, \, .
\label{eq:appendix-completeness-relation}
\end{gather}

Using the rigid structure of $\mathcal{C}$ (and the pivotal structure to identify left and right duals), there are isomorphisms \cite[Lemma 2.1.6]{bak-kir}
\begin{gather}
\Hom(k,i\otimes j) \cong \Hom(j^\ast \otimes i^\ast,k^\ast) \cong \Hom(j^\ast , k^\ast \otimes i) \, \, ,
\label{eq:hom-isomorphisms}
\end{gather}
which are obtained by appropriate compositions with (co)evaluations and give rise to other bases dual to $\{ \alpha \}$ with respect to different non-degenerate pairings. The goal is to elaborate how these different pairings are related. In particular, we will see that the notion of duality with respect to $\kappa$ is compatible with the canonical isomorphisms of Equation \eqref{eq:hom-isomorphisms}. In order to clearly distinguish the different notions of duality, we will refer to dual morphisms (Equations \eqref{eq:cat-dual-morphism-right} and \eqref{eq:cat-dual-morphism-left}) obtained from the rigid structure of the category as ``categorical duals'' and call duality in the sense of Equation \eqref{eq:dual-kappa} ``$\kappa$-duality''.\medskip

Let $\eta$ be the non-degenerate pairing
\begin{gather}
\eta\colon \Hom(i\otimes j,k) \otimes \Hom(j^\ast \otimes i^\ast,k^\ast) \to \mathds{k} \nonumber\\
\eta\left(
\begin{tikzpicture}[baseline=(current bounding box.center)]
\draw[-<-=.8,->-=.2] (2,0) node[above left] (i) {$\mathsmaller{i}$} .. controls +(0,1) and +(0,1) .. (2.8,0) node[above right] (j) {$\mathsmaller{j}$};
\draw[->-=.7] (2.4,0.75) to (2.4,1.5) node[below right] (k) {$\mathsmaller{k}$}; 
\draw[black,fill=white] (2.4,0.75) circle (0.285cm) node {$\alpha$};
\end{tikzpicture} \; \, \raisebox{-0.7cm}{,} \, \;
\begin{tikzpicture}[baseline=(current bounding box.center)]
\draw[->-=.9,-<-=.1] (2,0) node[above left] {$\mathsmaller{j}$} .. controls +(0,1) and +(0,1) .. (2.8,0) node[above right] {$\mathsmaller{i}$};
\draw[-<-=.7] (2.4,0.75) to (2.4,1.5) node[below right] {$\mathsmaller{k}$}; 
\draw[black,fill=white] (2.4,0.75) circle (0.285cm) node {$\hat{\beta}^\ast$};
\end{tikzpicture}
\right) \, \defeq  \, 
{\scaleleftright[3ex]{\Biggl\langle}{\begin{tikzpicture}[baseline=(current bounding box.center)]
\draw[-<-=.8,->-=.2] (2,0) node[above left] {$\mathsmaller{i}$} .. controls +(0,1) and +(0,1) .. (2.8,0) node[above right] {$\mathsmaller{j}$};
\draw[->-=.7] (2.4,0.75) to (2.4,1.5) node[below right] {$\mathsmaller{k}$}; 
\draw[white] (2.4,1.5) to (2.4,2);
\draw[black,fill=white] (2.4,0.75) circle (0.285cm) node {$\alpha$};
\draw[->-=.9,-<-=.1] (3.6,0) node[above left] {$\mathsmaller{j}$} .. controls +(0,1) and +(0,1) .. (4.4,0) node[above right] {$\mathsmaller{i}$};
\draw[-<-=.7] (4,0.75) to (4,1.5) node[below right] {$\mathsmaller{k}$}; 
\draw[white] (4,1.5) to (4,2);
\draw[black,fill=white] (4,0.75) circle (0.285cm) node {$\hat{\beta}^\ast$};
\draw (2.8,0) .. controls +(0,-0.75) and +(0,-0.75) .. (3.6,0);
\draw (2,0) .. controls +(0,-1.5) and +(0,-1.5) .. (4.4,0);
\end{tikzpicture}}{\Biggr\rangle}}_{\widetilde{\coev}_k} \, \, .
\label{eq:def-pairing-eta}
\end{gather}
If $\hat{\beta}^\ast$ comes from some $\hat{\beta}\in\Hom(k,i\otimes j)$ which is $\kappa$-dual to $\beta\in\Hom(i\otimes j,k)$, i.e.
\begin{gather}
\begin{tikzpicture}[baseline=(current bounding box.center)]
\draw[->-=.9,-<-=.1] (2,0) node[above left] {$\mathsmaller{j}$} .. controls +(0,1) and +(0,1) .. (2.8,0) node[above right] {$\mathsmaller{i}$};
\draw[-<-=.7] (2.4,0.75) to (2.4,1.5) node[below right] {$\mathsmaller{k}$}; 
\draw[black,fill=white] (2.4,0.75) circle (0.285cm) node {$\hat{\beta}^\ast$};
\end{tikzpicture} \, = \,
\begin{tikzpicture}[baseline=(current bounding box.center)]
\draw[->-=.9,-<-=.1] (2,1.5) node[below left] {$\mathsmaller{i}$} .. controls +(0,-1) and +(0,-1) .. (2.8,1.5) node[below right] {$\mathsmaller{j}$};
\draw[-<-=.6] (3.6,0.75) to (3.6,2.3) node[below left] {$\mathsmaller{k}$};
\draw (2.4,0.75) .. controls +(0,-1.2) and +(0,-1.2) .. (3.6,0.75);
\draw[black,fill=white] (2.4,0.75) circle (0.285cm) node {$\hat{\beta}$};
\draw[-<-=.5] (0.4,-0.2) node[above left] {$\mathsmaller{j}$} to (0.4,0.75);
\draw (0.4,0.75) .. controls +(0,1.5) and +(0,1.5) .. (2.8,1.5);
\draw[-<-=.5] (1.2,-0.2) node[above left] {$\mathsmaller{i}$} to (1.2,0.75);
\draw (1.2,0.75) .. controls +(0,1.15) and +(0,0.65) .. (2,1.5);
\end{tikzpicture}
\label{eq:beta-hat-star}
\end{gather}
(so $\hat{\beta}^\ast$ really is the categorical dual morphism of $\hat{\beta}$, as suggested by the notation), then
\begin{gather}
\eta\left(
\begin{tikzpicture}[baseline=(current bounding box.center)]
\draw[-<-=.8,->-=.2] (2,0) node[above left] {$\mathsmaller{i}$} .. controls +(0,1) and +(0,1) .. (2.8,0) node[above right] {$\mathsmaller{j}$};
\draw[->-=.7] (2.4,0.75) to (2.4,1.5) node[below right] {$\mathsmaller{k}$}; 
\draw[black,fill=white] (2.4,0.75) circle (0.285cm) node {$\alpha$};
\end{tikzpicture} \; \, \raisebox{-0.7cm}{,} \, \;
\begin{tikzpicture}[baseline=(current bounding box.center)]
\draw[->-=.9,-<-=.1] (2,0) node[above left] {$\mathsmaller{j}$} .. controls +(0,1) and +(0,1) .. (2.8,0) node[above right] {$\mathsmaller{i}$};
\draw[-<-=.7] (2.4,0.75) to (2.4,1.5) node[below right] {$\mathsmaller{k}$}; 
\draw[black,fill=white] (2.4,0.75) circle (0.285cm) node {$\hat{\beta}^\ast$};
\end{tikzpicture}
\right) \, = \, {\scaleleftright[3ex]{\Biggl\langle}{\begin{tikzpicture}[baseline=(current bounding box.center)]
\draw[->-=.9,-<-=.1] (2,1.5) node[below left] {$\mathsmaller{i}$} .. controls +(0,-1) and +(0,-1) .. (2.8,1.5) node[below right] {$\mathsmaller{j}$};
\draw[-<-=.5] (3.6,0.75) to (3.6,2.3) node[below left] {$\mathsmaller{k}$};
\draw (2.4,0.75) .. controls +(0,-1.2) and +(0,-1.2) .. (3.6,0.75);
\draw[black,fill=white] (2.4,0.75) circle (0.285cm) node {$\hat{\beta}$};
\draw[-<-=.5] (0.4,0.5) node[above left] {$\mathsmaller{j}$} to (0.4,0.75);
\draw (0.4,0.75) .. controls +(0,1.5) and +(0,1.5) .. (2.8,1.5);
\draw[-<-=.5] (1.2,0.5) node[above left] {$\mathsmaller{i}$} to (1.2,0.75);
\draw (1.2,0.75) .. controls +(0,1.15) and +(0,0.65) .. (2,1.5);
\draw[-<-=.8,->-=.2] (-1.2,0.75) node[above left] {$\mathsmaller{i}$} .. controls +(0,1) and +(0,1) .. (-0.4,0.75) node[above right] {$\mathsmaller{j}$};
\draw[->-=.7] (-0.8,1.5) to (-0.8,2.3) node[below right] {$\mathsmaller{k}$}; 
\draw[black,fill=white] (-0.8,1.5) circle (0.285cm) node {$\alpha$};
\draw (-0.4,0.75) .. controls +(0,-0.95) and +(0,-0.75) .. (0.4,0.5);
\draw (-1.2,0.75) .. controls +(0,-1.5) and +(0,-1.5) .. (1.2,0.5);
\end{tikzpicture} \; \;}{\Biggr\rangle}}_{\widetilde{\coev}_k} \, = \nonumber \\
= \; \; {\scaleleftright[3ex]{\Biggl\langle}{\begin{tikzpicture}[baseline=(current bounding box.center)]
\draw[-<-=.99,->-=.01] (2,0) node[left] {$\mathsmaller{i}$} .. controls +(0,1) and +(0,1) .. (2.8,0) node[right] {$\mathsmaller{j}$};
\draw[->-=.7] (2.4,0.75) to (2.4,1.5) node[below right] {$\mathsmaller{k}$}; 
\draw[white] (2.4,1.5) to (2.4,2);
\draw[black,fill=white] (2.4,0.75) circle (0.285cm) node {$\alpha$};
\draw (2,0) .. controls +(0,-1) and +(0,-1) .. (2.8,0);
\draw[-<-=.5] (3.6,0) to (3.6,1.5) node[below left] {$\mathsmaller{k}$};
\draw[white] (3.6,1.5) to (3.6,2);
\draw (2.4,-0.75) .. controls +(0,-1.2) and +(0,-1.6) .. (3.6,0);
\draw[black,fill=white] (2.4,-0.75) circle (0.285cm) node {$\hat{\beta}$};
\end{tikzpicture} \; \;}{\Biggr\rangle}}_{\widetilde{\coev}_k} = \, \delta_{\alpha\beta} \, \label{eq:eta-dual} = \\
= \, \kappa\left(
\begin{tikzpicture}[baseline=(current bounding box.center)]
\draw[-<-=.8,->-=.2] (2,0) node[above left] (i) {$\mathsmaller{i}$} .. controls +(0,1) and +(0,1) .. (2.8,0) node[above right] (j) {$\mathsmaller{j}$};
\draw[->-=.7] (2.4,0.75) to (2.4,1.5) node[below right] (k) {$\mathsmaller{k}$}; 
\draw[black,fill=white] (2.4,0.75) circle (0.285cm) node {$\alpha$};
\end{tikzpicture} \; \, \raisebox{-0.7cm}{,} \, \;
\begin{tikzpicture}[baseline=(current bounding box.center)]
\draw[->-=.9,-<-=.1] (2,0) node[below left] (i2) {$\mathsmaller{i}$} .. controls +(0,-1) and +(0,-1) .. (2.8,0) node[below right] (j2) {$\mathsmaller{j}$};
\draw[-<-=.5] (2.4,-0.75) to (2.4,-1.5) node[above right] (k2) {$\mathsmaller{k}$}; 
\draw[black,fill=white] (2.4,-0.75) circle (0.285cm) node {$\hat{\beta}$};
\end{tikzpicture}
\right) \, \equiv \, \kappa\left(
\begin{tikzpicture}[baseline=(current bounding box.center)]
\draw[-<-=.8,->-=.2] (2,0) node[above left] (i) {$\mathsmaller{i}$} .. controls +(0,1) and +(0,1) .. (2.8,0) node[above right] (j) {$\mathsmaller{j}$};
\draw[->-=.7] (2.4,0.75) to (2.4,1.5) node[below right] (k) {$\mathsmaller{k}$}; 
\draw[black,fill=white] (2.4,0.75) circle (0.285cm) node {$\alpha$};
\end{tikzpicture} \; \, \raisebox{-0.7cm}{,} \, \;
\begin{tikzpicture}[baseline=(current bounding box.center)]
\draw[->-=.9,-<-=.1] (2,0) node[above left] {$\mathsmaller{j}$} .. controls +(0,1) and +(0,1) .. (2.8,0) node[above right] {$\mathsmaller{i}$};
\draw[->-=.5] (3.6,-0.8) node[above left] {$\mathsmaller{k}$} to (3.6,0.75);
\draw (2.4,0.75) .. controls +(0,1.2) and +(0,1.2) .. (3.6,0.75);
\draw[black,fill=white] (2.4,0.75) circle (0.285cm) node {$\hat{\beta}^\ast$};
\draw[->-=.5] (0.4,0.75) to (0.4,1.7) node[below left] {$\mathsmaller{i}$};
\draw (0.4,0.75) .. controls +(0,-1.5) and +(0,-1.5) .. (2.8,0);
\draw[->-=.5] (1.2,0.75) to (1.2,1.7) node[below left] {$\mathsmaller{j}$};
\draw (1.2,0.75) .. controls +(0,-1.15) and +(0,-0.65) .. (2,0);
\end{tikzpicture} \; \;
\right) \, \, ,\nonumber
\end{gather}
where the snake identity (Equation \eqref{eq:snake-identities-left}) is used in the step from the first to the second line and then Equation \eqref{eq:dual-kappa} is employed.

Note that $\hat{\beta}^\ast$ in Equation \eqref{eq:beta-hat-star} is obtained by first dualizing $\beta$ with respect to $\kappa$ and then taking the categorical dual, i.e. it is in fact $(\hat{\beta})^\ast$, but exchanging these operations results in the same morphism: 
\begin{lemma}
\label{lem:cat-dual-kappa-dual}
Applying the categorical dual commutes with dualizing with respect to $\kappa$, i.e. $\widehat{\beta^\ast} = (\hat{\beta})^\ast \equiv \hat{\beta}^\ast$.
\end{lemma}
\begin{proof}
The categorical dual
\begin{gather*}
\beta^\ast \; = \; \, \, \begin{tikzpicture}[baseline=(current bounding box.center)]
\draw[-<-=.8,->-=.2] (2,0) node[above left] {$\mathsmaller{i}$} .. controls +(0,1) and +(0,1) .. (2.8,0) node[above right] {$\mathsmaller{j}$};
\draw[-<-=.2] (1.2,-0.8) node[above left] {$\mathsmaller{k}$} to (1.2,0.75);
\draw (1.2,0.75) .. controls +(0,1.2) and +(0,1.2) .. (2.4,0.75);
\draw[black,fill=white] (2.4,0.75) circle (0.285cm) node {$\beta$};
\draw[-<-=.5] (3.6,0.75) to (3.6,1.7) node[below left] {$\mathsmaller{j}$};
\draw (2.8,0) .. controls +(0,-0.65) and +(0,-1.15) .. (3.6,0.75);
\draw[-<-=.5] (4.4,0.75) to (4.4,1.7) node[below left] {$\mathsmaller{i}$};
\draw (2,0) .. controls +(0,-1.5) and +(0,-1.5) .. (4.4,0.75);
\end{tikzpicture}
\end{gather*}
of $\beta$ is $\kappa$-dual to $(\hat{\beta})^\ast$ (Equation \eqref{eq:beta-hat-star}), as is evident by using the snake identity several times and Equation \eqref{eq:dual-kappa}:
\begin{gather*}
\begin{tikzpicture}[baseline=(current bounding box.center)]
\draw[->-=.9,-<-=.1] (5.2,1.75) node[below left] {$\mathsmaller{i}$} .. controls +(0,-1) and +(0,-1) .. (6,1.75) node[below right] {$\mathsmaller{j}$};
\draw[-<-=.6] (6.8,1) to (6.8,2.55) node[below left] {$\mathsmaller{k}$};
\draw (5.6,1) .. controls +(0,-1.2) and +(0,-1.2) .. (6.8,1);
\draw[black,fill=white] (5.6,1) circle (0.285cm) node {$\hat{\gamma}$};
\draw (3.6,1) .. controls +(0,1.5) and +(0,1.5) .. (6,1.75);
\draw (4.4,1) .. controls +(0,1.15) and +(0,0.65) .. (5.2,1.75);
\draw[-<-=.8,->-=.2] (2,0) node[above left] {$\mathsmaller{i}$} .. controls +(0,1) and +(0,1) .. (2.8,0) node[above right] {$\mathsmaller{j}$};
\draw[-<-=.2] (1.2,-0.8) node[above left] {$\mathsmaller{k}$} to (1.2,0.75);
\draw (1.2,0.75) .. controls +(0,1.2) and +(0,1.2) .. (2.4,0.75);
\draw[black,fill=white] (2.4,0.75) circle (0.285cm) node {$\beta$};
\draw[-<-=.5] (3.6,0.75) to (3.6,1) node[below left] {$\mathsmaller{j}$};
\draw (2.8,0) .. controls +(0,-0.65) and +(0,-1.15) .. (3.6,0.75);
\draw[-<-=.5] (4.4,0.75) to (4.4,1) node[below left] {$\mathsmaller{i}$};
\draw (2,0) .. controls +(0,-1.5) and +(0,-1.5) .. (4.4,0.75);
\end{tikzpicture} \; \; \, \, = \; \, \, \delta_{\beta\gamma} \id_{k^\ast}
\end{gather*}
$\,$
\end{proof}

Finally we introduce the non-degenerate pairing $\vartheta$,
\begin{gather}
\vartheta\colon \Hom(i\otimes j,k) \otimes \Hom(j^\ast, k^\ast \otimes i) \to \mathds{k} \nonumber\\
\vartheta\left(
\begin{tikzpicture}[baseline=(current bounding box.center)]
\draw[-<-=.8,->-=.2] (2,0) node[above left] (i) {$\mathsmaller{i}$} .. controls +(0,1) and +(0,1) .. (2.8,0) node[above right] (j) {$\mathsmaller{j}$};
\draw[->-=.7] (2.4,0.75) to (2.4,1.5) node[below right] (k) {$\mathsmaller{k}$}; 
\draw[black,fill=white] (2.4,0.75) circle (0.285cm) node {$\alpha$};
\end{tikzpicture} \; \, \raisebox{-0.7cm}{,} \, \;
\begin{tikzpicture}[baseline=(current bounding box.center)]
\draw[->-=.9,->-=.2] (2,0) node[below left] {$\mathsmaller{k}$} .. controls +(0,-1) and +(0,-1) .. (2.8,0) node[below right] {$\mathsmaller{i}$};
\draw[->-=.7] (2.4,-0.75) to (2.4,-1.5) node[above right] {$\mathsmaller{j}$}; 
\draw[black,fill=white] (2.4,-0.75) circle (0.285cm) node {$\check{\beta}$};
\end{tikzpicture}
\right) \, \defeq  \, 
{\scaleleftright[3ex]{\Biggl\langle}{\begin{tikzpicture}[baseline=(current bounding box.center)]
\draw (2,0) .. controls +(0,-0.7) and +(0,-0.7) .. (2.8,-0.3);
\draw[->-=.7] (2.4,-0.75) to (2.4,-1.5) node[above right] {$\mathsmaller{j}$}; 
\draw[black,fill=white] (2.4,-0.75) circle (0.285cm) node {$\check{\beta}$};
\draw[-<-=.5] (2,0) to (2,1.2) node[below left] {$\mathsmaller{k}$};
\draw[white] (2,1.2) to (2,1.7);
\draw[-<-=.8,->-=.05] (2.8,-0.3) node[above left] {$\mathsmaller{i}$} .. controls +(0,0.7) and +(0,0.7) .. (3.6,-0.3) node[above right] {$\mathsmaller{j}$};
\draw[->-=.7] (3.2,0.45) to (3.2,1.2) node[below right] {$\mathsmaller{k}$}; 
\draw[white] (3.2,1.2) to (3.2,1.7);
\draw[black,fill=white] (3.2,0.25) circle (0.285cm) node {$\alpha$};
\draw (2.4,-1.5) .. controls +(0,-0.5) and +(0,-2) .. (3.6,-0.3);
\end{tikzpicture}}{\Biggr\rangle}}_{\coev_k} \, \, .
\label{eq:def-pairing-theta}
\end{gather}
If $\check{\beta}$ comes from some $\hat{\beta}\in\Hom(k,i\otimes j)$ which is $\kappa$-dual to $\beta\in\Hom(i\otimes j,k)$, i.e.
\begin{gather*}
\begin{tikzpicture}[baseline=(current bounding box.center)]
\draw[->-=.9,->-=.2] (2,0) node[below left] {$\mathsmaller{k}$} .. controls +(0,-1) and +(0,-1) .. (2.8,0) node[below right] {$\mathsmaller{i}$};
\draw[->-=.7] (2.4,-0.75) to (2.4,-1.5) node[above right] {$\mathsmaller{j}$}; 
\draw[black,fill=white] (2.4,-0.75) circle (0.285cm) node {$\check{\beta}$};
\end{tikzpicture} \, = \,
\begin{tikzpicture}[baseline=(current bounding box.center)]
\draw[->-=.9] (2,1.5) .. controls +(0,-1) and +(0,-1) .. (2.8,1.5) node[below right] {$\mathsmaller{j}$};
\draw[->-=.3] (2,1.5) to (2,2.2) node[below left] {$\mathsmaller{i}$};
\draw[-<-=.6] (1.2,0.75) to (1.2,2.2) node[below left] {$\mathsmaller{k}$};
\draw (1.2,0.75) .. controls +(0,-1.2) and +(0,-1.2) .. (2.4,0.75);
\draw[black,fill=white] (2.4,0.75) circle (0.285cm) node {$\hat{\beta}$};
\draw[-<-=.5] (3.6,-0.2) node[above left] {$\mathsmaller{j}$} to (3.6,0.75);
\draw (2.8,1.5) .. controls +(0,0.65) and +(0,1.85) .. (3.6,0.75);
\end{tikzpicture} \; \; ,
\end{gather*}
then
\begin{gather*}
\vartheta\left(
\begin{tikzpicture}[baseline=(current bounding box.center)]
\draw[-<-=.8,->-=.2] (2,0) node[above left] (i) {$\mathsmaller{i}$} .. controls +(0,1) and +(0,1) .. (2.8,0) node[above right] (j) {$\mathsmaller{j}$};
\draw[->-=.7] (2.4,0.75) to (2.4,1.5) node[below right] (k) {$\mathsmaller{k}$}; 
\draw[black,fill=white] (2.4,0.75) circle (0.285cm) node {$\alpha$};
\end{tikzpicture} \; \, \raisebox{-0.7cm}{,} \, \;
\begin{tikzpicture}[baseline=(current bounding box.center)]
\draw[->-=.9,->-=.2] (2,0) node[below left] {$\mathsmaller{k}$} .. controls +(0,-1) and +(0,-1) .. (2.8,0) node[below right] {$\mathsmaller{i}$};
\draw[->-=.7] (2.4,-0.75) to (2.4,-1.5) node[above right] {$\mathsmaller{j}$}; 
\draw[black,fill=white] (2.4,-0.75) circle (0.285cm) node {$\check{\beta}$};
\end{tikzpicture}
\right) \, = \, {\scaleleftright[3ex]{\Biggl\langle}{\begin{tikzpicture}[baseline=(current bounding box.center)]
\draw[->-=.9] (2,1.5) .. controls +(0,-1) and +(0,-1) .. (2.8,1.5) node[below right] {$\mathsmaller{j}$};
\draw[-<-=.8] (1.2,0.75) to (1.2,3.675) node[below left] {$\mathsmaller{k}$};
\draw[white] (1.2,3.675) to (1.2,4.175);
\draw (1.2,0.75) .. controls +(0,-1.2) and +(0,-1.2) .. (2.4,0.75);
\draw[black,fill=white] (2.4,0.75) circle (0.285cm) node {$\hat{\beta}$};
\draw[-<-=.5] (3.6,0.6) node[above left] {$\mathsmaller{j}$} to (3.6,0.75);
\draw (2.8,1.5) .. controls +(0,0.65) and +(0,1.85) .. (3.6,0.75);
\draw[-<-=.9,->-=.1] (2,1.8) node[above left] {$\mathsmaller{i}$} .. controls +(0,1.5) and +(0,1.5) .. (4.4,1.8) node[above right] {$\mathsmaller{j}$};
\draw[->-=.7] (3.2,2.925) to (3.2,3.675) node[below right] {$\mathsmaller{k}$}; 
\draw[white] (3.2,3.675) to (3.2,4.175);
\draw[black,fill=white] (3.2,2.925) circle (0.285cm) node {$\alpha$};
\draw (2,1.5) to (2,1.8);
\draw (3.6,0.6) .. controls +(0,-0.85) and +(0,-2.7) .. (4.4,1.8);
\end{tikzpicture}}{\Biggr\rangle}}_{\coev_k} \, = \\
= \; \; {\scaleleftright[3ex]{\Biggl\langle}{\begin{tikzpicture}[baseline=(current bounding box.center)]
\draw[-<-=.99,->-=.01] (2,0) node[left] {$\mathsmaller{i}$} .. controls +(0,1) and +(0,1) .. (2.8,0) node[right] {$\mathsmaller{j}$};
\draw[->-=.7] (2.4,0.75) to (2.4,1.5) node[below right] {$\mathsmaller{k}$}; 
\draw[white] (2.4,1.5) to (2.4,2);
\draw[black,fill=white] (2.4,0.75) circle (0.285cm) node {$\alpha$};
\draw (2,0) .. controls +(0,-1) and +(0,-1) .. (2.8,0);
\draw[-<-=.5] (1.2,0) to (1.2,1.5) node[below left] {$\mathsmaller{k}$};
\draw[white] (1.2,1.5) to (1.2,2);
\draw (1.2,0) .. controls +(0,-1.6) and +(0,-1.2) .. (2.4,-0.75);
\draw[black,fill=white] (2.4,-0.75) circle (0.285cm) node {$\hat{\beta}$};
\end{tikzpicture} \; \;}{\Biggr\rangle}}_{\coev_k} = \, \delta_{\alpha\beta} \, = \\
= \, \kappa\left(
\begin{tikzpicture}[baseline=(current bounding box.center)]
\draw[-<-=.8,->-=.2] (2,0) node[above left] (i) {$\mathsmaller{i}$} .. controls +(0,1) and +(0,1) .. (2.8,0) node[above right] (j) {$\mathsmaller{j}$};
\draw[->-=.7] (2.4,0.75) to (2.4,1.5) node[below right] (k) {$\mathsmaller{k}$}; 
\draw[black,fill=white] (2.4,0.75) circle (0.285cm) node {$\alpha$};
\end{tikzpicture} \; \, \raisebox{-0.7cm}{,} \, \;
\begin{tikzpicture}[baseline=(current bounding box.center)]
\draw[->-=.9,-<-=.1] (2,0) node[below left] (i2) {$\mathsmaller{i}$} .. controls +(0,-1) and +(0,-1) .. (2.8,0) node[below right] (j2) {$\mathsmaller{j}$};
\draw[-<-=.5] (2.4,-0.75) to (2.4,-1.5) node[above right] (k2) {$\mathsmaller{k}$}; 
\draw[black,fill=white] (2.4,-0.75) circle (0.285cm) node {$\hat{\beta}$};
\end{tikzpicture}
\right) \, \equiv \, \kappa\left(
\begin{tikzpicture}[baseline=(current bounding box.center)]
\draw[-<-=.8,->-=.2] (2,0) node[above left] (i) {$\mathsmaller{i}$} .. controls +(0,1) and +(0,1) .. (2.8,0) node[above right] (j) {$\mathsmaller{j}$};
\draw[->-=.7] (2.4,0.75) to (2.4,1.5) node[below right] (k) {$\mathsmaller{k}$}; 
\draw[black,fill=white] (2.4,0.75) circle (0.285cm) node {$\alpha$};
\end{tikzpicture} \; \, \raisebox{-0.7cm}{,} \, \;
\begin{tikzpicture}[baseline=(current bounding box.center)]
\draw[->-=.2] (2,1.5) node[below right] {$\mathsmaller{k}$} .. controls +(0,-1) and +(0,-1) .. (2.8,1.5);
\draw[->-=.3] (2.8,1.5) to (2.8,2.2) node[below left] {$\mathsmaller{i}$};
\draw[->-=.7] (3.6,0.75) to (3.6,2.2) node[below left] {$\mathsmaller{j}$};
\draw (2.4,0.75) .. controls +(0,-1.2) and +(0,-1.2) .. (3.6,0.75);
\draw[black,fill=white] (2.4,0.75) circle (0.285cm) node {$\check{\beta}$};
\draw[->-=.5] (1.2,-0.2) node[above left] {$\mathsmaller{k}$} to (1.2,0.75);
\draw (1.2,0.75) .. controls +(0,1.85) and +(0,0.65) .. (2,1.5);
\end{tikzpicture} \; \;
\right) \, \, ,
\end{gather*}
where the snake identity (Equation \eqref{eq:snake-identities-right}) is used in the step from the first to the second line and then the definition of duality with respect to $\kappa$ (Equation \eqref{eq:dual-kappa}) is inserted.

We remark that left and right categorical duals are identified by pivotality, thus
\begin{gather*}
\begin{tikzpicture}[baseline=(current bounding box.center)]
\draw[->-=.9,-<-=.1] (2,0) node[above left] {$\mathsmaller{j}$} .. controls +(0,1) and +(0,1) .. (2.8,0) node[above right] {$\mathsmaller{i}$};
\draw[-<-=.7] (2.4,0.75) to (2.4,1.5) node[below right] {$\mathsmaller{k}$}; 
\draw[black,fill=white] (2.4,0.75) circle (0.285cm) node {$\hat{\beta}^\ast$};
\end{tikzpicture} \, = \,
\begin{tikzpicture}[baseline=(current bounding box.center)]
\draw[->-=.9,-<-=.1] (2,1.5) node[below left] {$\mathsmaller{i}$} .. controls +(0,-1) and +(0,-1) .. (2.8,1.5) node[below right] {$\mathsmaller{j}$};
\draw[-<-=.6] (3.6,0.75) to (3.6,2.3) node[below left] {$\mathsmaller{k}$};
\draw (2.4,0.75) .. controls +(0,-1.2) and +(0,-1.2) .. (3.6,0.75);
\draw[black,fill=white] (2.4,0.75) circle (0.285cm) node {$\hat{\beta}$};
\draw[-<-=.5] (0.4,-0.2) node[above left] {$\mathsmaller{j}$} to (0.4,0.75);
\draw (0.4,0.75) .. controls +(0,1.5) and +(0,1.5) .. (2.8,1.5);
\draw[-<-=.5] (1.2,-0.2) node[above left] {$\mathsmaller{i}$} to (1.2,0.75);
\draw (1.2,0.75) .. controls +(0,1.15) and +(0,0.65) .. (2,1.5);
\end{tikzpicture} \; \; \, \, \equiv \; \, \, \begin{tikzpicture}[baseline=(current bounding box.center)]
\draw[->-=.9,-<-=.1] (2,1.5) node[below left] {$\mathsmaller{i}$} .. controls +(0,-1) and +(0,-1) .. (2.8,1.5) node[below right] {$\mathsmaller{j}$};
\draw[-<-=.6] (1.2,0.75) to (1.2,2.3) node[below left] {$\mathsmaller{k}$};
\draw (1.2,0.75) .. controls +(0,-1.2) and +(0,-1.2) .. (2.4,0.75);
\draw[black,fill=white] (2.4,0.75) circle (0.285cm) node {$\hat{\beta}$};
\draw[-<-=.5] (3.6,-0.2) node[above left] {$\mathsmaller{j}$} to (3.6,0.75);
\draw (2.8,1.5) .. controls +(0,0.65) and +(0,1.15) .. (3.6,0.75);
\draw[-<-=.5] (4.4,-0.2) node[above left] {$\mathsmaller{i}$} to (4.4,0.75);
\draw (2,1.5) .. controls +(0,1.5) and +(0,1.5) .. (4.4,0.75);
\end{tikzpicture}
\end{gather*}
and consequently, using the snake identity,
\begin{gather}
\begin{tikzpicture}[baseline=(current bounding box.center)]
\draw[-<-=.1] (2,0) node[above left] {$\mathsmaller{j}$} .. controls +(0,1.2) and +(0,0.5) .. (2.8,0.5);
\draw[-<-=.7] (2.4,0.75) to (2.4,1.5) node[below right] {$\mathsmaller{k}$}; 
\draw[black,fill=white] (2.4,0.75) circle (0.285cm) node {$\hat{\beta}^\ast$};
\draw[->-=.7] (3.6,1) to (3.6,1.5) node[below right] {$\mathsmaller{i}$};
\draw (2.8,0.5) .. controls +(0,-1) and +(0,-1.2) .. (3.6,1);
\end{tikzpicture} \; \, \, = \; \, \, \begin{tikzpicture}[baseline=(current bounding box.center)]
\draw[->-=.9,-<-=.1] (2,1.5) node[below left] {$\mathsmaller{i}$} .. controls +(0,-1) and +(0,-1) .. (2.8,1.5) node[below right] {$\mathsmaller{j}$};
\draw[-<-=.6] (1.2,0.75) to (1.2,2.3) node[below left] {$\mathsmaller{k}$};
\draw (1.2,0.75) .. controls +(0,-1.2) and +(0,-1.2) .. (2.4,0.75);
\draw[black,fill=white] (2.4,0.75) circle (0.285cm) node {$\hat{\beta}$};
\draw[-<-=.5] (3.6,-0.2) node[above left] {$\mathsmaller{j}$} to (3.6,0.75);
\draw (2.8,1.5) .. controls +(0,0.65) and +(0,1.15) .. (3.6,0.75);
\draw[-<-=.5] (4.4,0.5) node[above left] {$\mathsmaller{i}$} to (4.4,0.75);
\draw (2,1.5) .. controls +(0,1.5) and +(0,1.5) .. (4.4,0.75);
\draw[->-=.7] (5.2,1.5) to (5.2,2.3) node[below left] {$\mathsmaller{i}$};
\draw (4.4,0.5) .. controls +(0,-1) and +(0,-1.85) .. (5.2,1.5);
\end{tikzpicture} \; \; \, \, = \; \, \, \begin{tikzpicture}[baseline=(current bounding box.center)]
\draw[->-=.9] (2,1.5) .. controls +(0,-1) and +(0,-1) .. (2.8,1.5) node[below right] {$\mathsmaller{j}$};
\draw[-<-=.6] (1.2,0.75) to (1.2,2.3) node[below left] {$\mathsmaller{k}$};
\draw (1.2,0.75) .. controls +(0,-1.2) and +(0,-1.2) .. (2.4,0.75);
\draw[black,fill=white] (2.4,0.75) circle (0.285cm) node {$\hat{\beta}$};
\draw[-<-=.5] (3.6,-0.2) node[above left] {$\mathsmaller{j}$} to (3.6,0.75);
\draw (2.8,1.5) .. controls +(0,0.75) and +(0,1.75) .. (3.6,0.75);
\draw[->-=.8] (2,1.5) to (2,2.3) node[below left] {$\mathsmaller{i}$};
\end{tikzpicture} \; \; \, \, = \; \, \, \begin{tikzpicture}[baseline=(current bounding box.center)]
\draw[->-=.9,->-=.2] (2,0) node[below left] {$\mathsmaller{k}$} .. controls +(0,-1) and +(0,-1) .. (2.8,0) node[below right] {$\mathsmaller{i}$};
\draw[->-=.7] (2.4,-0.75) to (2.4,-1.5) node[above right] {$\mathsmaller{j}$}; 
\draw[black,fill=white] (2.4,-0.75) circle (0.285cm) node {$\check{\beta}$};
\end{tikzpicture} \; \, .
\label{eq:appendix-check-hatstar}
\end{gather} \medskip

To summarize, we have proved that $\kappa$-duality is compatible with the canonical isomorphisms of Hom spaces in Equation \eqref{eq:hom-isomorphisms}: 
\begin{lemma}
\label{lem:duals-are-equivalent}
Given a basis vector $\alpha \in \Hom(i\otimes j,k)$, the basis vectors
\begin{gather*}
\hat{\alpha} \equiv \, \begin{tikzpicture}[baseline=(current bounding box.center)]
\draw[->-=.9,-<-=.1] (2,0) node[below left] (i2) {$\mathsmaller{i}$} .. controls +(0,-1) and +(0,-1) .. (2.8,0) node[below right] (j2) {$\mathsmaller{j}$};
\draw[-<-=.5] (2.4,-0.75) to (2.4,-1.5) node[above right] (k2) {$\mathsmaller{k}$}; 
\draw[black,fill=white] (2.4,-0.75) circle (0.285cm) node {$\hat{\alpha}$};
\end{tikzpicture} \; \in \Hom(k,i \otimes j) \\
\hat{\alpha}^\ast \equiv \, \begin{tikzpicture}[baseline=(current bounding box.center)]
\draw[->-=.9,-<-=.1] (2,0) node[above left] {$\mathsmaller{j}$} .. controls +(0,1) and +(0,1) .. (2.8,0) node[above right] {$\mathsmaller{i}$};
\draw[-<-=.7] (2.4,0.75) to (2.4,1.5) node[below right] {$\mathsmaller{k}$}; 
\draw[black,fill=white] (2.4,0.75) circle (0.285cm) node {$\hat{\alpha}^\ast$};
\end{tikzpicture} \, = \,
\begin{tikzpicture}[baseline=(current bounding box.center)]
\draw[->-=.9,-<-=.1] (2,1.5) node[below left] {$\mathsmaller{i}$} .. controls +(0,-1) and +(0,-1) .. (2.8,1.5) node[below right] {$\mathsmaller{j}$};
\draw[-<-=.6] (3.6,0.75) to (3.6,2.3) node[below left] {$\mathsmaller{k}$};
\draw (2.4,0.75) .. controls +(0,-1.2) and +(0,-1.2) .. (3.6,0.75);
\draw[black,fill=white] (2.4,0.75) circle (0.285cm) node {$\hat{\alpha}$};
\draw[-<-=.5] (0.4,-0.2) node[above left] {$\mathsmaller{j}$} to (0.4,0.75);
\draw (0.4,0.75) .. controls +(0,1.5) and +(0,1.5) .. (2.8,1.5);
\draw[-<-=.5] (1.2,-0.2) node[above left] {$\mathsmaller{i}$} to (1.2,0.75);
\draw (1.2,0.75) .. controls +(0,1.15) and +(0,0.65) .. (2,1.5);
\end{tikzpicture} \; \in \Hom(j^\ast \otimes i^\ast, k^\ast) \\
\check{\alpha} \equiv \, \begin{tikzpicture}[baseline=(current bounding box.center)]
\draw[->-=.9,->-=.2] (2,0) node[below left] {$\mathsmaller{k}$} .. controls +(0,-1) and +(0,-1) .. (2.8,0) node[below right] {$\mathsmaller{i}$};
\draw[->-=.7] (2.4,-0.75) to (2.4,-1.5) node[above right] {$\mathsmaller{j}$}; 
\draw[black,fill=white] (2.4,-0.75) circle (0.285cm) node {$\check{\alpha}$};
\end{tikzpicture} \, = \,
\begin{tikzpicture}[baseline=(current bounding box.center)]
\draw[->-=.9] (2,1.5) .. controls +(0,-1) and +(0,-1) .. (2.8,1.5) node[below right] {$\mathsmaller{j}$};
\draw[->-=.3] (2,1.5) to (2,2.2) node[below left] {$\mathsmaller{i}$};
\draw[-<-=.6] (1.2,0.75) to (1.2,2.2) node[below left] {$\mathsmaller{k}$};
\draw (1.2,0.75) .. controls +(0,-1.2) and +(0,-1.2) .. (2.4,0.75);
\draw[black,fill=white] (2.4,0.75) circle (0.285cm) node {$\hat{\alpha}$};
\draw[-<-=.5] (3.6,-0.2) node[above left] {$\mathsmaller{j}$} to (3.6,0.75);
\draw (2.8,1.5) .. controls +(0,0.65) and +(0,1.85) .. (3.6,0.75);
\end{tikzpicture} \; \in \Hom(j^\ast, k^\ast \otimes i)
\end{gather*}
are dual to $\alpha$ with respect to the pairings $\kappa$, $\eta$ and $\vartheta$, respectively, and these pairings are the same up to the canonical isomorphisms of Equation \eqref{eq:hom-isomorphisms}.
\end{lemma}
For clarity we distinguish the different incarnations of dual basis vectors in the main text, but use the fact that they can be translated into each other using isomorphisms coming from the rigid, pivotal structure of the category.


\section{Module category theory}
\label{appendix:module-categories}

A module category over a monoidal category is a ``categorification'' of the algebraic concept of a module over a ring. For ease of reference and for the reader's convenience we collect some definitions and fundamental theorems which are used in the main text. We refer to \cite[Chapter 7]{egno} for details. By $\mathds{k}$ we denote an algebraically closed field of zero characteristic. All categories are assumed to be finite semisimple abelian $\mathds{k}$-linear categories. 

\subsection{Module and bimodule categories}

Recall that a monoidal category is a tuple $(\mathcal{C},\otimes,\mathds{1},\alpha,\rho,\lambda)$ consisting of a category $\mathcal{C}$, a tensor product bifunctor $\otimes\colon \mathcal{C}\times\mathcal{C}\to\mathcal{C}$, a unit object $\mathds{1} \in\mathcal{C}$, and natural isomorphisms $\alpha\colon ((-\otimes-)\otimes-) \to (-\otimes(-\otimes-))$, $\rho\colon (-)\otimes\mathds{1} \to (-)$, $\lambda\colon \mathds{1} \otimes(-) \to (-)$, called associativity, right and left unit constraint, respectively, satisfying certain coherence axioms \cite[Definition 2.2.8]{egno}.

\begin{definition}[Module category]
A left \textit{module category} $\mathcal{M} \equiv (\mathcal{M},\triangleright,a,l)$ over a monoidal category $\mathcal{C} \equiv (\mathcal{C},\otimes,\mathds{1},\alpha,\rho,\lambda)$ is a category $\mathcal{M}$ with a bifunctor $\triangleright \colon \mathcal{C} \times \mathcal{M} \to \mathcal{M}$ (the action functor) and natural isomorphisms $a \colon (-\otimes-)\triangleright(-) \to (-)\triangleright(-\triangleright-)$ (the module associativity isomorphism) and $l \colon \mathds{1} \triangleright (-) \to (-)$ (the module unit isomorphism) satisfying the commuting diagrams
\[
\begin{tikzcd}
{} & ((X \otimes Y) \otimes Z) \triangleright M \arrow[swap]{dl}{\alpha_{X,Y,Z}\triangleright \id_M} \arrow{dr}{a_{X\otimes Y,Z,M}} 
\\ 
(X \otimes (Y \otimes Z)) \triangleright M \dar[swap]{a_{X,Y\otimes Z,M}} && (X \otimes Y) \triangleright (Z \triangleright M) \dar{a_{X,Y,Z\triangleright M}} \\
X \triangleright ((Y \otimes Z) \triangleright M) \arrow{rr}{\id_X \triangleright a_{Y,Z,M}} && X \triangleright (Y \triangleright (Z \triangleright M))
\end{tikzcd}
\]
and
\[
\begin{tikzcd}
(X \otimes \mathds{1}) \triangleright M \arrow[swap]{dr}{\rho_X \triangleright \id_M} \arrow{rr}{a_{X,\mathds{1},M}} && X \triangleright (\mathds{1} \triangleright M) \arrow{dl}{\id_X \triangleright l_M}  \\
{} & X \triangleright M & {}
\end{tikzcd}
\]
for all $X,Y,Z \in \mathcal{C}$ and $M \in \mathcal{M}$. The isomorphisms $a$ and $l$ are called the ``module constraints'' of $\mathcal{M}$.
\end{definition}

\begin{definition}[Module functor and natural transformation]
\label{def:module-functor}
A left \textit{module functor} $F\colon {\mathcal{M}} \to {\mathcal{N}}$ between left $\mathcal{C}$-module categories $\mathcal{M} \equiv (\mathcal{M},\triangleright^{\mathcal{M}},a^{\mathcal{M}},l^{\mathcal{M}})$ and $\mathcal{N} \equiv (\mathcal{N},\triangleright^{\mathcal{N}},a^{\mathcal{N}},l^{\mathcal{N}})$ is a functor with a natural isomorphism $f\colon F(-\triangleright^{\mathcal{M}} -)\to (-)\triangleright^{\mathcal{N}} F(-)$ such that the diagrams
\[
\begin{tikzcd}
{} & F((X\otimes Y) \triangleright^{\mathcal{M}} M) \arrow[swap]{dl}{F(a_{X,Y,M}^{\mathcal{M}})} \arrow{dr}{f_{X\otimes Y,M}} 
\\ 
F(X \triangleright^{\mathcal{M}} (Y \triangleright^{\mathcal{M}} M)) \dar[swap]{f_{X,Y\triangleright^{\mathcal{M}} M}} && (X \otimes Y) \triangleright^{\mathcal{N}} F(M) \dar{a_{X,Y,F(M)}^{\mathcal{N}}} \\
X \triangleright^{\mathcal{N}} F(Y \triangleright^{\mathcal{M}} M) \arrow{rr}{\id_X \triangleright^{\mathcal{N}} f_{Y,M}} && X \triangleright^{\mathcal{N}} (Y \triangleright^{\mathcal{N}} F(M))
\end{tikzcd}
\]
and
\[
\begin{tikzcd}
F(\mathds{1}_{\mathcal{C}} \triangleright^{\mathcal{M}} M) \arrow[swap]{dr}{F(l_M^{\mathcal{M}})} \arrow{rr}{f_{\mathds{1},M}} && \mathds{1} \triangleright^{\mathcal{N}} F(M) \arrow{dl}{l_{F(M)}^{\mathcal{N}}}  \\
{} & F(M) & {}
\end{tikzcd}
\]
commute for all $X,Y \in \mathcal{C}$ and $M \in \mathcal{M}$.

A left \textit{module natural transformation} $\eta$ between module functors $(F,f), (G,g)$ is a natural transformation  which is compatible with the module functor data, i.e. the square
\[
\begin{tikzcd}
F(X \triangleright^{\mathcal{M}} M) \dar[swap]{\eta_{X \triangleright^{\mathcal{M}} M}} \arrow{rr}{f_{X,M}} && X \triangleright^{\mathcal{N}} F(M) \dar{\id_X \triangleright^{\mathcal{N}} \eta_M}  \\
G(X \triangleright^{\mathcal{M}} M) \arrow{rr}{g_{X,M}} && X \triangleright^{\mathcal{N}} G(M) 
\end{tikzcd}
\]
commutes for all $X \in \mathcal{C}$ and $M \in \mathcal{M}$.
\end{definition}

Right $\mathcal{C}$-module categories (with right action $\triangleleft$ and right module unit isomorphism $r$), their functors and natural transformations are defined analogously. Every monoidal category is a left and right module category over itself, with module action given by the tensor product and module constraints coming from the monoidal constraints.

\begin{definition}[Equivalence of module categories]
\label{def:appendix-equivalence-modulecats}
Two module categories are called \textit{equivalent}, if there exists a module functor between them inducing an equivalence of the underlying categories.
\end{definition}

\begin{definition}[Bimodule category]
\label{def:bimodule-cat}
A \textit{$\mathcal{C}$-$\mathcal{D}$-bimodule category} $\mathcal{M}$ is a left $\mathcal{C}\boxtimes\mathcal{D}^{\rev}$-module category with action $(X \boxtimes Y)\triangleright M \defeq  (X \triangleright M)\triangleleft Y$. A $\mathcal{C}$-$\mathcal{C}$-bimodule category will also be called a $\mathcal{C}$-bimodule category for short.
\end{definition}
If $\mathcal{M}$ is a left $\mathcal{C}$-module and $\mathcal{N}$ is a right $\mathcal{D}$-module category, then the Cartesian product is a $\mathcal{C}$-$\mathcal{D}$-bimodule category. Furthermore, every monoidal category $\mathcal{C}$ is a $\mathcal{C}$-bimodule category in the obvious way.

We remark that a $\mathcal{C}$-$\mathcal{D}$-bimodule category $\mathcal{M}$ can equivalently be described as a left $\mathcal{C}$-module and right $\mathcal{D}$-module category, such that the corresponding actions $\triangleright$ and $\triangleleft$ commute up to a natural isomorphism $\gamma\colon (- \triangleright -)\triangleleft (-) \to (-) \triangleright (- \triangleleft -)$ subject to certain coherence axioms, which describe the compatibility of $\triangleright$ and $\triangleleft$ with $\gamma$ (see \cite[Definition 7.1.7]{egno} and \cite[Proposition 1.3.10]{greenough-thesis} for the precise definition and the proof that both definitions are equivalent).

Bimodule functors and bimodule natural transformations are defined accordingly. Equivalently, a $\mathcal{C}$-$\mathcal{D}$-bimodule functor between $\mathcal{C}$-$\mathcal{D}$-bimodule categories is a left and right module functor such that a certain coherence hexagon commutes \cite[Remark 2.14]{greenough}. $\mathcal{C}$-$\mathcal{D}$-bimodule categories, their bimodule functors and bimodule natural transformations form a 2-category denoted $\Bimod(\mathcal{C},\mathcal{D})$ with composition of functors and composition of natural transformations as horizontal and vertical composition, respectively \cite[Lemma 2.3.11]{schaumann}.

\begin{lemma}[{\cite[Corollary 2.12]{dsps}}]
\label{lem:bimodule-eq}
Let $\mathcal{M}$ and $\mathcal{N}$ be bimodule categories and $F\colon \mathcal{M} \to \mathcal{N}$ a bimodule functor. Then $F$ is an equivalence of bimodule categories, if $F$ induces an equivalence of the underlying categories.
\end{lemma}

\begin{proposition}[{\cite[Proposition 7.1]{greenough}}]
\label{prop:from-module-to-bimodule-category}
Let $(\mathcal{M},\triangleright,a,l)$ be a left module category over a braided monoidal category $\mathcal{C}$ with braiding $c$. Then $\mathcal{M}$ is an $\mathcal{C}$-bimodule category: it is a right $\mathcal{C}$-module category $(\mathcal{M},\triangleleft,a',r)$ via
\begin{gather*}
M \triangleleft X \defeq  X \triangleright M \, \, ,\\
a'_{M,X,Y} \defeq a_{Y,X,M} \circ (\id_M \triangleleft c_{X,Y})\colon M \triangleleft (X \otimes Y) \to (M \triangleleft X) \triangleleft Y \, \, ,\\
r_M \defeq l_M \colon M \triangleleft \mathds{1} \equiv \mathds{1} \triangleright M \to M \, \, .
\end{gather*}
This means the $\mathcal{C}\boxtimes\mathcal{C}^{\rev}$-module action $\widetilde{\triangleright}$ is given by $(X\boxtimes Y)\widetilde{\triangleright} M \defeq  (X \triangleright M) \triangleleft Y \equiv Y \triangleright (X \triangleright M)$. The mixed associativity isomorphism is
\begin{gather*}
\gamma_{X,M,Y} \defeq a_{X,Y,M} \circ (c_{Y,X} \triangleright \id_M) \circ a_{Y,X,M}^{-1}\colon\\(X \triangleright M)\triangleleft Y \equiv Y \triangleright (X \triangleright M) \to X \triangleright (M \triangleleft Y) \equiv X \triangleright (Y \triangleright M) \, \, .
\end{gather*}
\end{proposition}
We remark that an analogous result is obtained when the inverse braiding $c^{-1}$ is used instead of $c$ in the definition of $a'$ and $\gamma$. 

The previous proposition immediately allows to extend all module functors to bimodule functors:
\begin{proposition}
\label{prop:from-module-functor-to-bimodule-functor}
Let $(\mathcal{M},\triangleright^{\mathcal{M}},a^{\mathcal{M}},l^{\mathcal{M}})$ and $(\mathcal{N},\triangleright^{\mathcal{N}},a^{\mathcal{N}},l^{\mathcal{N}})$ be left module categories over a braided monoidal category $\mathcal{C}$ with braiding $c$, and let $(F,f)\colon \mathcal{M} \to \mathcal{N}$ be a left module functor. Then $F$ can be extended to a bimodule functor between the bimodule categories $(\mathcal{M},\triangleright^{\mathcal{M}},a^{\mathcal{M}},l^{\mathcal{M}},a'^{\mathcal{M}})$ and $(\mathcal{N},\triangleright^{\mathcal{N}},a^{\mathcal{N}},l^{\mathcal{N}},a'^{\mathcal{N}})$ obtained from Proposition \ref{prop:from-module-to-bimodule-category}.
\end{proposition}
\begin{proof}
One can directly show that $(F,f^r)$ is a right module functor between the right module categories $(\mathcal{M},\triangleleft^{\mathcal{M}},a'^{\mathcal{M}},r^{\mathcal{M}})$ and $(\mathcal{N},\triangleleft^{\mathcal{N}},a'^{\mathcal{N}},r^{\mathcal{N}})$ (using the notation of Proposition \ref{prop:from-module-to-bimodule-category}), where $f^r_{M,X} \defeq f_{X,M} \colon F(M \triangleleft^{\mathcal{M}} X) = F(X \triangleright^{\mathcal{M}} M) \to F(M) \triangleleft^{\mathcal{N}} X = X \triangleright^{\mathcal{N}} F(M)$ is the right module functor constraint; the corresponding module axioms follow from the definition of the right module constraints $a'^{\mathcal{M}}$, $r^{\mathcal{M}}$, $a'^{\mathcal{N}}$, $r^{\mathcal{N}}$ in Proposition \ref{prop:from-module-to-bimodule-category} and the pentagon axiom of the left module functor $(F,f)$. We show that the coherence hexagon \cite[Remark 2.14]{greenough}, which describes the compatibility of left and right action, is satisfied. That is, the equation
\begin{gather}
\gamma^{\mathcal{N}}_{X,F(M),Y} \circ (f_{X,M} \triangleleft \id_Y) \circ f^r_{X\triangleright M,Y} = (\id_X \triangleright f^r_{M,Y}) \circ f_{X,M\triangleleft Y} \circ F(\gamma^{\mathcal{M}}_{X,M,Y})
\label{eq:coherence-hex}
\end{gather}
has to hold, where $\gamma^{\mathcal{M}}$ and $\gamma^{\mathcal{N}}$ are the mixed associativity isomorphisms as in Proposition \ref{prop:from-module-to-bimodule-category}. The right hand side can be rewritten in the following way:
\begin{gather}
(\id_X \triangleright f^r_{M,Y}) \circ f_{X,M\triangleleft Y} \circ F(\gamma^{\mathcal{M}}_{X,M,Y}) = \nonumber \\ = (\id_X \triangleright f_{Y,M}) \circ f_{X,Y\triangleright M} \circ F(a^{\mathcal{M}}_{X,Y,M}) \circ F(c_{Y,X} \triangleright \id_M) \circ F(a^{\mathcal{M}}_{Y,X,M})^{-1} = \nonumber \\ = a^{\mathcal{N}}_{X,Y,F(M)} \circ f_{X \otimes Y,M} \circ F(c_{Y,X} \triangleright \id_M) \circ F(a^{\mathcal{M}}_{Y,X,M})^{-1}
\label{eq:coherence-hex-left}
\end{gather}
Here we used the explicit form of $\gamma^{\mathcal{M}}_{X,M,Y}$ and the pentagon axiom of the left module functor $(F,f)$. The same pentagon axiom is also inserted into the left hand side of Equation \eqref{eq:coherence-hex}, together with the explicit form of $\gamma^{\mathcal{N}}_{X,F(M),Y}$:
\begin{gather*}
\gamma^{\mathcal{N}}_{X,F(M),Y} \circ (f_{X,M} \triangleleft \id_Y) \circ f^r_{X\triangleright M,Y} = \\ = \gamma^{\mathcal{N}}_{X,F(M),Y} \circ (\id_Y \triangleright f_{X,M}) \circ f_{Y,X\triangleright M} = \gamma^{\mathcal{N}}_{X,F(M),Y} \circ a^{\mathcal{N}}_{Y,X,F(M)} \circ f_{Y \otimes X,M} \circ F(a^{\mathcal{M}}_{Y,X,M})^{-1} = \\ = a^{\mathcal{N}}_{X,Y,F(M)} \circ (c_{Y,X} \triangleright \id_{F(M)}) \circ (a^{\mathcal{N}}_{Y,X,F(M)})^{-1} \circ a^{\mathcal{N}}_{Y,X,F(M)} \circ f_{Y \otimes X,M} \circ F(a^{\mathcal{M}}_{Y,X,M})^{-1} = \\ = a^{\mathcal{N}}_{X,Y,F(M)} \circ (c_{Y,X} \triangleright \id_{F(M)}) \circ f_{Y \otimes X,M} \circ F(a^{\mathcal{M}}_{Y,X,M})^{-1}
\end{gather*}
This equals the expression \eqref{eq:coherence-hex-left} if the equation
\begin{gather*}
(c_{Y,X} \triangleright \id_{F(M)}) \circ f_{Y \otimes X,M} = f_{X \otimes Y,M} \circ F(c_{Y,X} \triangleright \id_M)
\end{gather*}
is satisfied, which is indeed the case by functoriality of the braiding $c$. Hence $(F,f,f^r)$ is a $\mathcal{C}$-bimodule functor from $(\mathcal{M},\triangleright^{\mathcal{M}},a^{\mathcal{M}},l^{\mathcal{M}},a'^{\mathcal{M}})$ to $(\mathcal{N},\triangleright^{\mathcal{N}},a^{\mathcal{N}},l^{\mathcal{N}},a'^{\mathcal{N}})$.
\end{proof}

\subsection{Categories of modules}
\label{appendix:cats-of-modules-algebras}

Let $(\mathcal{C},\otimes,\mathds{1},\alpha,\rho,\lambda)$ be a monoidal category. We want to study algebras and modules internal to $\mathcal{C}$, thus generalizing the classical notions from basic algebra (which are internal to $\Vectk$, the category of $\mathds{k}$-vector spaces and linear maps).

\begin{definition}[Algebra object]
\label{def:appendix-algebra-object}
An \textit{algebra object} (or \textit{algebra}) $A$ in $\mathcal{C}$ is a triple $(A,m_A,u_A)$ consisting of an object $A \in \mathcal{C}$, a morphism $m_A \colon A \otimes A \to A$ called ``multiplication'', and a morphism $u_A \colon \mathds{1} \to A$ called ``unit'', such that the diagrams
\[
\begin{tikzcd}
(A \otimes A) \otimes A \dar[swap]{m_A \otimes \id_A} \arrow{rr}{\alpha_{A,A,A}} && A \otimes (A\otimes A) \dar{id_A \otimes m_A} \\
A \otimes A \arrow[swap]{dr}{m_A} && A \otimes A \arrow{dl}{m_A} \\
{} & A & {}
\end{tikzcd}
\]
and
\[
\begin{tikzcd}
\mathds{1} \otimes A \arrow[swap]{dr}{\lambda_A} \arrow{rr}{u_A \otimes \id_A} && A \otimes A \arrow{dl}{m_A}  \\
{} & A & {}
\end{tikzcd}\;\;\;
\begin{tikzcd}
A \otimes \mathds{1} \arrow[swap]{dr}{\rho_A} \arrow{rr}{\id_A \otimes u_A} && A \otimes A \arrow{dl}{m_A}  \\
{} & A & {}
\end{tikzcd}
\]
commute.
\end{definition}
For $\mathcal{C}=\Vectk$ this definition coincides with the definition of an associative $\mathds{k}$-algebra with unit. 

\begin{definition}[Algebra homomorphism]
\label{def:algebra-homomorphism}
An \textit{algebra homomorphism} between two algebras $(A,m_A,u_A),(B,m_B,u_B)$ in $\mathcal{C}$ is a morphism $f\colon A\to B$ in $\mathcal{C}$ which is compatible with multiplications and units in the sense that $f \circ u_A = u_B$ and $f \circ m_A = m_B \circ (f \otimes f)$.
\end{definition}
This leads to the notion of an algebra isomorphism in the obvious way. If $f\colon A\to B$ is an isomorphism (in the category) and an algebra homomorphism, then it is also an algebra isomorphism.

\begin{proposition}[Opposite algebra]
\label{def:op-alg}
Let $A\equiv(A,m_A,u_A)$ be an algebra in a braided monoidal category $\mathcal{C}$ with braiding $c$. Then $(A,m_A \circ c_{A,A},u_A)$ is an algebra, called the \emph{opposite algebra} $A^{\op}$.
\end{proposition}
It follows directly from the definitions that the opposite algebra is indeed an algebra.
\begin{proposition}[Tensor product algebra]
\label{prop:alg-structure-on-tensor-product}
Let $(A,m_A,u_A)$ and $(B,m_B,u_B)$ be algebras in a braided strict monoidal category $\mathcal{C}$ with braiding $c$. Then the \emph{tensor product algebra} $A \otimes B$ is $(A\otimes B, m_{A\otimes B}, u_{A\otimes B})$ with $m_{A\otimes B} = (m_A \otimes m_B) \circ (\id_A \otimes c_{B,A} \otimes \id_B)$ and $u_{A\otimes B} = u_A \otimes u_B$. The tensor product of algebras is associative.
\end{proposition}
This is a direct consequence of naturality of the braiding and associativity of the strict monoidal structure. Note that in the definition of the opposite and the tensor product algebra one could also use the inverse braiding $c^{-1}$ instead of $c$, resulting in algebras $(A,m_A \circ c_{A,A}^{-1},u_A)$ and $(A\otimes B, \tilde{m}_{A\otimes B}, u_A \otimes u_B)$, where $\tilde{m}_{A\otimes B}$ is $m_{A\otimes B}$ with $c_{B,A}$ replaced by $c_{A,B}^{-1}$. However, if $\mathcal{C}$ is ribbon, these algebras are isomorphic to $(A,m_A \circ c_{A,A},u_A)$ and $(A\otimes B, m_{A\otimes B}, u_A \otimes u_B)$, respectively \cite[Remark 3.23 (i)]{fuchs-runkel-schweigert-tft}. In this paper we only work with the definitions of $A^{\op}$ and $A\otimes B$ from Proposition \ref{def:op-alg} and Proposition \ref{prop:alg-structure-on-tensor-product}.

\begin{definition}[Module]
A \textit{right module} over an algebra $(A,m_A,u_A)$, or right $A$-module for short, is a pair $(M,p)$, where $M \in \mathcal{C}$ is an object and $p\colon M \otimes A \to M$ is a morphism (the module action) such that the diagrams
\[
\begin{tikzcd}
(M \otimes A) \otimes A \dar[swap]{p \otimes \id_A} \arrow{rr}{\alpha_{M,A,A}} && M \otimes (A\otimes A) \dar{id_M \otimes m_A} \\
M \otimes A \arrow[swap]{dr}{p} && M \otimes A \arrow{dl}{p} \\
{} & M & {}
\end{tikzcd}\;\;\;
\begin{tikzcd}
M \otimes \mathds{1} \arrow[swap]{dr}{\rho_M} \arrow{rr}{\id_M \otimes u_A} && M \otimes A \arrow{dl}{p}  \\
{} & M & {}
\end{tikzcd}
\]
commute.
\end{definition}

\begin{definition}[Module homomorphism]
A \textit{right module homomorphism} between two right modules $(M,p)$, $(N,q)$ over an algebra $(A,m_A,u_A)$ in $\mathcal{C}$ is a morphism $\varphi\colon M\to N$ in $\mathcal{C}$ which is compatible with the module actions in the sense that $q \circ (\varphi \otimes \id_A) = \varphi \circ p$.
\end{definition}
Left $A$-modules and left $A$-module homomorphisms are defined analogously.

Let $M_1,M_2$ be $A$-modules in $\mathcal{C}$. Then module homomorphisms between them form a subspace of $\Hom (M_1,M_2)$ that is closed under composition. Consequently:

\begin{proposition}[Categories of modules]
Right and left $A$-modules in $\mathcal{C}$ form categories $\Mod[r]{A}(\mathcal{C})$ and $\Mod[l]{A}(\mathcal{C})$, respectively.
\end{proposition}

\begin{lemma}
The category $\Mod[r]{A}(\mathcal{C})$ can be equipped with the structure of a left $\mathcal{C}$-module category: For any algebra $A \in \mathcal{C}$, right $A$-module $(M,p)$ and $X \in \mathcal{C}$, the object $X \otimes M$ has a right $A$-module structure via
\begin{gather*}
(X\otimes M) \otimes A \xrightarrow[]{\alpha_{X,M,A}} X \otimes (M\otimes A) \xrightarrow[]{id_X \otimes p} X \otimes M
\end{gather*}
which provides an action functor $\mathcal{C} \times \Mod[r]{A}(\mathcal{C}) \to \Mod[r]{A}(\mathcal{C})$ given on objects by $(X,M) \mapsto X\otimes M$. The module associativity isomorphisms $a_{X,Y,M}\defeq \alpha_{X,Y,M} \colon (X\otimes Y)\otimes M \to X \otimes (Y \otimes M)$ and unit isomorphisms $l_M\defeq \lambda_M \colon \mathds{1} \otimes M \to M$ come from the monoidal structure of $\mathcal{C}$, are isomorphisms of $A$-modules and define a left $\mathcal{C}$-module structure on $\Mod[r]{A}(\mathcal{C})$.
\end{lemma}
This is proven by direct verification. A corresponding result also holds for $\Mod[l]{A}(\mathcal{C})$, which is a right $\mathcal{C}$-module category.

\begin{definition}[Morita equivalence]
\label{def:morita-eq}
Two algebras $A,B \in \mathcal{C}$ are \textit{Morita equivalent} if $\Mod[r]{A}(\mathcal{C})$ and $\Mod[r]{B}(\mathcal{C})$ (or $\Mod[l]{A}(\mathcal{C})$ and $\Mod[l]{B}(\mathcal{C})$) are equivalent $\mathcal{C}$-module categories.
\end{definition}

\begin{definition}[Bimodule]
Let $A,B\in\mathcal{C}$ be algebras. An $A$-$B$-\textit{bimodule} in $\mathcal{C}$ is a triple $(M,p,q)$ with $M \in \mathcal{C}$, $p\colon A \otimes M \to M$, $q\colon M \otimes B \to M$, such that $(M,p)$ is a left $A$-module in $\mathcal{C}$, $(M,q)$ is a right $B$-module in $\mathcal{C}$, and 
\[
\begin{tikzcd}
(A \otimes M) \otimes B \dar[swap]{p \otimes \id_B} \arrow{rr}{\alpha_{A,M,B}} && A \otimes (M\otimes B) \dar{id_A \otimes q} \\
M \otimes B \arrow[swap]{dr}{q} && A \otimes M \arrow{dl}{p} \\
{} & M & {}
\end{tikzcd}
\]
commutes. A homomorphism of $A$-$B$-bimodules is a morphism in $\mathcal{C}$ which is both a homomorphism of left $A$-modules and right $B$-modules. $A$-$B$-bimodules and their homomorphisms form a category called $\AlgBiMod{A}{B}(\mathcal{C})$. The category of $A$-$A$-bimodules will be denoted $A\mh\Bimod(\mathcal{C})$ for short.
\end{definition}

\begin{theorem}[{\cite[Theorem 7.10.1]{egno}}]
\label{thm:egno-mod-alg}
Let $\mathcal{M}$ be a finite (as linear category) left respectively right module category over a finite tensor category $\mathcal{C}$ with right exact action in $\mathcal{C}$. Then there is an algebra object $A \in \mathcal{C}$ and an equivalence $\mathcal{M} \simeq \Mod[r]{A}(\mathcal{C})$ ($\simeq \Mod[l]{A}(\mathcal{C})$) as left respectively right $\mathcal{C}$-module categories.
\end{theorem}
If $\mathcal{C}$ is braided then there is an equivalence $\AlgBiMod{A}{B}(\mathcal{C}) \simeq \Mod[l]{(A\otimes B^{\op})}(\mathcal{C})$ of module categories.

\subsection{Relative Deligne product}
\label{sec:appendix-relative-deligne}

We assume $\mathcal{C}$ to be a fusion category and all module categories to be semisimple.

\begin{definition}[Balanced functor]
Let $\mathcal{M}$ be a right and $\mathcal{N}$ a left $\mathcal{C}$-module category. A functor $F\colon\mathcal{M}\boxtimes\mathcal{N}\to\mathcal{A}$ to some abelian category $\mathcal{A}$ is called \textit{$\mathcal{C}$-balanced} if there exists a natural isomorphism $\beta\colon F((-\triangleleft -) \boxtimes -) \to F(- \boxtimes (- \triangleright -))$ satisfying the pentagon coherence diagram
\[
\begin{tikzcd}
F((M \triangleleft (X\otimes Y))\boxtimes N) \dar[swap]{} \arrow{rr}{\beta_{M,X\otimes Y,N}} && F(M \boxtimes ((X\otimes Y)\triangleright N)) \dar{} \\
F(((M \triangleleft X)\triangleleft Y)\boxtimes N) \arrow[swap]{dr}{\beta_{M\triangleleft X,Y,N}} && F(M\boxtimes (X\triangleright(Y\triangleright N))) \arrow{dl}{\beta_{M,X,Y \triangleright N}^{-1}} \\
{} & F((M \triangleleft X)\boxtimes(Y \triangleright N)) & {}
\end{tikzcd}
\]
for all $X,Y\in \mathcal{C}$, $M \in \mathcal{M}$, $N \in \mathcal{N}$. Module associativity isomorphisms act along the unmarked arrows.
\end{definition}
This means for all $C \in \mathcal{C}$, $M \in \mathcal{M}$ and $N \in \mathcal{N}$, one has a natural isomorphism $\beta_{M,C,N} \colon F((M \triangleleft C)\boxtimes N) \to F(M \boxtimes (C \triangleright N))$, which can also be expressed as $\beta_{M,C,N} \colon F((M \boxtimes N)\triangleleft C) \to F(C \triangleright (M \boxtimes N))$.

\begin{definition}[Relative Deligne product]
\label{def:rel-Deligne-product}
The \textit{relative Deligne product} (also called ``balanced tensor product'') of a right $\mathcal{C}$-module category $\mathcal{M}$ with a left $\mathcal{C}$-module category $\mathcal{N}$ consists of an abelian category $\mathcal{M} \boxtimes_\mathcal{C} \mathcal{N}$ and a $\mathcal{C}$-balanced functor $B_{\mathcal{M},\mathcal{N}} \colon \mathcal{M}\boxtimes\mathcal{N} \to \mathcal{M} \boxtimes_\mathcal{C} \mathcal{N}$ (called the universal balanced functor of $\mathcal{M} \boxtimes_\mathcal{C} \mathcal{N}$), which is universal for $\mathcal{C}$-balanced functors out of $\mathcal{M}\boxtimes\mathcal{N}$, i.e. $B_{\mathcal{M},\mathcal{N}}$ induces for every abelian category $\mathcal{A}$ an equivalence
\begin{gather*}
\Phi_{\mathcal{M},\mathcal{N}} \colon \Fun(\mathcal{M}\boxtimes_\mathcal{C} \mathcal{N},\mathcal{A}) \to \Fun_{\bal} (\mathcal{M}\boxtimes\mathcal{N},\mathcal{A})
\end{gather*}
from the category of functors from $\mathcal{M}\boxtimes_\mathcal{C} \mathcal{N}$ to $\mathcal{A}$ to the category of $\mathcal{C}$-balanced functors and their natural transformations (which are natural transformations compatible with the balancing isomorphisms of the respective functors in the obvious way).
\end{definition}

Setting $\mathcal{C}=\Vectk$ reduces the relative Deligne product to the ordinary Deligne product of $\mathds{k}$-linear abelian categories \cite[Definition 3.2.1]{schaumann}. Existence of the relative Deligne product and uniqueness up to unique adjoint equivalence are guaranteed by \cite[Theorem 3.3 (1)]{dsps}.

\begin{proposition}[{\cite[Proposition 3.4.1]{schaumann}}]
Let $\mathcal{M}$ be an $\mathcal{A}$-$\mathcal{C}$-bimodule category and $\mathcal{N}$ a $\mathcal{C}$-$\mathcal{B}$-bimodule category. Then $\mathcal{M}\boxtimes_\mathcal{C} \mathcal{N}$ is an $\mathcal{A}$-$\mathcal{B}$-bimodule category and $B_{\mathcal{M},\mathcal{N}}$ is a balanced $\mathcal{A}$-$\mathcal{B}$-bimodule functor.
\end{proposition}

\begin{proposition}[{\cite[Proposition 3.15]{greenough}}]
\label{prop:unitality}
The relative Deligne product is weakly unital: for a $\mathcal{C}$-$\mathcal{D}$-bimodule category $\mathcal{M}$ there are canonical equivalences $\mathcal{M}\boxtimes_\mathcal{D} \mathcal{D} \simeq \mathcal{M} \simeq \mathcal{C} \boxtimes_\mathcal{C} \mathcal{M}$ of $\mathcal{C}$-$\mathcal{D}$-bimodule categories.
\end{proposition}

\begin{proposition}[{\cite[Proposition 4.4]{greenough}}]
\label{prop:assoc-rel-deligne}
The relative Deligne product is weakly associative: for bimodule categories $\mathcal{L}$, $\mathcal{M}$ and $\mathcal{N}$ there is a bimodule equivalence $(\mathcal{L}\boxtimes_\mathcal{C} \mathcal{M}) \boxtimes_\mathcal{D} \mathcal{N} \simeq \mathcal{L}\boxtimes_\mathcal{C} (\mathcal{M} \boxtimes_\mathcal{D} \mathcal{N})$.
\end{proposition}

\begin{proposition}[{\cite[Theorem 3.3]{dsps}}]
\label{appendix-prop-relative-deligne}
Let $\mathcal{M}$ be a right and $\mathcal{N}$ a left module category over a fusion category $\mathcal{C}$. Assume the left and right action functors are right exact in $\mathcal{C}$. Let $A,B \in \mathcal{C}$ be algebra objects such that $\mathcal{M} \simeq \Mod[l]{A}(\mathcal{C})$ and $\mathcal{N} \simeq \Mod[r]{B}(\mathcal{C})$ as $\mathcal{C}$-module categories (cf. Theorem \ref{thm:egno-mod-alg}). Then $\mathcal{M} \boxtimes_\mathcal{C} \mathcal{N} \simeq \AlgBiMod{A}{B}(\mathcal{C})$ (the category of $A$-$B$-bimodule objects in $\mathcal{C}$).
\end{proposition}

The relative Deligne product of $\mathcal{C}$-bimodule categories is unital and associative up to equivalence and satisfies certain coherence axioms \cite{greenough} (with naturality up to 2-isomorphism, i.e. bimodule natural transformation). This structure can be encoded into a tricategory.

\begin{theorem}[{\cite[Theorem 3.6.1]{schaumann}}]
\label{thm:tricat-bimod}
There is a tricategory $\Bimod$: fusion categories are the objects, bimodule categories the 1-morphisms, bimodule functors the 2-morphisms, and bimodule natural transformations are the 3-morphisms, where the horizontal composition 2-functor is given by the relative Deligne product, horizontal composition in the morphism bicategories is given by the composition of functors, and the vertical composition of 3-morphisms is given by the vertical composition of natural transformations. In particular, this implies the following:
\begin{itemize}
\item[(i)] The bicategory $\Bimod(\mathcal{C},\mathcal{D})$ of $\mathcal{C}$-$\mathcal{D}$-bimodule categories and their bimodule functors and bimodule natural transformations is a 2-category (strict bicategory) with composition of functors as horizontal composition and composition of natural transformations as vertical composition.
\item[(ii)] For any three fusion categories $\mathcal{C}$, $\mathcal{D}$, $\mathcal{E}$, the relative Deligne product induces a 2-functor $\boxtimes_\mathcal{D} \colon \Bimod(\mathcal{C},\mathcal{D}) \times \Bimod(\mathcal{D},\mathcal{E}) \to \Bimod(\mathcal{C},\mathcal{E})$.
\item[(iii)] For every fusion category $\mathcal{C}$, the strict unit 2-functor $I_\mathcal{C} \colon I \to \Bimod(\mathcal{C},\mathcal{C})$ (with $I$ the unit 2-category with a single object, 1-morphism and 2-morphism) is given by $\mathcal{C}$ as $\mathcal{C}$-bimodule category.
\end{itemize}
\end{theorem}
$\,$\bigskip

\subsubsection*{Acknowledgments}

The contents of this paper are part of the author's master's thesis ``Permutation actions on modular tensor categories'' completed in May 2018 at the Faculty of Mathematics, University of Vienna, under the supervision of Nils Carqueville and Gregor Schaumann. Numerous discussions with the supervisors and Gregor Schaumann's useful input for the main theorem are gratefully acknowledged. The author also wants to thank G\"{u}nther H\"{o}rmann for his encouragement and advice.


\end{document}